%% file: MISO_RIS.tex
\documentclass[lettersize,journal]{IEEEtran}
\usepackage{graphicx}
\usepackage{amsmath}
\usepackage{xcolor}
\usepackage{amsfonts}
\usepackage{bm}
\usepackage{enumitem}
\usepackage{acronym}
\usepackage{cite}
\usepackage{subfig}
\usepackage[font=footnotesize]{caption}

\renewcommand{\baselinestretch}{0.98}

\acrodef{RIS}{reconfigurable intelligent surface}
\acrodef{OFDM}{orthogonal frequency division multiplexing}
\acrodef{LoS}{line-of-sight}
\acrodef{NLoS}{non-line-of-sight}
\acrodef{SNR}{signal-to-noise ratio}
\acrodef{SINR}{signal-to-interference-noise ratio}
\acrodef{BS}{base station}
\acrodef{UE}{user equipment}
\acrodef{FR2}{frequency range 2}
\acrodef{DL}{downlink}
\acrodef{TDoA}{time-difference-of-arrival}
\acrodef{AoA}{angle-of-arrival}
\acrodef{AoD}{angle-of-departure}
\acrodef{ML}{maximum likelihood}
\acrodef{PEB}{position error bound}
\acrodef{CEB}{clock error bound}
\acrodef{RMSE}{root-mean-squared error}
\acrodef{FIM}{Fisher Information Matrix}
\acrodef{CRLB}{Cram\'er-Rao lower bound}
\acrodef{LMR}{LoS-to-multipath ratio}

\usepackage{algorithm}
\usepackage{algpseudocode}

\algdef{SE}[SUBALG]{Indent}{EndIndent}{}{\algorithmicend\ }%
\algtext*{Indent}
\algtext*{EndIndent}

\allowdisplaybreaks

\usepackage{xr}

\makeatletter
\newcommand*{\addFileDependency}[1]{
  \typeout{(#1)}
  \@addtofilelist{#1}
  \IfFileExists{#1}{}{\typeout{No file #1.}}
}
\makeatother

\newcommand*{\myexternaldocument}[1]{%
    \externaldocument{#1}%
    \addFileDependency{#1.tex}%
    \addFileDependency{#1.aux}%
}

\myexternaldocument{supp}

\include{commands}

\ifCLASSINFOpdf
\else
\fi
\graphicspath{{./Figures/}}

\begin{document}

\bstctlcite{IEEEexample:BSTcontrol}

\title{RIS-aided Joint Localization and Synchronization with a Single-Antenna Receiver: Beamforming Design and Low-Complexity Estimation}

\author{\IEEEauthorblockN{Alessio Fascista, \IEEEmembership{Member, IEEE}, Musa Furkan Keskin, \IEEEmembership{Member, IEEE}, Angelo Coluccia, \IEEEmembership{Senior Member, IEEE}, Henk Wymeersch, \IEEEmembership{Senior Member, IEEE}, and Gonzalo Seco-Granados, \IEEEmembership{Senior Member, IEEE}}
\thanks{Part of this work was presented at the IEEE International Conference on Acoustics, Speech and Signal Processing, Toronto, Canada, June 2021 \cite{fascista2021ris}.}
\thanks{This work is supported, in part, by the EU H2020 RISE-6G project under grant 101017011, the Vinnova 5GPOS project under grant 2019-03085,  the MSCA-IF grant 888913 (OTFS-RADCOM), and in part by the Spanish Ministry of Science and Innovation projects PID2020-118984GB-I00, and by the Catalan ICREA Academia Programme.}
\thanks{A. Fascista and A. Coluccia are with the Department of Innovation Engineering, Universit\`a del Salento, Via Monteroni, 73100 Lecce, Italy (e-mail: \{alessio.fascista,angelo.coluccia\}@unisalento.it).
}
\thanks{M. F. Keskin and H. Wymeersch are with the Department of Electrical Engineering, Chalmers University of Technology, 412 96 Gothenburg, Sweden (e-mail: \{furkan,henkw\}@chalmers.se).}
\thanks{G. Seco-Granados is with the Department of Telecommunications and Systems Engineering, Universitat Aut\`onoma de Barcelona, 08193 Barcelona, Spain (e-mail: gonzalo.seco@uab.cat).}
}

\markboth{IEEE Journal of Selected Topics in Signal Processing}%
{Shell \MakeLowercase{\textit{et al.}}: A Sample Article Using IEEEtran.cls for IEEE Journals}

\maketitle

\begin{abstract}
Reconfigurable intelligent surfaces (RISs) have attracted enormous interest thanks to their ability to overcome line-of-sight blockages in mmWave systems, enabling in turn accurate localization with minimal infrastructure. Less investigated are however the benefits of exploiting RIS with suitably designed beamforming strategies for optimized localization and synchronization performance. In this paper, a novel low-complexity method for joint localization and synchronization based on an optimized design of the base station (BS) active precoding and RIS passive phase profiles is proposed, for the challenging case of a single-antenna receiver. The theoretical position error bound is first derived and used as metric to jointly optimize the BS-RIS beamforming, assuming a priori knowledge of the user position. By exploiting the low-dimensional structure of the solution, a novel codebook-based robust design strategy with optimized beam power allocation is then proposed, which provides low-complexity while taking into account the uncertainty on the user position. Finally, a reduced-complexity maximum-likelihood based estimation procedure is devised to jointly recover the user position and the synchronization offset. Extensive numerical analysis shows that the proposed joint BS-RIS beamforming scheme provides enhanced localization and synchronization performance compared to existing solutions, with the proposed estimator attaining the theoretical bounds even at low signal-to-noise-ratio and in the presence of additional uncontrollable multipath propagation.

\end{abstract}


%
\IEEEpeerreviewmaketitle

\section{Introduction}
With the introduction of 5G, radio localization has finally been able to support industrial verticals and is no longer limited to emergency call localization \cite{survey_1g_5g,henk_WCOMM_2017,TVT_2018_Gonzalo,bartoletti2021positioning,Fascista_ICASSP2020}. This ability is enabled by a combination of wideband signals (up to 400 MHz in \ac{FR2}), higher carrier frequencies (e.g., around 28 GHz), multiple antennas, and a low latency and flexible architecture \cite{henk_WCOMM_2017,dwivedi2021positioning,TWC_Letter}. Common localization methods rely on \ac{TDoA} or multi-cell round trip time (multi-RTT) measurements, requiring at least 4 or 3 \acp{BS}, respectively.
In order to enable accurate localization with minimal infrastructure, there have been several studies to further reduce the number of \acp{BS} needed for localization. These studies can be broadly grouped in three categories: \emph{(i)} data-driven, based on fingerprinting and deep learning \cite{GPS_5G,gante2020deep}; \emph{(ii)} geometry-driven, based on exploiting passive multipath in the environment \cite{ge20205g,Fascista_2021} (which is itself derived from the multipath-assisted localization \cite{witrisal2016high}); and, more recently, \emph{(iii)} \ac{RIS}-aided approaches \cite{RIS_Access_2019,hu2018beyond,wymeersch2020radio,RIS_loc_2021_TWC,RIS_bounds_TSP_2021,Keykhosravi2020_SisoRIS,RIS_loc_LCOMM}. The latter category extends the concept of multipath-aided localization to \ac{RIS}, which can actively control the multipath. \acp{RIS} have attracted enormous interest in the past few years, mainly for their ability to overcome \ac{LoS} blockages in mmWave communications \cite{di2019smart,RIS_Access_2019,RIS_commag_2021}. From the localization point of view, \ac{RIS} fundamentally offers two benefits: it introduces an extra location reference and provides additional measurements, independent of the passive, uncontrolled multipath \cite{wymeersch2020radio}. Hence, it avoids the reliance on strong reflectors in the environment, needed by standard multipath-aided localization \cite{RIS_bounds_TSP_2021}, while also having the potential to low-complexity model based solution, in contrast to deep learning methods. 

The use of \ac{RIS} for localization has only recently been developed, and a number of papers have been dedicated to \ac{RIS}-aided localization \cite{hu2018beyond,wymeersch2020radio,RIS_loc_2021_TWC,RIS_bounds_TSP_2021,Keykhosravi2020_SisoRIS,RIS_loc_LCOMM,rahal2021risenabled,he_vtc_2020}. Interestingly, \acp{RIS} allow us to solve very challenging localization problems, such as single-antenna \ac{UE} localization with a single-antenna \ac{BS} in \ac{LoS} \cite{Keykhosravi2020_SisoRIS} and even \ac{NLoS} conditions (i.e., where the \ac{LoS} path is blocked) \cite{rahal2021risenabled}. While an \ac{RIS} renders these problems solvable, high propagation losses (especially at mmWave bands) necessitates long coherent processing intervals to obtain sufficient integrated \ac{SNR}, thus limiting supported mobility. Shorter integration times can be achieved with directional beamforming at the \ac{BS} side \cite{mmwave_beamform_2014}, provided it is equipped with many antennas. Such beamforming becomes especially powerful when there exists a priori \ac{UE} location information \cite{loc_com_SPM_2014}. Hence, with the goal of improving localization performance, recent studies have focused on BS precoder optimization in the case of passive multipath \cite{successiveLocBF_2019,tasos_precoding2020}, while optimization in the presence of \ac{RIS} involves \textit{joint design of BS precoder and RIS phase profiles}, and thus can provide further accuracy enhancements via additional degrees of freedom. Nevertheless, such studies have been limited to SNR-maximizing heuristics \cite{RIS_bounds_TSP_2021}, leading to \textit{directional} RIS phase profiles, which may not necessarily lead to localization-optimal solutions \cite{tasosMultiBeam2019,keskin2021optimal}. Within the context of \textit{RIS-aided communications}, several works investigate joint design of active transmit precoding at the BS and passive phase shifts at the RIS to optimize various performance objectives, including sum-rate \cite{joint_BS_RIS_BF_JSAC_2020,joint_BS_RIS_BF_TWC_2021,joint_BS_RIS_BF_JSAC_2020_DRL,joint_BS_RIS_BF_TWC_2020}, effective mutual information \cite{joint_BS_RIS_BF_TCOM_2021}, outage probability \cite{joint_BS_RIS_BF_TSP_2021} and \ac{SINR} \cite{joint_BS_RIS_BF_TCOM_2021_Poor}. However, to the best of authors' knowledge, no studies have tackled the problem of joint BS-RIS beamforming to maximize the performance of  \textit{\ac{RIS}-aided localization and synchronization}.

In this paper, we propose a novel joint BS-RIS beamforming design and a low-complexity \ac{ML} estimator for \ac{RIS}-aided joint localization and synchronization supported by a single \ac{BS},  considering the challenging case of a UE equipped with a single-antenna receiver. The optimized design exploits a priori \ac{UE} location information and considers the \ac{BS} precoders and \ac{RIS} phase configurations jointly, in order to minimize the \ac{PEB}. 
The main contributions are as follows:
\begin{itemize}
\item We derive the \ac{FIM} 
for localization and synchronization of a UE equipped with a single-antenna receiver, and conduct a theoretical analysis of the achievable performance.
    \item We formulate the joint design of \ac{BS} precoder and \ac{RIS} phase profile as a bi-convex optimization problem for the non-robust case, and propose a solution via alternating optimization. Interestingly, the solution reveals that at both the \ac{BS} and \ac{RIS} sides, a certain sequence of beams (namely, \textit{directional} and \textit{derivative} beams \cite{li2007range,tasosMultiBeam2019,keskin2021optimal}) is required to render the problem feasible, in contrast to the corresponding communication problem. 
    \item Based on the optimal solution under perfect knowledge of UE location, we propose a codebook-based design in the robust case, including a set of BS and RIS beams determined by the uncertainty region of UE location, where power optimization across BS beams 
    is formulated as a convex problem. 
    \item Elaborating on the ideas preliminarily introduced in \cite{fascista2021ris}, we devise a reduced-complexity estimation procedure based on the \ac{ML} criterion, which attains the CRLBs even at low \ac{SNR}s, and exhibits robustness against the presence of uncontrollable multipath.
    \item We compare the proposed algorithms against different approaches in literature, and show that the proposed designs outperform these benchmarks, not only in terms of \ac{PEB} and \ac{CEB}, but also localization and synchronization \acp{RMSE}. 
\end{itemize}

\section{System Model and Problem Formulation}
In this section, we describe the \ac{RIS}-aided mmWave \ac{DL} localization scenario including a \ac{BS}, an \ac{RIS} and a \ac{UE}, derive the received signal expression at the UE, and formulate the problem of joint localization and synchronization.  

\begin{figure}%
    \centering
    {\includegraphics[width=0.45\textwidth]{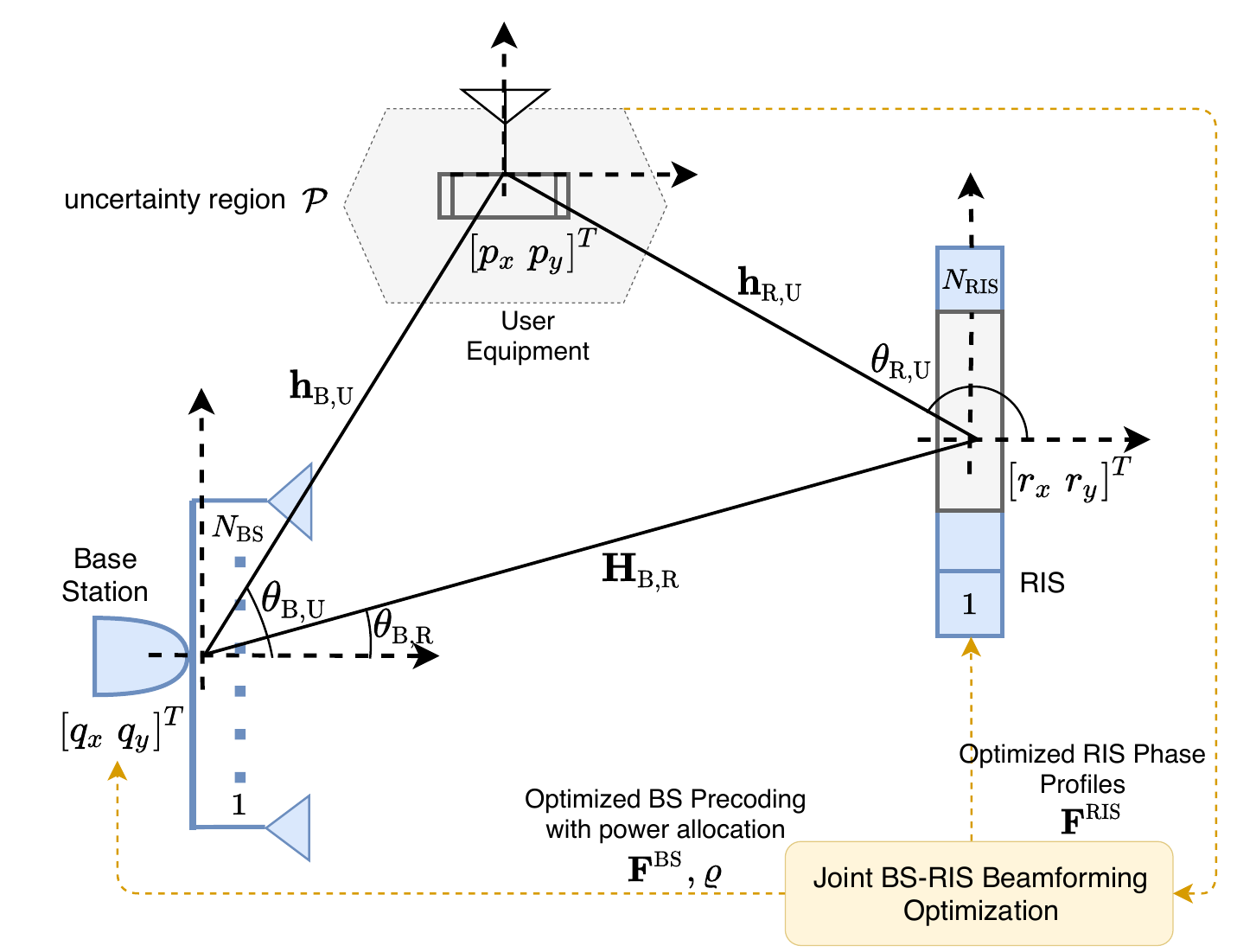}}
    \caption{Considered localization and synchronization scenario with optimized BS active precoding and RIS phase profiles.}%
    \label{fig:scenario}%
\end{figure}

\subsection{RIS-aided Localization Scenario}
We consider a \ac{DL} localization scenario, as shown in Fig.~\ref{fig:scenario}, consisting of a single BS at known location $\bm{q} = [q_x \ q_y]^\mathsf{T}$ equipped with multiple antennas, a single-antenna UE at unknown location $\bm{p} = [p_x \ p_y]^\mathsf{T}$, and an RIS at known location $\bm{r} = [r_x \ r_y]^\mathsf{T}$. The UE has an unknown clock offset $\Delta$ with respect to the BS. We assume a two-dimensional (2D) scenario with uniform linear arrays (ULAs) for both the BS and RIS deployments\footnote{\ed We address the joint localization and synchronization problem in 2D, a common choice in the literature because it greatly simplifies the exposition. Moreover, since in mmWave scenarios the distance between transmitter and receiver is large compared to the height of the antennas, considering the projection onto the 2D horizontal plane provides a fairly realistic representation of the dominant propagation phenomena. The proposed methodology can be extended in principle to address the 3D localization setup; we will discuss this possibility after the derivation of the proposed joint BS and RIS beamforming design strategy in Sec.~\ref{sec_robust_design} and low-complexity estimation algorithm in Sec.~\ref{sec_estimator}, so that the necessary modifications can be described.}. The numbers of antenna elements at the BS and RIS are $\nbs$ and $\nr$, respectively. The goal of the UE is to estimate its location and clock offset by exploiting the \ac{DL} signals it receives through the direct \ac{LoS} path and through the reflected (controllable) \ac{NLoS} path generated by the  RIS.

\subsection{Signal Model}\label{sec_tx_model}
The BS communicates by transmitting single-stream \ac{OFDM} pilots with $N$ subcarriers over $G$ transmissions. Particularly, the $g$-th transmission uses an OFDM symbol $\bm{s}_g = [\sgngen{0} \cdots [\sgngen{N-1} ]^\mathsf{T} \in \complexset{N}{1}$ with $\frac{1}{N}\norm{\ssb_g}^2 = 1$ and is precoded by the weight vector $\ff_g \in \complexset{\nbs}{1}$. To keep the transmit energy constant over the entire transmission period, the precoding matrix $\FF = [\ff_1  \cdots  \ff_G] \in \complexset{\nbs}{G}$ is assumed to satisfy $\tracesmall{\FF \FF^H} = 1$. The BS can be equipped with an analog active phased array \cite{phasedArray_2016} or a fully digital array.

The \ac{DL} received signal at the UE associated to the $g$-th transmission over subcarrier $n$ is given by
\begin{equation}\label{eq_ygn}
y^g[n] = \sqrt{P} \bm{h}^\mathsf{T}[n]\ff_g s_g[n] + \nu^g[n] 
\end{equation}
for $n = 0, \ldots, N-1$ and $g = 1, \ldots, G$, with $P$ denoting the transmit power and $\nu^g[n]$ circularly symmetric complex Gaussian noise having zero mean and variance $\sigma^2$. In \eqref{eq_ygn}, $\bm{h}[n] \in \complexset{\nbs}{1}$ represents the entire channel, including both the \ac{LoS} path and the \ac{NLoS} path (i.e., the reflection path via the RIS), between the BS and the UE for the $n$-th subcarrier: 
\begin{equation}\label{eq_hall}
\bm{h}^\mathsf{T}[n] = \bm{h}^\mathsf{T}_{{\BM}}[n] + \bm{h}^\mathsf{T}_{{\RM}}[n]\bm{\Omega}^g\bm{H}_{\text{\tiny{B,R}}}[n]
\end{equation}
where $\bm{h}_{{\BM}}[n] \in \complexset{\nbs}{1}$ is the direct (i.e., \ac{LoS}) channel between the BS and the UE, $\bm{H}_{\text{\tiny{B,R}}}[n] \in 
\complexset{\nr}{\nbs}$ denotes the channel from the BS to the RIS, $\bm{\Omega}^g = \mathrm{diag}(\e^{j\omega^g_1},\ldots, \e^{j\omega^g_{\nr}}) \in  \mathbb{C}^{\nr \times \nr}$ is the RIS phase control matrix at transmission $g$, and $\bm{h}_{{\RM}}[n] \in \mathbb{C}^{\nr \times 1}$ represents the channel from the RIS to the UE\footnote{For simplicity, the model in \eqref{eq_hall} does not involve uncontrolled multipath components and will be used to derive bounds and algorithms from Sec.~\ref{sec_fim} to Sec.~\ref{sec_estimator}, while the robustness of the algorithms against the presence of uncontrolled \ac{NLoS} paths will be evaluated through simulations in Sec.~\ref{sec_uncontrolled_path}.}.

The \ac{LoS} channel in \eqref{eq_hall} can be expressed as $\bm{h}_{{\BM}}[n] = \alpha_{{\BM}}\e^{-j2\pi n \frac{\tau_{{\BM}} }{NT}}\bm{a}_{{\text{\tiny{BS}}}}(\theta_{\BM})$, 
where $T = 1/B$ is the sampling period with $B$ the bandwidth, $\alpha_{{\BM}} = \rho_{{\BM}}\e^{j\varphi_{{\BM}}}$ with $\rho_{{\BM}}$ and $\varphi_{{\BM}}$ modulus and phase of the complex amplitude $\alpha_{{\BM}}$, $\theta_{\BM}$ is the \ac{AoD} from the BS to the UE, and $\tau_{{\BM}}$ is the delay between the BS and the UE up to a clock offset $\Delta$, as better specified later. As to $\bm{a}_{{\text{\tiny{BS}}}}(\cdot)$, it represents the BS array steering vector whose expression is given by
$\bm{a}_{{\text{\tiny{BS}}}}(\theta) = [1 \ e^{j\frac{2\pi}{\lambda_c} d\sin\theta} \cdots \ e^{j(N_{\text{\tiny BS}}-1)\frac{2\pi}{\lambda_c}d\sin\theta}]^\mathsf{T}$
with $\lambda_c = c/f_c$, $f_c$ being the carrier frequency and $c$ the speed of light, and $d = \lambda_c/2$.

In \eqref{eq_hall}, the first tandem channel (i.e., from the BS to the RIS) in the \ac{NLoS} path through the RIS is defined as
\begin{equation}\label{eq_hbr}
\bm{H}_{\text{\tiny{B,R}}}[n] = \alpha_{\text{\tiny{B,R}}}\e^{-j2\pi n \frac{\tau_{\text{\tiny{B,R}}}}{NT}}\bm{a}_{{\text{\tiny{RIS}}}}(\phi_\text{\tiny{B,R}})\bm{a}^\mathsf{T}_{{\text{\tiny{BS}}}}(\theta_\text{\tiny{B,R}})
\end{equation}
where $\alpha_{\text{\tiny{B,R}}} = \rho_{\text{\tiny B,R}}\e^{j\varphi_{\text{\tiny B,R}}}$ is the complex gain over the BS-RIS path, $\phi_\text{\tiny{B,R}}$ the \ac{AoA} and $\theta_\text{\tiny{B,R}}$  the \ac{AoD} from the BS to the RIS, and $\tau_{\text{\tiny{B,R}}}$ the delay between the BS and the RIS. In addition, $\bm{a}_{{\text{\tiny{RIS}}}}(\cdot) \in \complexset{\nr}{1}$ denotes the array steering vector of the RIS, given by 
$\bm{a}_{{\text{\tiny{RIS}}}}(\theta) = [1 \ e^{j\frac{2\pi}{\lambda_c} d\sin\theta} \cdots \ e^{j(\nr-1)\frac{2\pi}{\lambda_c}d\sin\theta}]^\mathsf{T}.$
Finally, the second tandem channel in \eqref{eq_hall} is given by
\begin{equation}\label{eq_hru}
    \bm{h}^\mathsf{T}_{{\RM}}[n] = \alpha_{{\RM}}\e^{-j2\pi n \frac{\tau_{{\RM}} }{NT}}\bm{a}^\mathsf{T}_{{\text{\tiny{RIS}}}}(\theta_{\RM})
\end{equation}
with the notations $\alpha_{{\RM}} = \rho_{{\RM}}\e^{j\varphi_{{\RM}}}$, $\tau_{{\RM}}$, and $\theta_{\RM}$ having the same meaning as in the BS-to-UE channel model.

The geometric relationships among the BS, RIS, and UE are as follows (assuming for simplicity that the BS is placed at the origin of the reference system, i.e., $\bm{q} = [0 \ 0]^\mathsf{T}$):
\begin{align}\label{geomrelationships}
&\tau_{{\BM}} = \|\bm{p}\|/c +\Delta \nonumber \\
& \tau_{\text{\tiny{R}}} = \tau_{\text{\tiny{B,R}}} +  \tau_{{\RM}} = (\|\bm{r}\| + \|\bm{r} - \bm{p}\|)/c +\Delta \nonumber \\          &\theta_{{\BM}} = \mathrm{atan2}(p_y,p_x), \quad \theta_{{\RM}} = \mathrm{atan2}(p_y-r_y,p_x-r_x)\nonumber \\
&\theta_{\text{\tiny{B,R}}} = \mathrm{atan2}(r_y,r_x), \quad \phi_{\text{\tiny{B,R}}} = -\pi + \theta_{\text{\tiny{B,R}}}.
\end{align}
Notice that $\tau_{\text{\tiny{B,R}}}$, $\theta_{\text{\tiny{B,R}}}$ and $\phi_{\text{\tiny{B,R}}}$ are known quantities being the BS and RIS placed at known positions.

\subsection{Joint Localization and Synchronization Problem}\label{sec_prob_form}
From the DL received signal $\{\ygn\}_{\forall n, g}$ in \eqref{eq_ygn} over $N$ subcarriers and $G$ transmissions, the problems of interest are as follows: i) design the BS precoder matrix $\FF$ and the RIS phase profiles $\{\Omegag\}_{\forall g}$ to maximize the accuracy of UE location and clock offset estimation; ii) estimate the unknown location $\pp$ and the unknown clock offset $\Delta$ of the UE.
To tackle these problems, we first derive a performance metric to quantify the accuracy of localization and synchronization in Sec.~\ref{sec_fim}. Based on this metric, Secs.~\ref{sec_joint_design_opt}-\ref{sec_robust_design} focus on the joint design of $\FF$ and $\{\Omegag\}_{\forall g}$. Finally, Sec.~\ref{sec_estimator} develops an estimator for $\pp$ and $\Delta$.

\section{Fisher Information Analysis}\label{sec_fim}
In this section, we perform a Fisher information analysis to obtain a performance measure for localization and synchronization of the UE, which is needed for the design of $\FF$ and $\{\Omegag\}_{\forall g}$ in Sec.~\ref{sec_joint_design_opt} and Sec.~\ref{sec_robust_design}.

\subsection{\ac{FIM} in the Channel Domain}
For the estimation problem in Sec.~\ref{sec_prob_form}, we compute the \ac{FIM} of the unknown channel parameter vector $ \gammab = [\tau_{{\BM}} \ \theta_{{\BM}} \ \rho_{{\BM}} \
    \varphi_{{\BM}} \ \tau_{{\RM}} \ \theta_{{\RM}} \ \rho_{\text{\tiny{R}}} \ \varphi_{\text{\tiny{R}}} ]^\mathsf{T}$
where $\rho_\text{\tiny{R}} = \rho_{\text{\tiny B,R}}\rho_{{\RM}}$ and $\varphi_\text{\tiny{R}} = \varphi_\text{\tiny{B,R}} + \varphi_{\RM}$.
The FIM  $\JJgamma \in \mathbb{R}^{8 \times 8}$ satisfies the information inequality \cite[Thm.~(3.2)]{kay1993fundamentals}
\begin{align}
    \E\left\{  (\gammabhat - \gammab) (\gammabhat - \gammab)^\trpose \right\} \succeq \JJgamma^{-1}
\end{align}
for any unbiased estimator $\gammabhat$ of $\gammab$, where $\AAb \succeq \BB$ means $\AAb - \BB$ is positive semi-definite. Since the observations in \eqref{eq_ygn} are complex Gaussian, the $(h,k)$-th FIM entry $[\JJgamma]_{h,k} \eqdef \Lambda(\gamma_h,\gamma_k)$ can be expressed using the Slepian-Bangs formula as \cite[Eq.~(15.52)]{kay1993fundamentals}
\begin{equation}\label{eq::FIMelemformula}
\Lambda(\gamma_h,\gamma_k) = \frac{2}{\sigma^2} \sum_{g=1}^{G} \sum_{n=0}^{N-1} \Re\left\{\left(\frac{\partial m^g[n]}{\partial \gamma_h}\right)^\conj \frac{\partial m^g[n]}{\partial \gamma_k}\right\}     
\end{equation}
where $x^\conj$ denotes the complex conjugate of $x$ and $m^g[n] = \sqrt{P} \bm{h}^\mathsf{T}[n] \ff_g s_g[n]$
is the noise-free version of the received signal in \eqref{eq_ygn}. Using \eqref{eq_hall}--\eqref{eq_hru}, $m^g[n]$ 
can be re-written as $\mgn = \mgnbm + \mgnr$, 
where
\begin{align}\label{eq_mgnbm}
    \mgnbm &\eqdef \sqrt{P} \rhobm \e^{j\varphibm} \left[ \cc(\taubm) \right]_n \aabs^\trpose(\thetabm) \ff_g \sgn \\ \nonumber
    \mgnr &\eqdef \sqrt{P} \rhor \e^{j\varphir} \left[ \cc(\taur) \right]_n \aarisw^\trpose(\thetarm) \oomegag \aabs^\trpose(\thetabr) \ff_g \sgn  ~.
\end{align}
In \eqref{eq_mgnbm}, $\taur = \taurm + \taubr$ is the delay of the BS-RIS-UE path, 
\begin{align}
    \aarisw(\theta) \eqdef \aaris(\theta) \odot \aaris(\phibr)
\end{align}
denotes the combined RIS steering vector including the effect of both the \ac{AoD} $\theta$ and the \ac{AoA} $\phibr$ as a function of $\theta$, 
\begin{align}
    \cc(\tau) \eqdef \left[ 1 ~ \e^{-j \kappa_1 \tau} \, \cdots \, \e^{-j \kappa_{N-1} \tau} \right]^\trpose
\end{align}
represents the frequency-domain steering vector with $\kappa_n = 2\pi \frac{n}{NT}$, and $\oomegag \in \complexset{\nr}{1}$ is the vector consisting of the diagonal entries of $\Omegag$, i.e., $\Omegag = \diag{\oomegag}$. Here, $\odot$ is the Hadamard (element-wise) product. Hereafter, $\aarisw$ will be used to denote $\aarisw(\thetarm)$ for the sake of brevity. For the derivative expressions in \eqref{eq::FIMelemformula}, we refer the reader to Appendix~\ref{sec_der_fim}.

We now express the FIM elements in \eqref{eq::FIMelemformula} as a function of the BS precoder $\FF$ and the RIS phase profiles $\{ \oomegag \}_{g=1}^{G}$. To that end, the FIM $\JJgamma$ can be written as 
\begin{align}\label{eq_JJgamma_bdiag}
    \JJgamma = \begin{bmatrix} \JJbm & \JJcross \\
    \JJcross^\trpose & \JJr
    \end{bmatrix} 
\end{align}
where $\JJbm \in \realset{4}{4}$ and $\JJr \in \realset{4}{4}$ are the FIM submatrices corresponding to the \ac{LoS} path and the \ac{NLoS} (i.e., BS-RIS-UE) path, respectively, and $\JJcross \in \realset{4}{4}$ represents the \ac{LoS}-\ac{NLoS} path cross-correlation. In addition, let us define
\begin{align} \label{eq_ffg}
    \XX_g &\eqdef \ff_g \ff_g^\hermit \in \complexset{\nbs}{\nbs} 
    \\ \label{eq_oomegag}
    \Psibbig_g &\eqdef \oomegag (\oomegag)^\hermit \in \complexset{\nr}{\nr} ~.
\end{align}
The following \rev{remark} reveals the dependency of the FIM submatrices in \eqref{eq_JJgamma_bdiag} on BS precoder and RIS phase profiles \rev{using} Appendix~\ref{sec_fim_func}.
\rev{\begin{remark}    \label{lemma_fim_func}
    The dependency of the FIM $\JJgamma$ in \eqref{eq_JJgamma_bdiag} on $\FF$ and $\{ \oomegag \}_{g=1}^{G}$ can be specified as follows:
    \begin{itemize}
        \item $\JJbm$ is a linear function of $\{\XX_g\}_{g=1}^{G}$.
        \item $\JJr$ is a bi-linear function of $\{\XX_g\}_{g=1}^{G}$ and $\{\Psibbig_g\}_{g=1}^{G}$.
        \item $\JJcross$ is a bi-linear function of $\{\XX_g\}_{g=1}^{G}$ and $\{\oomegag\}_{g=1}^{G}$.
    \end{itemize}
\end{remark}}

\subsection{FIM in the Location Domain}
To obtain the location-domain FIM from the channel-domain FIM $\JJgamma$ in \eqref{eq_JJgamma_bdiag}, we apply a transformation of variables from the vector of the unknown channel parameters $\bm{\gamma}$ to the vector of location parameters\rev{\footnote{\rev{Note that the channel gains $\rho_{\BM}$, $\varphi_{\BM}$, $\rho_\text{\tiny R}$ and $\varphi_\text{\tiny R}$ are nuisance parameters that need to be estimated for localization, but do not convey any geometric information that can be useful for localization. Hence, they cannot be expressed as a function of other unknown (geometric) parameters and thus appear in both channel and location domain parameter vectors.}}}
\begin{equation}\label{eq_etab}
    \bm{\eta} = \left[p_x \ p_y  \ \rho_{\BM} \ \varphi_{\BM} \ \rho_\text{\tiny R} \ \varphi_\text{\tiny R} \  \Delta \right]^\mathsf{T}. 
\end{equation}
The FIM of $\bm{\eta}$, denoted as $\bm{J}_{\bm{\eta}} \in \realset{7}{7}$, is obtained by means of the transformation matrix $\bm{T} \eqdef \frac{\partial \bm{\gamma}^{\mathsf{T}}}{\partial \bm{\eta}} \in \realset{7}{8}$ as
\begin{equation}\label{eq_Jeta}
\bm{J}_{\bm{\eta}} = \bm{T}\bm{J}_{\bm{\gamma}}\bm{T}^\mathsf{T} ~,
\end{equation}
which preserves the linearity and bi-linearity properties of $\JJgamma$ in \rev{Remark}~\ref{lemma_fim_func}. Please see Appendix~\ref{sec_transform_mat} for the expressions of the elements of $\bm{T}$.


\section{Joint Transmit Precoding and RIS Phase Profile Design}\label{sec_joint_design_opt}
In this section, assuming \textit{perfect knowledge} of  $\etab$ in \eqref{eq_etab}, we tackle the problem of joint design of the transmit BS precoding matrix $\FF$ and the RIS phase profiles $\{ \oomegag \}_{g=1}^{G}$ to maximize the performance of joint localization and synchronization of the UE. First, we apply convex relaxation and alternating optimization techniques to obtain two convex subproblems to optimize BS precoders for a given RIS phase profile and vice versa. Then, we demonstrate the \textit{low-dimensional structure} of the optimal BS precoders and RIS phase profiles, which will be highly instrumental in Sec.~\ref{sec_robust_design} in designing codebooks under imperfect knowledge of UE location.

\subsection{Problem Formulation for Joint Optimization}
To formulate the joint BS precoding and RIS phase profile optimization problem, we adopt the position error bound (PEB) as metric\footnote{Since positioning and synchronization are tightly coupled \cite{Mendrzik_JSTSP_2019}, considering PEB as the optimization metric would also improve the synchronization performance, which will be verified through simulation results in Sec.~\ref{sec_sim_res}. \rev{In particular, please see Fig.~\ref{fig:comparison_RMSE_3m} for an illustration of how PEB-based optimization provides noticeable improvements in the RMSE of both the position and clock offset over the benchmark schemes. For a more detailed comparison between PEB- and CEB-based optimization, we refer the reader to Appendix~\ref{sec::ceb_metric}.}}. From \eqref{eq_etab}-\eqref{eq_Jeta} and \rev{Remark}~\ref{lemma_fim_func}, the PEB can be obtained as a function of the BS beam covariance matrices, the RIS phase profiles and their covariance matrices $\{ \XX_g, \oomegag, \Psibbig_g  \}_{g=1}^{G}$, as follows:
\begin{align} \nonumber
      \E\big\{  \norm{\pphat - \pp}^2  \big\} &\geq  \tracesmall{  \left[ \Jeta^{-1} \right]_{1:2,1:2}  } 
      \\ \label{eq_peb_fim}
      &\eqdef \fpeb\left(\{ \XX_g, \oomegag, \Psibbig_g  \}_{g=1}^{G}; \etab\right)~.
\end{align}
We note that the PEB depends on the unknown parameters $\etab$ (see Appendix~\ref{sec_fim_func}). \rev{In addition, the dependency of the PEB on $\{ \XX_g, \oomegag, \Psibbig_g  \}_{g=1}^{G}$ can be observed through \eqref{eq_Jeta}, \eqref{eq_peb_fim} and Remark~\ref{lemma_fim_func}.} Under perfect knowledge of $\etab$, the PEB minimization problem can be formulated as
\begin{subequations} \label{eq_problem_peb}
\begin{align} \label{eq_problem_obj}
	\mathop{\mathrm{min}}\limits_{\{ \XX_g, \oomegag, \Psibbig_g  \}_{g=1}^{G}} &~~
	\fpeb\left(\{ \XX_g, \oomegag, \Psibbig_g  \}_{g=1}^{G}; \etab\right)
	\\ \label{eq_power_cons}
    \mathrm{s.t.}&~~ \tracebig{\sum_{g=1}^{G}\XX_g} = 1 ~, 
    \\ \label{eq_rank_Xg}
    &~~ \XX_g \succeq 0 \,,\, \rank(\XX_g) = 1 \, ,
    \\ \label{eq_psig_cons}
    &~~ \Psibbig_g = \oomegag (\oomegag)^\hermit \, , \, \abs{\omega_g[n]} = 1 ~,
    \\ \nonumber
    &~~ g=1,\ldots,G ~,
\end{align}
\end{subequations}
where the total power constraint in \eqref{eq_power_cons} is due to Sec.~\ref{sec_tx_model}, \eqref{eq_rank_Xg} results from \eqref{eq_ffg}, and \eqref{eq_psig_cons} comes from the definition in \eqref{eq_oomegag} and the unit-modulus constraints on the elements of the RIS control matrix. We now provide the following lemma on the structure of the objective function in \eqref{eq_problem_obj}.
\begin{lemma}\label{lemma_biconvex}
 $\fpeb\left(\{ \XX_g, \oomegag, \Psibbig_g  \}_{g=1}^{G}; \etab\right)$ is a multi-convex function of $\{ \XX_g  \}_{g=1}^{G}$, $\{ \oomegag  \}_{g=1}^{G}$ and $\{ \Psibbig_g  \}_{g=1}^{G}$.
\end{lemma}
\begin{proof}
    From \rev{Remark}~\ref{lemma_fim_func} and \eqref{eq_Jeta}, we see that $\JJeta$ is a multi-linear function of $\{ \XX_g  \}_{g=1}^{G}$, $\{ \oomegag  \}_{g=1}^{G}$ and $\{ \Psibbig_g  \}_{g=1}^{G}$. Based on the composition rules \cite[Ch.~3.2.4]{boyd2004convex}, $\tracenormal{  \left[ \Jeta^{-1} \right]_{1:2,1:2}  }$ is a multi-convex function of $\{ \XX_g  \}_{g=1}^{G}$, $\{ \oomegag  \}_{g=1}^{G}$ and $\{ \Psibbig_g  \}_{g=1}^{G}$.
\end{proof}

\rev{To clarify the use of $\{ \XX_g, \oomegag, \Psibbig_g  \}_{g=1}^{G}$ instead of $\{ \ff_g, \oomegag  \}_{g=1}^{G}$ in \eqref{eq_problem_peb} as the optimization variables, the following remark is provided.}
\rev{\begin{remark}    \label{lemma_coupled}
    While our goal is to optimize the BS precoders $\ff_g$ and the RIS phase profiles $ \oomegag$, we employ the covariances $\XX_g = \ff_g \ff_g^\hermit$ and $\Psibbig_g = \oomegag (\oomegag)^\hermit$ as the optimization variables in \eqref{eq_problem_peb}. The reason is that the FIM $\JJgamma$ in \eqref{eq_JJgamma_bdiag} is \textit{linear} in $\XX_g$ and $\Psibbig_g$, but \textit{quadratic} in $\ff_g$ and $\oomegag$, as seen from Appendix~\ref{sec_fim_func}. In particular, 
    \begin{itemize}
        \item all the submatrices in \eqref{eq_JJgamma_bdiag} are linear in $\XX_g$, but quadratic in $\ff_g$,
        \item $\JJcross$ is linear in $\oomegag$, and $\JJr$ is linear in $\Psibbig_g$, but quadratic in $\oomegag$,
    \end{itemize}
    as specified in Lemma~1. For PEB minimization, we wish to keep the variables for which the dependencies are linear and discard the remaining ones. This results from the fact that the PEB minimization problem, when written in the epigraph form as will be shown in \eqref{eq_peb_robust3}, induces a matrix inequality (MI) constraint that involves the FIM, such as in \eqref{eq_problem_peb_robust3_discrete_cons1}, which is convex only if the MI is linear \cite[Ch.~4.6.2]{boyd2004convex}, \cite{fukuda2001branch}. Therefore, to have a convex problem, the FIM needs to depend linearly on the optimization variables. This implies that we should keep $\XX_g$, $\Psibbig_g$ and $\oomegag$ as the variables in \eqref{eq_problem_peb}, where $\JJcross$ is defined as a linear function of $\oomegag$, $\JJr$ is defined as a linear function of $\Psibbig_g$ and the coupling between the two variables $\Psibbig_g = \oomegag (\oomegag)^\hermit$ is imposed as a constraint in \eqref{eq_psig_cons}.
\end{remark}}

\subsection{Relaxed Problem for PEB Minimization}\label{sec_relaxed_peb}
To transform \eqref{eq_problem_peb} into a tractable form, we will perform two simplifications. First, we approximate the channel-domain FIM in \eqref{eq_JJgamma_bdiag} as a block-diagonal matrix, i.e.,
\begin{align}\label{eq_JJgamma_bdiag_approx}
    \JJgamma \approx \JJgammablkdiag \eqdef \begin{bmatrix} \JJbm & \boldzero \\
    \boldzero & \JJr
    \end{bmatrix} ~,
\end{align}
by assuming $\JJcross \approx \boldzero$, which can be justified by the assumption of non-interfering paths under the large bandwidth and large array regime \cite{mendrzik2018harnessing,abu2018error,kakkavas2019performance}. Based on \rev{Remark}~\ref{lemma_fim_func}, this enables removing the dependency of $\JJgamma$ and, consequently, the PEB in \eqref{eq_peb_fim} on $\{ \oomegag \}_{g=1}^{G}$. In this case, the constraint in \eqref{eq_psig_cons} should be replaced by
\begin{align}\label{eq_rank_psig}
    \Psibbig_g \succeq 0 \,,\, \rank(\Psibbig_g) = 1 \, , \,
    \diag{\Psibbig_g} = \boldone ~.
\end{align}
Second, we drop non-convex rank constraints in \eqref{eq_rank_Xg} and \eqref{eq_rank_psig}\!. 

After these simplifications\footnote{It should be emphasized that these two relaxations are performed only to reveal the underlying \textit{low-dimensional structure} of the optimal BS and RIS transmission strategies in Sec.~\ref{sec_ao_optimal}, which facilitates \textit{codebook design} in Sec.~\ref{sec_codebook}. For power optimization in Algorithm~\ref{alg_codebook} of Sec.~\ref{sec_codebook} and for simulation results in Sec.~\ref{sec_sim_res}, we employ the true FIM $\JJgamma$ instead of the approximated FIM $\JJgammablkdiag$.}, a relaxed version of the PEB optimization problem in \eqref{eq_problem_peb} can be cast as
\begin{subequations} \label{eq_problem_peb_relaxed}
\begin{align} \label{eq_problem_peb_relaxed_obj}
	\mathop{\mathrm{min}}\limits_{\{ \XX_g, \Psibbig_g  \}_{g=1}^{G}} &~~
	\fpebblkdiag\left(\{ \XX_g, \Psibbig_g  \}_{g=1}^{G}; \etab\right)
	\\ \label{eq_const_xx}
    \mathrm{s.t.}&~~ \eqref{eq_power_cons}\, , \, \XX_g \succeq 0 ~, 
    \\ \label{eq_const_psi_g}
    &~~ \Psibbig_g \succeq 0 \, , \, \diag{\Psibbig_g} = \boldone ~, \\ 
    \nonumber
    &~~  g=1,\ldots,G ~, 
\end{align}
\end{subequations}
where $  \fpebblkdiag(\{ \XX_g, \Psibbig_g  \}_{g=1}^{G}; \etab) \eqdef \mathrm{tr}(  [ ( \bm{T} \JJgammablkdiag \bm{T}^\mathsf{T} ) ^{-1} ]_{1:2,1:2}  )$.
Using Lemma~\ref{lemma_biconvex} and the linearity of the constraints in \eqref{eq_const_xx} and \eqref{eq_const_psi_g}, it is observed that the problem \eqref{eq_problem_peb_relaxed} is convex in $\{ \XX_g  \}_{g=1}^{G}$ for fixed $\{ \Psibbig_g  \}_{g=1}^{G}$ and convex in $\{ \Psibbig_g  \}_{g=1}^{G}$ for fixed  $\{ \XX_g  \}_{g=1}^{G}$. This motivates alternating optimization to solve \eqref{eq_problem_peb_relaxed}, iterating BS precoders update for fixed RIS phase profiles and RIS phase profiles update for fixed BS precoders.

\subsection{Alternating Optimization to Solve Relaxed Problem}\label{sec_ao_optimal}
\subsubsection{Optimize BS Precoders for Fixed RIS Phase Profiles}
For fixed $\{ \Psibbig_g  \}_{g=1}^{G}$, the subproblem of \eqref{eq_problem_peb_relaxed} to optimize $\{ \XX_g  \}_{g=1}^{G}$ can be expressed as 
\begin{align} \label{eq_problem_peb_relaxed_fixed_RIS}
	\mathop{\mathrm{min}}\limits_{\{ \XX_g \}_{g=1}^{G}} &~~
	\fpebblkdiag\left(\{ \XX_g, \Psibbig_g  \}_{g=1}^{G}; \etab\right)
	\\ \nonumber 
    \mathrm{s.t.}&~~ \eqref{eq_const_xx} ~,
\end{align}
which is a convex problem and can be solved using off-the-shelf solvers \cite{cvx}. To achieve low-complexity optimization, we can exploit the low-dimensional structure of the optimal precoder covariance matrices, as shown in the following result.

\begin{proposition}\label{prop_BS_precoder}
    The optimal BS precoder covariance matrices $\{ \XX_g \}_{g=1}^{G}$ in \eqref{eq_problem_peb_relaxed_fixed_RIS} can be written as $\XX_g = \aabsbig \Upsilonb_g \aabsbig^\hermit $
    where
    \begin{align} \label{eq_a_bs}
        \aabsbig \eqdef \left[ \aabs(\thetabr) ~ \aabs(\thetabm) ~  \aabsdt(\thetabm)  \right]^\conj ~,
    \end{align}
    $\aabsdt(\theta) \eqdef \partial \aabs(\theta)/\partial \theta$ and $\Upsilonb_g \in \complexset{3}{3}$ is a positive semidefinite matrix.
\end{proposition}
\begin{proof}
    Please see Appendix~\ref{app_proof_BS}.
\end{proof}

\subsubsection{Optimize RIS Phase Profiles for Fixed BS Precoders}
For fixed $\{ \XX_g  \}_{g=1}^{G}$, we can formulate the subproblem of \eqref{eq_problem_peb_relaxed} to optimize $\{ \Psibbig_g  \}_{g=1}^{G}$ as follows:
\begin{align} \label{eq_problem_peb_relaxed_fixed_BS}
	\mathop{\mathrm{min}}\limits_{\{ \Psibbig_g \}_{g=1}^{G}} &~~
	\fpebblkdiag\left(\{ \XX_g, \Psibbig_g  \}_{g=1}^{G}; \etab\right)
	\\ \nonumber 
    \mathrm{s.t.}&~~  \eqref{eq_const_psi_g} ~, 
\end{align}
which is again a convex problem \cite{boyd2004convex}. Similar to \eqref{eq_problem_peb_relaxed_fixed_RIS}, the inherent low-dimensional structure of the optimal phase profiles can be exploited to obtain fast solutions to \eqref{eq_problem_peb_relaxed_fixed_BS}, as indicated in the following proposition.

\begin{proposition}\label{prop_RIS_profile}
    The optimal RIS phase profile covariance matrices $\{ \Psibbig_g  \}_{g=1}^{G}$ in \eqref{eq_problem_peb_relaxed_fixed_BS} in the absence of the unit-modulus constraints $\diag{\Psibbig_g} = \boldone$ can be expressed as $\Psibbig_g = \aariswbig \, \Xib_g \, \aariswbig^\hermit$, 
    where
    \begin{align} \label{eq_bris}
        \aariswbig \eqdef \left[ \aarisw  ~ \aariswdt \right]^\conj ~,
    \end{align}
    $\aariswdt(\theta) \eqdef \partial \aarisw(\theta)/\partial \theta$, $\aariswdt \equiv \aariswdt(\thetarm)$ and $\Xib_g \in \complexset{2}{2}$ is a positive semidefinite matrix.
\end{proposition}
\begin{proof}
    Please see Appendix~\ref{app_proof_RIS}.
\end{proof}
\rev{Fig.~\ref{fig_beams_prop12} provides a graphical representation of the beams in \eqref{eq_a_bs} and \eqref{eq_bris}.}

\rev{\begin{remark}    \label{lemma_alt_opt}
    It is worth emphasizing that we never solve the problem \eqref{eq_problem_peb_relaxed} to obtain $\XX_g$ and $\Psibbig_g$. The sole purpose of the alternating optimization is to formulate the subproblems \eqref{eq_problem_peb_relaxed_fixed_RIS} and \eqref{eq_problem_peb_relaxed_fixed_BS}, and, based on that, to uncover the low-dimensional structure of the optimal BS and RIS transmission strategies, as shown in Prop.~\ref{prop_BS_precoder} and Prop.~\ref{prop_RIS_profile}. The derived low-dimensional structure will be exploited in Sec.~\ref{sec_robust_design} to design the codebooks in \eqref{eq_prop_codebook} under imperfect knowledge of UE location. Hence, the aim of Sec.~\ref{sec_joint_design_opt} is not to solve the PEB minimization problem under perfect knowledge of UE location, but to extract analytical insights from the structure of the solution that will be conducive to tackling the more practical problem of PEB optimization under UE location uncertainty in Sec.~\ref{sec_robust_design}. 
\end{remark}}

\subsection{Interpretation of Proposition~\ref{prop_BS_precoder} and Proposition~\ref{prop_RIS_profile}}\label{sec_interp}
By focusing on the optimal structure of the precoder covariance matrices obtained in Prop.~\ref{prop_BS_precoder}, it emerges that the BS should transmit different beams along the two main directions of the \ac{AoD}s $\thetabr$ and $\thetabm$, i.e., the BS should serve both the RIS and the UE. Interestingly, a sort of asymmetry exists in \eqref{eq_a_bs}: while for the \ac{AoD} with respect to the RIS, the optimal structure of the precoder includes only the \textit{directional} beam $\aabs(\thetabr)$, for the \ac{AoD} with respect to the UE, the BS employs both a \textit{directional} beam $\aabs(\thetabm)$ and its \textit{derivative} $\aabsdt(\thetabm)$ \cite{li2007range,tasosMultiBeam2019,keskin2021optimal}. This  can be explained by noting that, for positioning purposes, the UE needs to estimate the \ac{AoD} with respect to the BS, and to do so a certain degree of diversity in the received beams should exist \cite{beamspace_loc_2019_TSP,Fascista_FirstTWC}. 
 
On the other hand, in the first tandem channel between the BS and the RIS, there is no need to estimate the \ac{AoD} $\thetabr$ (its value is known a priori, given the known positions of both BS and RIS), and from a PEB perspective, the transmitted power should be concentrated in a single directional beam towards the RIS, so as to maximize the received SNR over the whole BS-RIS-UE channel.

Similar conclusions can be derived from Prop.~\ref{prop_RIS_profile}. Namely, RIS phase profiles should be steered towards the \ac{AoD} $\thetarm$ with respect to the UE. In addition, both the \textit{directional} beam $\aarisw(\thetarm)$ and its \textit{derivative} $\aariswdt(\thetarm)$ should be employed to maximize the performance of \ac{AoD} estimation at the UE, which corresponds to the same principle as used in \textit{sum} and \textit{difference} beams of monopulse radar \cite{monopulse_review}.


\begin{figure}%
    \centering
    \subfloat[\centering \ed Beampatterns of the optimal BS beams $ \aabsbig$ in \eqref{eq_a_bs}.]{{\includegraphics[width=0.35\textwidth]{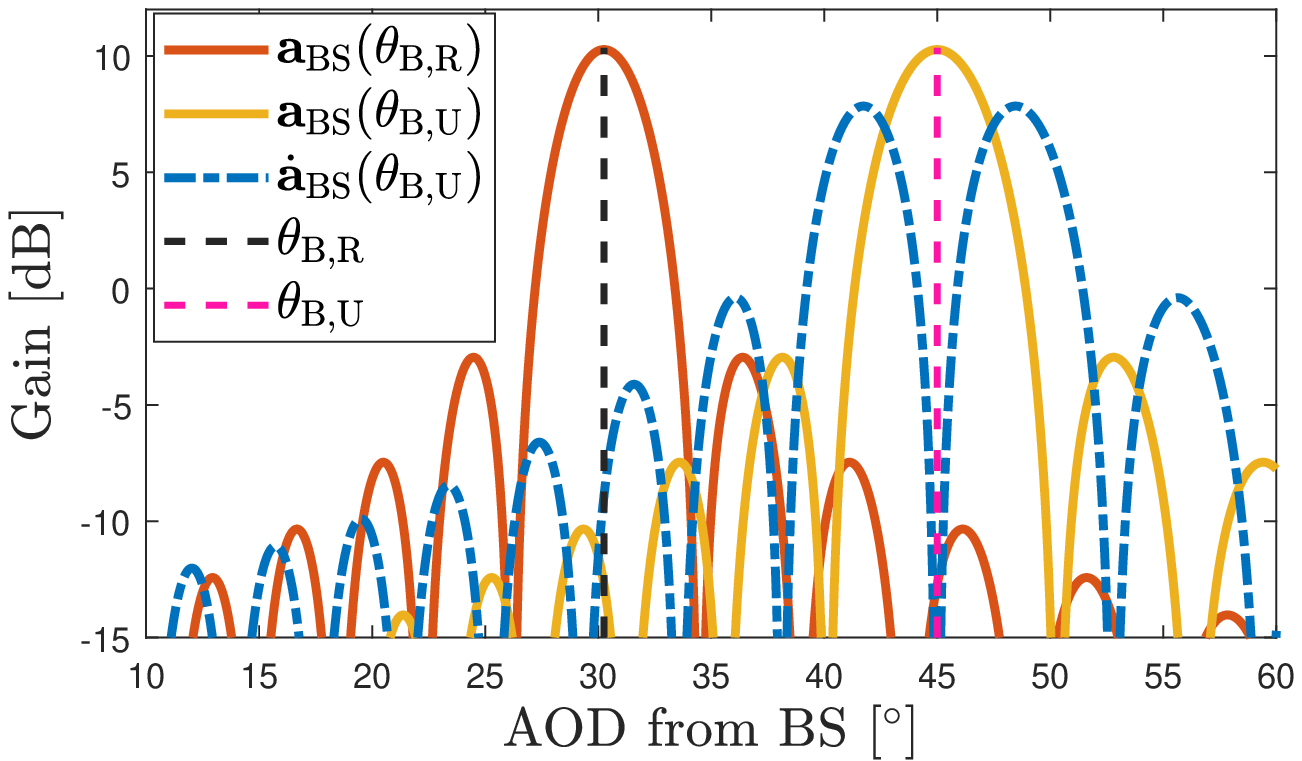} } \label{fig_beams_prop1}}%
    \qquad
    \subfloat[\centering \ed Beampatterns of the optimal RIS beams $\aariswbig$ in \eqref{eq_bris}.]{{\includegraphics[width=0.35\textwidth]{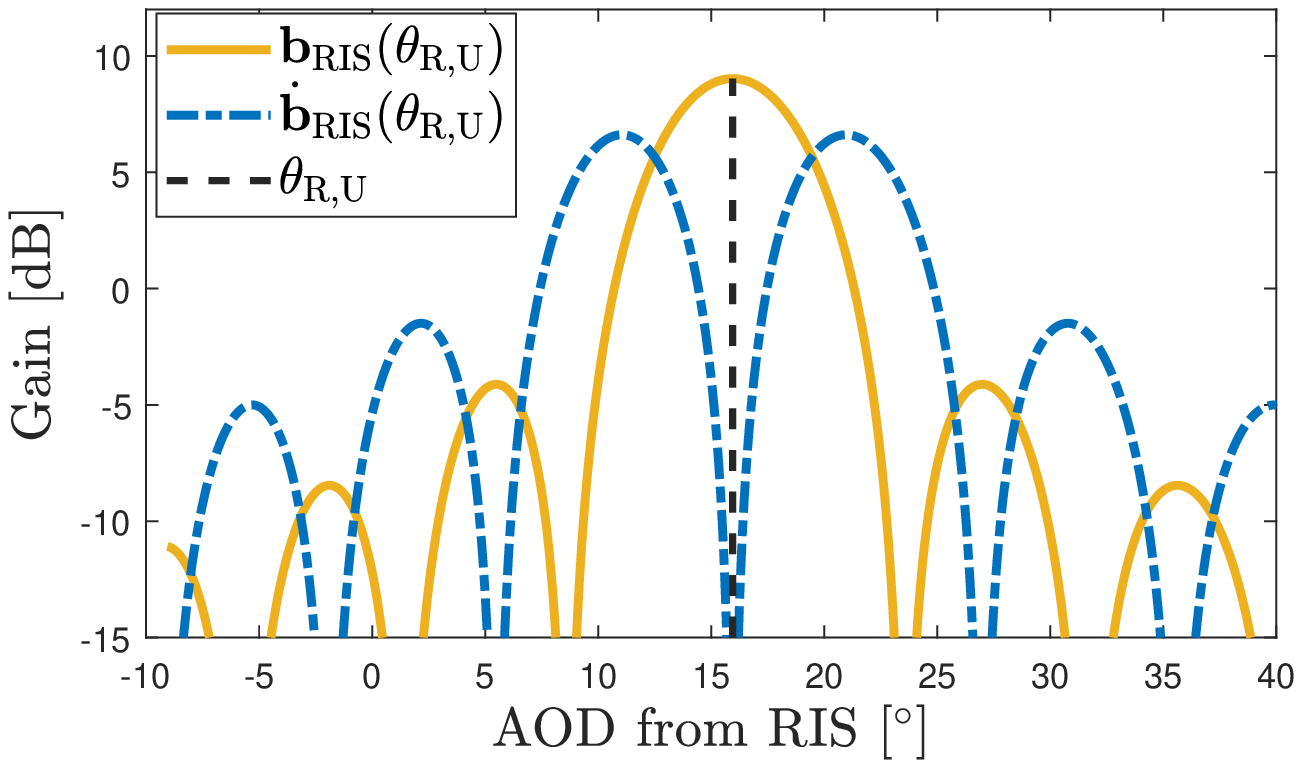} }\label{fig_beams_prop2}}
    \caption{\rev{Beampatterns of the localization-optimal BS and RIS beams, including both the directional and derivative beams, obtained for the setup in Sec.~\ref{sec_sim_setup}. The \textit{directional} beams maximize the SNR at the targeted UE location ($\thetabm$ and $\thetarm$), which serves to provide sufficient SNR for localization, while the \textit{derivative} beams enable the UE to detect small deviations around the nominal direction, similar to monopulse track radars \cite{monopulse_review,zhang2005monopulse}, which can be noticed through sharp bending of the beampattern around $\thetabm$ and $\thetarm$. This sharp curvature around the targeted location allows small deviations in angle to induce large changes in amplitude, thereby facilitating highly accurate mapping from complex amplitude measurements to angles.}}%
    \label{fig_beams_prop12}%
\end{figure}


\section{Robust Joint Design of BS Precoder and RIS Phase Profiles under Location Uncertainty}\label{sec_robust_design}
In this section, inspired by Prop.~\ref{prop_BS_precoder} and Prop.~\ref{prop_RIS_profile} in Sec.~\ref{sec_joint_design_opt}, we develop \textit{robust joint design} strategies for BS precoder and RIS phase profiles under \textit{imperfect knowledge of UE location} $\pp$ in \eqref{eq_etab}. To this end, we consider an optimal unconstrained design (without any specific codebook), which turns out to be intractable, and propose a novel codebook-based design with optimized power allocation for joint BS-RIS beamforming.

\subsection{Optimal Unconstrained Design}\label{sec_uncons}
Solving the PEB minimization problem in \eqref{eq_problem_peb} requires the knowledge of precise UE location\footnote{From the viewpoint of joint BS-RIS beamforming, the most essential information required to solve \eqref{eq_problem_peb} is the UE location (i.e., where to steer the BS and RIS beams). Regarding the other unknown parameters in $\etab$ in \eqref{eq_etab}, we note from Appendix~\ref{sec_fim_func} that the FIM does not depend on \rev{a \textit{specific value} of} the clock offset $\Delta$ \rev{(though the FIM depends \textit{functionally} on $\Delta$, as seen from \eqref{eq_Jeta} and Appendix~\ref{sec_transform_mat})}. Hence, the PEB minimization problem in \eqref{eq_problem_peb} can be solved without the knowledge of $\Delta$. On the other hand, we assume the channel gains in \eqref{eq_etab} are perfectly known. As seen from Appendix~\ref{sec_fim_corr}, the case of uncertain gains leads to intractable PEB expressions due to \ac{LoS}-\ac{NLoS} correlations, is therefore left outside the scope of the current work and will be investigated in a future study.} $\pp$ which, however, may not be available in practice due to  measurement noise and tracking errors. Hence, we assume an uncertainty region $\pp \in \ppreg$ for the UE location and consider the robust design problem that minimizes the worst-case PEB over $\ppreg$ \cite{robust_2013,win_navigation_2018,beamspace_loc_2019_TSP,robust_VLP_TCOM_2018}:
\begin{align}  \label{eq_problem_peb_robust}
	\mathop{\mathrm{min}}\limits_{\{ \XX_g, \oomegag, \Psibbig_g  \}_{g=1}^{G}} &~ \mathop{\mathrm{max}}\limits_{\pp \in \ppreg} ~
	\fpeb\left(\{ \XX_g, \oomegag, \Psibbig_g  \}_{g=1}^{G}; \etab(\pp) \right)
	\\ \nonumber
    \mathrm{s.t.}&~~ \eqref{eq_power_cons}-\eqref{eq_psig_cons} ~,
\end{align}
where $\etab$ is replaced by $\etab(\pp)$ in the $\fpeb$ to highlight its dependency on $\pp$. The epigraph form of \eqref{eq_problem_peb_robust} can be expressed as%
\begin{subequations} \label{eq_problem_peb_robust2}
\begin{align} 
	\mathop{\mathrm{min}}\limits_{\{ \XX_g, \oomegag, \Psibbig_g  \}_{g=1}^{G},t} &~ t \\ \label{eq_problem_peb_robust2_cons1}
    \mathrm{s.t.}&~~ \fpeb\left(\{ \XX_g, \oomegag, \Psibbig_g  \}_{g=1}^{G}; \etab(\pp) \right) \leq t, ~ \forall \pp \in \ppreg  \\ \nonumber
    &~~ \eqref{eq_power_cons}-\eqref{eq_psig_cons} ~.
\end{align}
\end{subequations}
To tackle the semi-infinite optimization problem in \eqref{eq_problem_peb_robust2}, we can discretize $\ppreg$ into $\ngrid$ grid points $\{ \pp_m \}_{m=0}^{\ngrid-1}$ \cite{beamspace_loc_2019_TSP} and obtain the following approximated version using \eqref{eq_peb_fim}:%
\begin{subequations} \label{eq_peb_robust3}
\begin{align} 
	\mathop{\mathrm{min}}\limits_{ \substack{ \{ \XX_g, \oomegag, \Psibbig_g  \}_{g=1}^{G} \\ t,\{u_{m,k}\} } } &~ t \\ \label{eq_problem_peb_robust3_discrete_cons1}
    \mathrm{s.t.}&~~ 
    \begin{bmatrix}  \JJeta(\{\XX_g, \oomegag, \Psibbig_g  \}_{g=1}^{G}; \etab(\pp_m)) & \ekk{k} \\ \ekkt{k} & u_{m,k}  \end{bmatrix} \succeq 0
	\\ \label{eq_problem_peb_robust3_discrete_cons2} &~~ u_{m,0} +  u_{m,1} \leq t ~,
	\\ \nonumber
	&~~ k=0,1, ~ m = 0, \ldots, \ngrid-1 ~,
    \\ \nonumber
    &~~  \eqref{eq_power_cons}-\eqref{eq_psig_cons},
\end{align}
\end{subequations}
where $\ekk{k}$ is the $k$-th column of the identity matrix, and the equivalence between \eqref{eq_problem_peb_robust3_discrete_cons1}, \eqref{eq_problem_peb_robust3_discrete_cons2} and the discretized version of \eqref{eq_problem_peb_robust2_cons1} stems from \cite[Eq.~(7.28)]{boyd2004convex}. In \eqref{eq_problem_peb_robust3_discrete_cons1}, $\JJeta(\{\XX_g, \oomegag, \Psibbig_g  \}_{g=1}^{G}; \etab(\pp_m))$ is the FIM in \eqref{eq_Jeta} evaluated at the grid location $\pp_m$.

Two issues arise that make the problem  \eqref{eq_peb_robust3} intractable. First, \eqref{eq_power_cons}--\eqref{eq_psig_cons} involve non-convex rank and unit-modulus constraints, which can only be handled via relaxations in Sec.~\ref{sec_relaxed_peb}. Second, since $\JJeta(\{\XX_g, \oomegag, \Psibbig_g  \}_{g=1}^{G}; \etab(\pp_m))$ is not linear with respect to $\{\XX_g, \oomegag, \Psibbig_g  \}_{g=1}^{G}$ according to \rev{Remark}~\ref{lemma_fim_func}, \eqref{eq_problem_peb_robust3_discrete_cons1} does not represent a linear matrix inequality (LMI) \cite{fukuda2001branch}, implying that \eqref{eq_peb_robust3} is not convex \cite[Ex.~(2.10)]{boyd2004convex}. As a possible remedy, alternating optimization (AO) of $\{ \XX_g  \}_{g=1}^{G}$ and $\{ \Psibbig_g  \}_{g=1}^{G}$ can be performed (after eliminating the dependency of $\JJeta$ on $\{ \oomegag \}_{g=1}^{G}$ using the approximation in \eqref{eq_JJgamma_bdiag_approx}), where each subproblem becomes convex as bi-linear matrix inequalities (BMIs) degenerate to LMIs when one of the variables is fixed. However, this leads to a high computational complexity roughly given by $O(\nbs^6)$ and $O(\nr^6)$ \cite[Ch.~11]{nemirovski2004interior} for the BS and RIS subproblems, respectively. To devise a practically implementable solution, we propose a low-complexity codebook-based design strategy, as detailed in Sec.~\ref{sec_codebook}.

\subsection{Low-Complexity Codebook-Based Design}\label{sec_codebook}
Motivated by the optimal low-dimensional structure of the BS precoder and the RIS phase profile covariance matrices, derived in Prop.~\ref{prop_BS_precoder} and Prop.~\ref{prop_RIS_profile}, we develop a codebook-based low-complexity design approach as a practical alternative to unconstrained design in Sec.~\ref{sec_uncons}. To this end, let $\{ \thetabm^{(i)} \}_{i=1}^{\lb}$ and $\{ \thetarm^{(i)} \}_{i=1}^{\lr}$ denote the uniformly spaced \ac{AoD}s from the BS to the UE and from the RIS to the UE, respectively, that span the uncertainty region $\ppreg$ of the UE location, where the angular spacing is set to $3 \, \rm{dB}$ (half-power) beamwidth of the corresponding array \cite{zhang2018multibeam,keskin2021optimal}, \cite[Ch.~22.10]{orfanidis2002electromagnetic}. 

Relying on Prop.~\ref{prop_BS_precoder}, Prop.~\ref{prop_RIS_profile} and their interpretation in Sec.~\ref{sec_interp}, we propose the following codebooks for the BS precoder and the RIS phase profiles \cite{keskin2021optimal} consisting of both \textit{directional} and \textit{derivative} beams \rev{(please refer to Appendix~\ref{sec_supp_codebook} for additional details on how to obtain these codebooks)}:
\begin{subequations}\label{eq_prop_codebook}
\begin{align} \label{eq_fbs}
    \FFbs &= [ \aabs(\thetabr) ~ \FFbssum ~ \FFbsdiff ]^\conj \in \complexset{\nbs}{(2 \lb + 1)} ~, \\ \label{eq_fris}
    \FFris &= [ \FFrissum ~ \FFrisdiff ]^\conj \in \complexset{\nr}{2 \lr } ~,
\end{align}
\end{subequations}
where $ \FFbssum \eqdef [ \aabs(\thetabm^{(1)}) \, \cdots \, \aabs(\thetabm^{(\lb)}) ]$, $ \FFbsdiff \eqdef [ \aabsdt(\thetabm^{(1)}) \, \cdots \, \aabsdt(\thetabm^{(\lb)}) ]$ and 
\begin{align} \label{eq_fris_sum}
     \FFrissum &\eqdef [ \aarisw(\thetarm^{(1)}) \, \cdots \, \aarisw(\thetarm^{(\lr)}) ] ~, \\ \label{eq_fris_der}
     \FFrisdiff &\eqdef [ \aariswdttilde(\thetarm^{(1)}) \, \cdots \, \aariswdttilde(\thetarm^{(\lr)}) ] ~.
\end{align}
In \eqref{eq_fris_der}, due to phase-only control of RIS profiles, we employ $\aariswdttilde(\theta)$, which is the best approximation with unit-modulus entries to $\aariswdt(\theta)$ in \eqref{eq_bris}. To obtain $\aariswdttilde(\theta)$ from $\aariswdt(\theta)$, the projected gradient descent algorithm in \cite[Alg.~1]{analogBeamformerDesign_TSP_2017} is used.

For each transmission, we choose a BS-RIS signal pair $\{ \FFbs_{:,i}, \FFris_{:,j} \}$, corresponding to the $\thn{i}$ beam in $\FFbs$ and the $\thn{j}$ beam in $\FFris$, which leads to $G = (2 \lb + 1) 2 \lr$ transmissions in total\rev{\footnote{\rev{Due to the dependence of $\lb$ and $\lr$ on the $3 \, \rm{dB}$ beamwidth of the respective arrays at the BS and RIS, $G$ is a function of the number of elements at the BS and RIS as well as the size of the uncertainty region $\ppreg$. In addition, depending on whether the SNR is sufficient using a single slot of $G$ transmissions, the slot can be repeated multiple times to reach the desired level of SNR.}}}. To minimize the worst-case PEB using this codebook-based approach, we formulate a beam power allocation problem that finds the optimal power $\varrhob = \left[\varrho_1 \ldots \varrho_G \right]^T$ of BS beams in each transmission under total power constraint\footnote{Each beam in $\FFbs$ and $\FFris$ is normalized to have unit norm prior to power optimization.}:%
\begin{subequations} \label{eq_peb_codebook}
\begin{align} \label{eq_peb_codebook_obj}
	\mathop{\mathrm{min}}\limits_{ \substack{\varrhob,t \\\{u_{m,k}\}} } &~ t \\ \label{eq_peb_codebook_lmi}
    \mathrm{s.t.}&~~ 
    \begin{bmatrix}  \JJeta(\{\XX_g, \oomegag, \Psibbig_g  \}_{g=1}^{G}; \etab(\pp_m)) & \ekk{k} \\ \ekkt{k} & u_{m,k}  \end{bmatrix} \succeq 0 ~, 
	\\ \nonumber &~~ u_{m,0} +  u_{m,1} \leq t,  ~ k=0,1, ~m=0,\ldots,\ngrid-1 ~,
    \\ \nonumber &~~  \tracebig{\sum_{g=1}^{G}\XX_g} = 1 \, , \, \varrhob \succeq \boldzero \, , \, \XX_g = \varrho_g \FFbs_{:,i} (\FFbs_{:,i})^\hermit ~,
    \\ \nonumber &~~ \oomegag = \FFris_{:,j} \, , \, \Psibbig_g = \oomegag (\oomegag)^\hermit  \,, \, g=1,\ldots,G ~,
\end{align}
\end{subequations}
where the mapping between the transmission index $g$ and the BS-RIS beam index pair $(i,j)$ is performed according to $g = i + (2\lb+1)(j-1)$ for $i = 1, \ldots, 2\lb+1$ and $j = 1, \ldots, 2 \lr$. As \eqref{eq_peb_codebook_lmi} is LMI in $\varrhob$ and $\{u_{m,k}\}$ (see \rev{Remark}~\ref{lemma_fim_func}), the problem \eqref{eq_peb_codebook} is convex. After obtaining the optimal power allocation vector $\varrhob^{\star} = \left[\varrho_1^{\star} \ldots \varrho_G^{\star} \right]^T$ as the solution to \eqref{eq_peb_codebook}, the optimized codebook is given by the collection of the BS-RIS signal pairs $\Big\{ \sqrt{\varrho_g^{\star}}  \FFbs_{:,i}, \FFris_{:,j} \Big\}_{\forall i,j}$. The overall BS-RIS signal design algorithm is summarized in Algorithm~\ref{alg_codebook}. The computational complexity of \eqref{eq_peb_codebook} is approximately given by \rev{$O(M^3)$} \cite[Ch.~11]{nemirovski2004interior}, \cite{keskin2021optimal}, \rev{under the assumption that $M$ is on the same order as $G$}. Since \rev{$M < \nbs^2$} and \rev{$M < \nr^2$} in practice (see Sec.~\ref{sec_sim_setup}), the proposed (non-iterative) design strategy in Algorithm~\ref{alg_codebook} is more efficient than \rev{even the} individual iterations of an AO approach in Sec.~\ref{sec_uncons}.

{\ed As anticipated, the proposed robust joint design of BS precoders and RIS phase profiles can be in principle extended to the 3D case, using a 2D array (e.g., a URA) in place of the ULA. In this case, three types of beams need to be employed, namely, \textit{directional} beams, \textit{azimuth derivative} beams and \textit{elevation derivative} beams, in contrast to only directional and derivative beams as in the 2D scenario.}

\begin{algorithm}[t]
	\caption{Joint BS Precoder and RIS Phase Profile Design with Power Optimized Codebooks}
	\label{alg_codebook}
	\begin{algorithmic}[1]
	    \State \textbf{Input:} Uncertainty region $\ppreg$ of the UE location $\pp$ in \eqref{eq_etab}.
	    \State \textbf{Output:} Optimized BS-RIS signal pairs $\big\{ \sqrt{\varrho_g^{\star}}  \FFbs_{:,i}, \FFris_{:,j} \big\}_{\forall i,j}$ with the optimal powers $\{ \varrho_g^{\star} \}_{\forall g}$.
	    \begin{enumerate}[label=(\alph*)]
	        \item Determine the uniformly spaced \ac{AoD}s from the BS to the UE $\{ \thetabm^{(i)} \}_{i=1}^{\lb}$ and those from the RIS to the UE $\{ \thetarm^{(i)} \}_{i=1}^{\lr}$ based on $\ppreg$.
	        \item Construct the BS and RIS codebooks in \eqref{eq_prop_codebook}.
	        \item Perform power allocation across $G = (2 \lb + 1) 2 \lr$ transmissions, each employing a different BS-RIS signal pair $\{ \FFbs_{:,i}, \FFris_{:,j} \}$, by solving the problem in \eqref{eq_peb_codebook}.
	    \end{enumerate}
	\end{algorithmic}
	\normalsize
\end{algorithm}

\section{Maximum Likelihood Joint Localization and Synchronization}\label{sec_estimator}
In this section, we first derive the joint ML estimator of the desired position $\bm{p}$ and clock offset $\Delta$. To overcome the need of an exhaustive 3D grid-based optimization of the resulting compressed log-likelihood function, we then provide a reduced-complexity estimator that leverage a suitable reparameterization of the signal model to decouple the dependencies on the delays and \ac{AoD}s, enabling a separate though accurate initial estimation of both $\bm{p}$ and $\Delta$. Such estimated values are subsequently used as initialization for an iterative low-complexity optimization of the joint ML cost function, which provides the  refined position and clock offset estimates.

\subsection{Joint Position and Clock Offset Maximum Likelihood Estimation}\label{sec:JML}
To formulate the joint ML estimation problem, let $\bm{\Theta} = [p_x \ p_y \ \Delta]^\mathsf{T}$ denote the vector containing the desired UE position and clock offset parameters. By parameterizing the unknown \ac{AoD}s ($\theta_\BM$ and $\theta_\RM$) and delays ($\tau_\BM$ and $\tau_\R$) as a function of the sought $\bm{\Theta}$ through  \eqref{geomrelationships}, and stacking all the $N$ signals received over each transmission $g$, we obtain the more compact expression 
 \begin{equation}\label{eq::Model1}
    \bm{y}_g = \sqrt{P}\bm{B}_g \bm{\alpha} +\bm{\nu}_g
    \end{equation}
   with 
   \begin{align*}
  \bm{y}_g &= [y_g[0] \ \cdots \ y_g[N-1]]^\mathsf{T}\\ \bm{\alpha} &=[\alpha_\BM \ \alpha_\R]^\mathsf{T}\\ \bm{B}_g &= [
(\tilde{\bm{S}}_\BM^g)^\mathsf{T}\bm{a}_\BS(\theta_\BM), \;\;  (\tilde{\bm{S}}_\R^g)^\mathsf{T}\bm{A}^\mathsf{T}(\bm{\Omega}^g)^\mathsf{T}\bm{a}_\RIS(\theta_\RM)]\\  
\tilde{\bm{S}}^g_\BM &= [\bm{s}_g[0] \;\; \cdots \;\;\e^{-j\kappa_{N-1} \tau_{\BM}}\bm{s}_g[N-1]]
\end{align*}
where $\bm{s}_g[n] = \bm{f}_g s_g[n]$, $\tilde{\bm{S}}_\R^g$ is defined as $\tilde{\bm{S}}_\BM^g$ but with $\tau_\R$ in place of $\tau_\BM$, $\bm{A} = \bm{a}_{{\text{\tiny{RIS}}}}(\phi_\text{\tiny{B,R}})\bm{a}^\mathsf{T}_{{\text{\tiny{BS}}}}(\theta_\text{\tiny{B,R}})$, and $\alpha_\R = \alpha_\BR \alpha_\RM$. Without loss of generality, we assume that $\sigma^2$ is already known (its estimate can be straightforwardly obtained as $\hat{\sigma}^2 = \sum_{g=1}^G \|\bm{y}_{g} - \sqrt{ P} \bm{B}_{g}\bm{\alpha}\|^2 / (NG)$ once the rest of parameters have been estimated), so leaving $\bm{\alpha}$ as the sole vector of unknown nuisance parameters. Following the ML criterion, the estimation problem can be thus formulated as
    \begin{equation}
    \hat{\bm{\Theta}}^\ML = \arg \min_{\bm{\Theta}} \left[ \min_{\bm{\alpha}} L(\bm{\Theta},\bm{\alpha}) \right]  
    \end{equation}
where 
\begin{equation}\label{original_likelihood}
    L(\bm{\Theta},\bm{\alpha}) = \sum_{g=1}^G \| \bm{y}_g -\sqrt{P}\bm{B}_g \bm{\alpha} \|^2
\end{equation}
represents the likelihood function. It is not difficult to show that the value of the complex vector $\bm{\alpha} \in \mathbb{C}^{2 \times 1}$ minimizing \eqref{original_likelihood} is given by
$
\hat{\bm{\alpha}}^\ML = \frac{1}{\sqrt{P}} \bm{B}^{-1} \sum_{g=1}^G \bm{B}_g^\mathsf{H}\bm{y}_g
$
where $\bm{B} = \sum_{g=1}^G \bm{B}_g^\mathsf{H}\bm{B}_g$. Substituting $\hat{\bm{\alpha}}^\ML$ back into the likelihood function \eqref{original_likelihood} leads to%
\begin{equation}
    L(\bm{\Theta}) = \sum_{g=1}^G \| \bm{y}_g -\sqrt{P}\bm{B}_g(\bm{\Theta}) \hat{\bm{\alpha}}^\ML(\bm{\Theta}) \|^2
\end{equation}
where we explicitly highlighted the remaining dependency on the sole desired parameter vector $\bm{\Theta}$. Accordingly, the final joint ML (JML) estimator of UE position and clock offset is 
\begin{equation}\label{eq::MLcost}
\hat{\bm{\Theta}}^\ML = \arg \min_{\bm{\Theta}}     L(\bm{\Theta}).
\end{equation}
Unfortunately, $\hat{\bm{\Theta}}^\ML$ cannot be effortlessly retrieved being $L(\bm{\Theta})$ a highly non-linear function with multiple potential local minima. A more practical solution consists in finding a good initial estimate of $\bm{\Theta}$ and use it to compute $\hat{\bm{\Theta}}^\ML$ by means of a low-complexity iterative optimization. {\ed The latter consists in adopting a numerical optimization approach such as the Nelder-Mead algorithm to iteratively optimize the JML cost function in \eqref{eq::MLcost} starting from a more accurate initial estimate $\hat{\bm{\Theta}}$. As well-known, the Nelder-Mead procedure does not require any derivative information, which makes it suitable for problems with non-smooth functions like \eqref{eq::MLcost}, and is recognized to be extremely fast to converge (in all our trials, the number of required iterations was always less than 30).} A direct way to obtain such initialization is to perform an exhaustive grid search over the 3D space of the unknown $\bm{p}$ and $\Delta$. To overcome the burden of a full-dimensional optimization, in the next section we present a relaxed ML estimator of the position and the clock offset, able to provide a good initialization for the iterative optimization of \eqref{eq::MLcost}, but at a considerably lower computational complexity.
\subsection{Proposed Reduced-Complexity Estimator}\label{sec::RML_section}
\subsubsection{Relaxed Maximum Likelihood Position Estimation} We start by stacking all the observations collected over the $G$ transmissions and by further manipulating the resulting model, obtaining the new expression
    \begin{equation}\label{relaxedmodel}
        \underbrace{\begin{bmatrix}
        \bm{y}_1 \\
        \vdots \\
        \bm{y}_G
        \end{bmatrix}}_{\bm{y} \in \mathbb{C}^{GN\times 1}} = \underbrace{\begin{bmatrix}
        \bm{\Phi}^1_{\text{\tiny B,U}}(\theta_{\text{\tiny B,U}}(\bm{p})) & \bm{\Phi}^1_{\text{\tiny R,U}}(\theta_{\text{\tiny R,U}}(\bm{p})) \\
        \vdots & \vdots \\
        \bm{\Phi}^G_{\text{\tiny B,U}}(\theta_{\text{\tiny B,U}}(\bm{p})) & \bm{\Phi}^G_{\text{\tiny R,U}}(\theta_{\text{\tiny R,U}}(\bm{p}))
        \end{bmatrix}}_{\bm{\Phi}(\theta_\BM(\bm{p}),\theta_\RM(\bm{p})) \eqdef \bm{\Phi}(\bm{p}) \in \mathbb{C}^{GN \times 2N}}
        \underbrace{\begin{bmatrix}
        \bm{e}_{\text{\tiny B,U}} \\
        \bm{e}_{\text{\tiny R}}
        \end{bmatrix}}_{\bm{e}\in \mathbb{C}^{2N \times 1}} + \begin{bmatrix}
        \bm{\nu}_1 \\
        \vdots \\
        \bm{\nu}_G
        \end{bmatrix}
    \end{equation}
where $\bm{\Phi}^g_{\text{\tiny B,U}}(\theta_{\text{\tiny B,U}}(\bm{p})) = \mathrm{diag}(\bm{a}^\mathsf{T}_{{\text{\tiny{BS}}}}(\theta_\text{\tiny{B,U}}(\bm{p}))\bm{S}^g)$, $\bm{\Phi}^g_{\text{\tiny R,U}}(\theta_{\text{\tiny R,U}}(\bm{p})) = \mathrm{diag}(\bm{a}^\mathsf{T}_{{\text{\tiny{RIS}}}}(\theta_\text{\tiny{R,U}}(\bm{p}))\bm{\Omega}^g\bm{A}\bm{S}^g)$, $\bm{S}^g = [\bm{s}_g[0] \ \cdots \ \bm{s}_g[N-1]]$, $g=1,\ldots,G$, and
\begin{equation}\label{eq:vectore}
\bm{e}_{\text{\tiny B,U}} =  \sqrt{P}\alpha_{\text{\tiny B,U}}  \begin{bmatrix}
1 \\
\e^{-j \kappa_1 \tau_\text{\tiny{B,U}}} \\
\vdots \\
\e^{-j\kappa_{N-1}\tau_\text{\tiny{B,U}} }
\end{bmatrix}, \ \bm{e}_{\text{\tiny R}} = \sqrt{P}\alpha_{\text{\tiny R}}\begin{bmatrix}
1 \\
\e^{-j \kappa_1\tau_\text{\tiny{R}}} \\
\vdots \\
\e^{-j\kappa_{N-1}\tau_\text{\tiny{R}}}
\end{bmatrix}.
\end{equation}
We now observe that \eqref{relaxedmodel} allows us to decouple the dependencies on the delays and \ac{AoD}s in \eqref{eq::Model1}, with the new matrix $\bm{\Phi}$ that depends only on the desired $\bm{p}$ through the geometric relationships with the corresponding \ac{AoD}s $\theta_{\text{\tiny B,U}}(\bm{p})$ and $\theta_{\text{\tiny R,U}}(\bm{p})$. By relaxing the dependency of $\bm{e}$ on the delays $\tau_\BM$ and $\tau_\R$, and considering it as a generic unstructured $2N$-dimensional vector, a relaxed ML-based estimator (RML) of $\bm{p}$ can be derived as
\begin{equation}\label{eq::unstructML}
\hat{\bm{p}}^\RML = \arg \min_{\bm{p}} \left[ \min_{\bm{e}} \| \bm{y} - \bm{\Phi}(\bm{p})\bm{e} \|^2\right].   
\end{equation}
The
inner minimization of \eqref{eq::unstructML} can be more easily solved by decomposing it over the different $N$ subcarriers as
\begin{equation}\label{new_lowcost}
 \min_{\bm{e}} \| \bm{y} - \bm{\Phi}(\bm{p})\bm{e} \|^2 =  \min_{\bm{e}_0,\ldots,\bm{e}_{N-1}}  \sum_{n=0}^{N-1} \|\bm{y}_n - \bm{\Phi}_n \bm{e}_n\|^2
\end{equation}
where we exploited the peculiar structure of $\bm{\Phi}(\bm{p})$, which consists of blocks of $N \times N$ diagonal matrices, with $\bm{y}_n = [y_1[n] \ \cdots \ y_G[n]]^\mathsf{T}$, 
\begin{equation}
\bm{\Phi}_n(\bm{p}) = \begin{bmatrix}
    \phi^1_{\text{\tiny B,U},n}(\bm{p}) & \phi^1_{\text{\tiny R,U},n}(\bm{p}) \\
    \vdots & \vdots \\
    \phi^G_{\text{\tiny B,U},n}(\bm{p}) & \phi^G_{\text{\tiny R,U},n}(\bm{p})
\end{bmatrix} \in \mathbb{C}^{G \times 2}
\end{equation}
$\bm{e}_n = [\bm{e}_\BM[n] \ \bm{e}_\text{\tiny R}[n]]^\mathsf{T} \in \mathbb{C}^{2 \times 1}$, $\phi^g_{\text{\tiny B,U},n}(\bm{p}) = \bm{a}^\mathsf{T}_{{\text{\tiny{BS}}}}(\bm{p}) \bm{s}_g[n]$ and $\phi^g_{\text{\tiny R,U},n}(\bm{p}) = \bm{a}^\mathsf{T}_{{\text{\tiny{RIS}}}}(\bm{p})\bm{\Omega}^g\bm{A}\bm{s}_g[n]$, for $n=0,\ldots,N-1$, $g=1,\ldots,G$. Each unknown vector $\bm{e}_n$ minimizing \eqref{new_lowcost} can be separately obtained as
\begin{equation}\label{eq::estvecte}
\hat{\bm{e}}^\RML_n(\bm{p}) = (\bm{\Phi}_n^\mathsf{H}(\bm{p})\bm{\Phi}_n(\bm{p}))^{-1} \bm{\Phi}_n^\mathsf{H}(\bm{p}) \bm{y}_n
\end{equation}
that is, each $\bm{e}_n$ is estimated by pseudo-inverting the corresponding matrix $\bm{\Phi}_n(\bm{p})$. 
The inverse in \eqref{eq::estvecte} can be  computed in closed-form 
\begin{equation}
    (\bm{\Phi}_n^\mathsf{H}\bm{\Phi}_n)^{-1} = \frac{1}{u_n z_n - v_n w_n}\begin{bmatrix}
        z_n & -v_n \\
        -w_n & u_n
    \end{bmatrix}
\end{equation}
where $u_n = \sum_{g=1}^G \left|\phi^g_{\text{\tiny B,U},n} \right|^2$,  $ v_n = \sum_{g=1}^G (\phi^g_{\text{\tiny B,U},n})^*\phi^g_{\text{\tiny R,U},n}$, $w_n  = \sum_{g=1}^G (\phi^g_{\text{\tiny R,U},n})^*\phi^g_{\text{\tiny B,U},n}$, and $ z_n = \sum_{g=1}^G \left|\phi^g_{\text{\tiny R,U},n}\right|^2$, 
and we omitted the dependency on $\bm{p}$ for brevity.
Accordingly, the RML estimator can be more conveniently obtained as
\begin{equation}\label{eq::final_RMLest}
\hat{\bm{p}}^\RML = \arg \min_{\bm{p}} \sum_{n=0}^{N-1}\| \bm{y}_n - \bm{l}_n(\bm{p}) \|^2 
\end{equation}
with the elements of the vector $\bm{l}_n(\bm{p})$ given by
\begin{align}
    \!\!\!l_g[n](\bm{p}) = &\frac{1}{u_n z_n - v_n w_n}\Big[(\phi^g_{\text{\tiny B,U},n} z_n - \phi^g_{\text{\tiny R,U},n} w_n)\sum_{\ell=1}^G (\phi^\ell_{\text{\tiny B,U},n})^* y_{\ell}[n] \nonumber \\
    &+ (\phi^g_{\text{\tiny R,U},n} u_n - \phi^g_{\text{\tiny B,U},n} v_n)\sum_{\ell=1}^G (\phi^\ell_{\text{\tiny R,U},n})^* y_{\ell}[n]\Big].
\end{align}
{\ed A 2D grid search is then performed on the RML cost function provided in \eqref{eq::final_RMLest} to obtain the initial UE position estimate $\hat{\bm{p}}^\text{\tiny RML}$, which will be used together with the clock offset estimate obtained in the next section as initial point to iteratively optimize the 3D plain JML cost function given in \eqref{eq::MLcost}.}

\subsubsection{FFT-based Clock Offset Estimation} As a byproduct of the above estimation of $\bm{p}$, it is possible to derive an efficient estimator of the unknown delays $\tau_\BM$ and $\tau_\R$, which in turn will be used to retrieve a closed-form estimate of the sought $\Delta$. Specifically, we first plug  $\hat{\bm{p}}^\RML$ back in \eqref{eq::estvecte} to obtain an estimate of the vectors $\bm{e}_n$ $n=0,\ldots,N-1$. The elements of the estimated vectors $\hat{\bm{e}}^\RML_n$ can be then merged according to \eqref{eq:vectore} to obtain an estimate of the two vectors $\hat{\bm{e}}_\BM(\hat{\bm{p}}^\RML)$ and $\hat{\bm{e}}_\R(\hat{\bm{p}}^\RML)$, respectively. The key observation consists in the fact that the elements of both $\hat{\bm{e}}_\BM(\hat{\bm{p}}^\RML)$ and $\hat{\bm{e}}_\R(\hat{\bm{p}}^\RML)$ can be interpreted as discrete samples of  complex exponentials having normalized frequencies $\nu_\BM = -\frac{\tau_\BM}{NT_s}$ and $\nu_\R =-\frac{\tau_\R}{NT_s}$, respectively. This allows to estimate the delays $\tau_\BM$ and $\tau_\R$ by searching for the dominant peaks in the FFT of the corresponding vectors $\hat{\bm{e}}_\BM(\hat{\bm{p}}^\RML)$ and $\hat{\bm{e}}_\R(\hat{\bm{p}}^\RML)$. By defining $\bm{f}_h(\hat{\bm{p}}^\RML) = \mathrm{FFT}(\hat{\bm{e}}_h(\hat{\bm{p}}^\RML))$ as the FFT of the vector $\hat{\bm{e}}_h(\hat{\bm{p}}^\RML)$ (with either $h = B,U$ or $h = R$) computed on $N_{\text{F}}$ points, we first seek for the index corresponding to the maximum element in $\bm{f}_h(\hat{\bm{p}}^\RML)$
\begin{equation}\label{eq::FFTmax}
\hat{k}_h(\hat{\bm{p}}^\RML) = \arg \max_{k} \left[|f_h(\hat{\bm{p}}^\RML)[k]|: 0 \leq k \leq N_F-1\right]
\end{equation}
with $|f_h(\hat{\bm{p}}^\RML)[k]|$ denoting the absolute value of the $k$-th element of $\bm{f}_h(\hat{\bm{p}}^\RML)$. Since the first $N_F/2 +1$ elements correspond to positive values of the normalized frequency $\nu_o \in [0,1/2]$, while the remaining $N_F/2 - 1$ are associated to the negative part of the spectrum, i.e., $\nu_h \in (-1/2,0)$, the estimate of the delays can be obtained by mapping the corresponding $\hat{k}_h(\hat{\bm{p}}^\RML)$ as
\begin{equation}\label{eq::map_index}
\hat{\tau}^\FFT_h = \begin{cases}
-\frac{\hat{k}_h}{N_F}NT_S &\mbox{if } 0 \leq \hat{k}_h \leq N_F/2 \\
(1/2 - \frac{\hat{k}_h}{N_F})NT_S &\mbox{if } N_F/2+1 \leq \hat{k}_h \leq N_F-1
\end{cases}
\end{equation}
where we omitted the dependency on $\hat{\bm{p}}^\RML$ for conciseness. Once the two delays have been estimated, the sought clock offset $\Delta$ can be obtained in closed-form as
\begin{equation}\label{eq::Delta_RML}
\hat{\Delta}^\FFT= \frac{1}{2} \Big[\hat{\tau}^\FFT_\BM - \|\hat{\bm{p}}^\RML\|/c + \hat{\tau}^\FFT_\R - (\|\bm{r}\| + \| \bm{r}-\hat{\bm{p}}^\RML \|)/c \Big].    
\end{equation}
The obtained estimate $\hat{\bm{\theta}}^\RML = [\hat{\bm{p}}^\RML \ \hat{\Delta}^\FFT]^\mathsf{T}$ is then used to initialize an iterative optimization procedure (e.g., Nelder-Mead) to efficiently solve the JML estimation problem in \eqref{eq::MLcost}. The main steps of the proposed reduced-complexity estimation algorithm are summarized in Algorithm \ref{alg_estimator}.

{\ed It is worth noting that also the proposed joint localization and estimation algorithm can be  extended to the 3D case, in which also elevation angles are considered. In fact, the properties used to obtain the relaxation of the ML cost function and to estimate the delays via FFT are fulfilled not only by ULAs but also by uniform rectangular arrays (URAs). The final position estimation in the RML approach would be then performed on a 3D grid instead of a 2D one. The computational complexity of the procedure, of course, would be higher as in any higher-dimensional problem, but no additional theoretical issues arise.}

\begin{algorithm}[t]
	\caption{Low-Complexity Joint Localization and Synchronization Algorithm}
	\label{alg_estimator}
	\begin{algorithmic}[1]
	    \State \textbf{Input:} Received signals $\left\{y^g[n]\right\}_{\forall n,g}$, optimized BS-RIS precoders $\big\{ \sqrt{\varrho_g^{\star}}  \FFbs_{:,i}, \FFris_{:,j} \big\}_{\forall i,j}$.
	    \State \textbf{Output:} UE position  $\hat{\bm{p}}^\ML$ and clock offset $\hat{\Delta}^\ML$.
	    \begin{enumerate}[label=(\alph*)]
	        \item Perform a coarse 2D search to obtain an initial estimate $\hat{\bm{p}}^\RML$ via RML in \eqref{eq::final_RMLest}.
	        \item Use $\hat{\bm{p}}^\RML$ to reconstruct the two vectors $\hat{\bm{e}}_\BM(\hat{\bm{p}}^\RML)$ and $\hat{\bm{e}}_\R(\hat{\bm{p}}^\RML)$ based on \eqref{eq::estvecte} and \eqref{eq:vectore}.
	        \item Search for the dominant peaks in the FFT-transformed vectors ($\bm{f}_\BM(\hat{\bm{p}}^\RML),\bm{f}_\R(\hat{\bm{p}}^\RML))$ and compute the corresponding delays estimates $(\hat{\tau}^\FFT_\BM,\hat{\tau}^\FFT_\R$).
	        \item Compute the initial estimate $\hat{\Delta}^\FFT$ using \eqref{eq::Delta_RML}.
	        \item Use $\hat{\bm{\theta}}^\RML = [\hat{\bm{p}}^\RML \ \hat{\Delta}^\FFT]^\mathsf{T}$ as initialization to iteratively solve the JML in \eqref{eq::MLcost} and obtain the final estimates $\hat{\bm{p}}^\ML$ and $\hat{\Delta}^\ML$.
	    \end{enumerate}
	\end{algorithmic}
	\normalsize
\end{algorithm}

{\ed 
\subsection{Complexity Analysis}
In this section, we analyze the computational complexity of the joint localization and synchronization algorithm proposed in Sec.~\ref{sec::RML_section}, also in comparison to the plain 3D JML estimator derived in Sec.~\ref{sec:JML}. Asymptotically speaking, we observe that the complexity in performing the 3D optimization required by the plain JML estimator in \eqref{eq::MLcost} is on the order of $O(Q^3GN N_\text{\tiny E})$, where $Q$ denotes the number of evaluation points per dimension (either $p_x$ coordinate, $p_y$ coordinate of the UE position, or clock offset $\Delta$), assumed to be the same for all the three dimensions for the sake of exposition, and $N_\text{\tiny E} = N_\text{\tiny BS}+N_\text{\tiny BS}N_\text{\tiny RIS} + N^2_\text{\tiny RIS}$ a term related to the number of elements at both BS and RIS. On the other hand, by analyzing the different steps involved in the proposed joint localization and synchronization algorithm (Algorithm \ref{alg_estimator}), it emerges that the overall  complexity is given by the sum of three terms 
\begin{equation}\label{eq::complexity}
O(Q^2GN N_\text{\tiny E}) + O(N_F \log N_F) + O(N_\text{\tiny I}GN N_\text{\tiny E}).
\end{equation}
The first term $O(Q^2GN N_\text{\tiny E})$ corresponds to the two-dimensional optimization required to obtain the initial UE position estimate $\hat{\bm{p}}^\text{\tiny RML}$ according to \eqref{eq::final_RMLest}. The second term $O(N_F \log N_F)$ denotes the complexity required to compute the FFT of the two vectors $\hat{\bm{e}}_\BM(\hat{\bm{p}}^\RML)$ and $\hat{\bm{e}}_\R(\hat{\bm{p}}^\RML)$ and to search for the dominant peaks yielding the clock offset estimate $\hat{\Delta}^\text{\tiny FFT}$ in \eqref{eq::Delta_RML}. The third and last term represents instead the complexity required by the Nelder-Mead procedure to iteratively optimize the JML cost function starting from the initial point $\hat{\bm{\Theta}}^\text{\tiny RML}$, with $N_\text{\tiny I}$ denoting the number of total iterations. This contribution is practically negligible compared to the first term in \eqref{eq::complexity} being the Nelder-Mead procedure extremely fast and typically converging in a few iterations (in all our trials, the number of iterations was always $N_\text{\tiny I} < 30$).

Considering that the minimum number of points required to compute the FFT is equal to the length of the involved vectors, i.e. $N_F \geq N$ (in our simulations we set $N_F = 512$) and that the FFT step is performed just once, it is apparent that the overall complexity is practically dominated by the first term $O(Q^2GN N_\text{\tiny E})$, namely by the two-dimensional search required  to obtain an initial estimate of $\bm{p}$. In this respect, the proposed joint localization and synchronization algorithm is able to reduce the complexity required by the plain JML estimator, which is cubic in $Q$, to a quadratic cost in $Q$ plus two very low-cost subsequent estimation steps (FFT and iterative optimization).
}
\section{Simulation Analysis and  Results}\label{sec_sim_res}
In this section, we conduct a numerical analysis to assess the performance of the low-complexity localization and synchronization algorithm presented in Sec.~\ref{sec_estimator}, when fed with the robust strategy for joint design of BS precoding and RIS phase profiles proposed in Sec.~\ref{sec_robust_design}.
The performance of the proposed approach is compared with the theoretical lower bounds derived in Sec.~\ref{sec_fim}, as well as against other state-of-the-art strategies for the design of BS and RIS precoding matrices, under different values of the main system parameters. We consider the 
root mean squared error (RMSE) as performance metric, estimated on  1000 independent Monte Carlo trials.

\subsection{Simulation Setup}\label{sec_sim_setup}
The analyzed scenario consists of a single BS placed at known position $\bm{q} = [0 \ 0]^\mathsf{T}$ m, a RIS placed at $\bm{r} = [12 \ 7]^\mathsf{T}$ m, and a UE with unknown location $\bm{p} = [5 \ 5]^\mathsf{T}$ m. The numerical evaluations are conducted assuming the transmission of $G = (2 \lb + 1) 2 \lr$ OFDM pilot signals in DL over a typical mmWave carrier frequency $f_c = 28$ GHz with bandwidth $B = 100$ MHz, along $N$ subcarriers equally spaced in frequency by $\Delta f = 240$ kHz. The BS is equipped with $N_\BS = 16$ antennas, while the RIS has $N_\R = 32$ elements. The channel amplitudes are generated according to the common path loss model in free space, i.e., $\rho_\BR = \lambda_c/(4\pi\|\bm{r}\|)$,  $\rho_\BM = \lambda_c/(4\pi\|\bm{p}\|)$, and $\rho_\RM = \lambda_c/(4\pi\|\bm{p}-\bm{r}\|)$, respectively, while the phases   $\varphi_\BM$ and $\varphi_\R$ are assumed to be uniformly distributed over $[-\pi,\pi]$. We set the clock offset to $\Delta =\frac{1}{8} \cdot NT_s$, while the transmitted power $P$ is varied in order to obtain different ranges of the received SNR over the \ac{LoS} path, defined as $\text{SNR} = 10 \log_{10} (P\rho^2_\BM/(N_0B))$, where $N_0$ is the noise power spectral density and $\sigma^2 = N_0B$. In the following, we consider an uncertainty region $\ppreg$ for the UE position having an extent of 3 m along both $x$ and $y$ directions. For this setup, using the typical $3 \, \rm{dB}$ beamwidth angular spacing of an ULA (about $1.8/N_\BS)$ leads to $\lb = 7$ and $\lr = 6$, which in turn correspond to $G = 180$ OFDM symbols. In the Appendices, we report additional performance analyses also for the case in which the uncertainty is increased to 5 m. \rev{The number of discrete UE positions $\left\{\mathbf{p}_m\right\}_{m=1}^M$ used to solve \eqref{eq_peb_codebook} is set to $M=9$. For more details on the setting of $M$, we refer the reader to Appendix~\ref{sec::choiceM}.}

\subsection{Benchmark Precoding Schemes}\label{subs:benchmarks}
To benchmark the proposed joint BS-RIS signal design algorithm proposed in Algorithm~\ref{alg_codebook}, we consider the following state-of-the-art schemes.
\subsubsection*{Directional Codebook (Uniform)}
    This  scheme considers only directional beams in the codebook and uses uniform (equal) power allocation among them, i.e, the BS does not implement the optimal power allocation provided in \eqref{eq_peb_codebook}. To guarantee a fair comparison, we double the angular sampling rate of the uncertainty region of the UE, so obtaining the same number of transmissions $G$ used by proposed codebooks in \eqref{eq_prop_codebook}. This leads to a set of \ac{AoD}s from the BS to the UE $\{ \thetabmtilde^{(i)} \}_{i=1}^{2\lb}$ and of \ac{AoD}s from the RIS to the UE $\{ \thetarmtilde^{(i)} \}_{i=1}^{2\lr}$. Accordingly, we consider the following directional codebooks for the BS and RIS transmissions:
    \begin{subequations}\label{eq_directional_codebook_uniform}
        \begin{align} \label{eq_fbs_dir}
    \FFbs &= [ \aabs(\thetabr) ~ \FFbssumdouble  ]^\conj \in \complexset{\nbs}{(2 \lb + 1)} ~, \\ \label{eq_fris_dir}
    \FFris &=  [ \aarisw(\thetarmtilde^{(1)}) \, \ldots \, \aarisw(\thetarmtilde^{(2\lr)}) ]^\conj  \in \complexset{\nr}{2     \lr } ~,
        \end{align}
    \end{subequations}
    where $\FFbssumdouble \eqdef [ \aabs(\thetabmtilde^{(1)}) \, \ldots \, \aabs(\thetabmtilde^{(2\lb)}) ].$

\subsubsection*{Directional Codebook (Optimized)}

This scheme uses the same directional codebook in \eqref{eq_directional_codebook_uniform} and performs the optimal power allocation for the BS beams in \eqref{eq_fbs_dir} according to \eqref{eq_peb_codebook}.

\subsubsection*{DFT Codebook (Optimized)} 

Let $\GG^N \in \complexset{N}{N}$ denote a DFT matrix. In addition, denote by $\thetabmc$ and $\thetarmc$, respectively, the \ac{AoD} from BS to UE and the \ac{AoD} from RIS to UE, corresponding to the center of the two \ac{AoD}s uncertainty regions computed from $\ppreg$. This scheme selects the columns from the corresponding DFT matrices that are closest to the center \ac{AoD}s as:
        \begin{subequations}\label{eq_dft_codebook_}
        \begin{align} 
    \FFbsdft &= \GG^{\nbs}_{:, (\ellbm-\lb+1):(\ellbm+\lb)} \in \complexset{\nbs}{2 \lb} ~, \\ 
     \FFrisdft &= \GG^{\nr}_{:, (\ellrm-\lr+1):(\ellrm+\lr)} \in \complexset{\nr}{2 \lr} ~, 
        \end{align}
    \end{subequations}
    where $ \ellbm \eqdef \arg \min_{\ell} \Vert \GG^{\nbs}_{:, \ell} - \aabs^\conj(\thetabmc) \Vert$ and $ \ellrm \eqdef \arg \min_{\ell} \Vert \GG^{\nr}_{:, \ell} - \aarisw^\conj(\thetarmc) \Vert$.
    Based on \eqref{eq_dft_codebook_}, the DFT codebooks \cite{mmWave_codebook_2017,egc_mimo_2003,tasos_precoding2020} for the BS and RIS transmission can be expressed as follows:
        \begin{subequations}\label{eq_dft_codebook}
        \begin{align} \label{eq_fbs_dir_dft}
    \FFbs &= \left[ \aabs^\conj(\thetabr) ~ \FFbsdft  \right] \in \complexset{\nbs}{(2 \lb + 1)} ~, \\ \label{eq_fris_dir_dft}
    \FFris &=  \FFrisdft  \in \complexset{\nr}{2     \lr } ~.
        \end{align}
    \end{subequations}
   Also for this scheme, we perform power allocation for the BS-RIS beams in \eqref{eq_dft_codebook} using \eqref{eq_peb_codebook}.

\begin{figure}
 \includegraphics[width=0.5\textwidth]{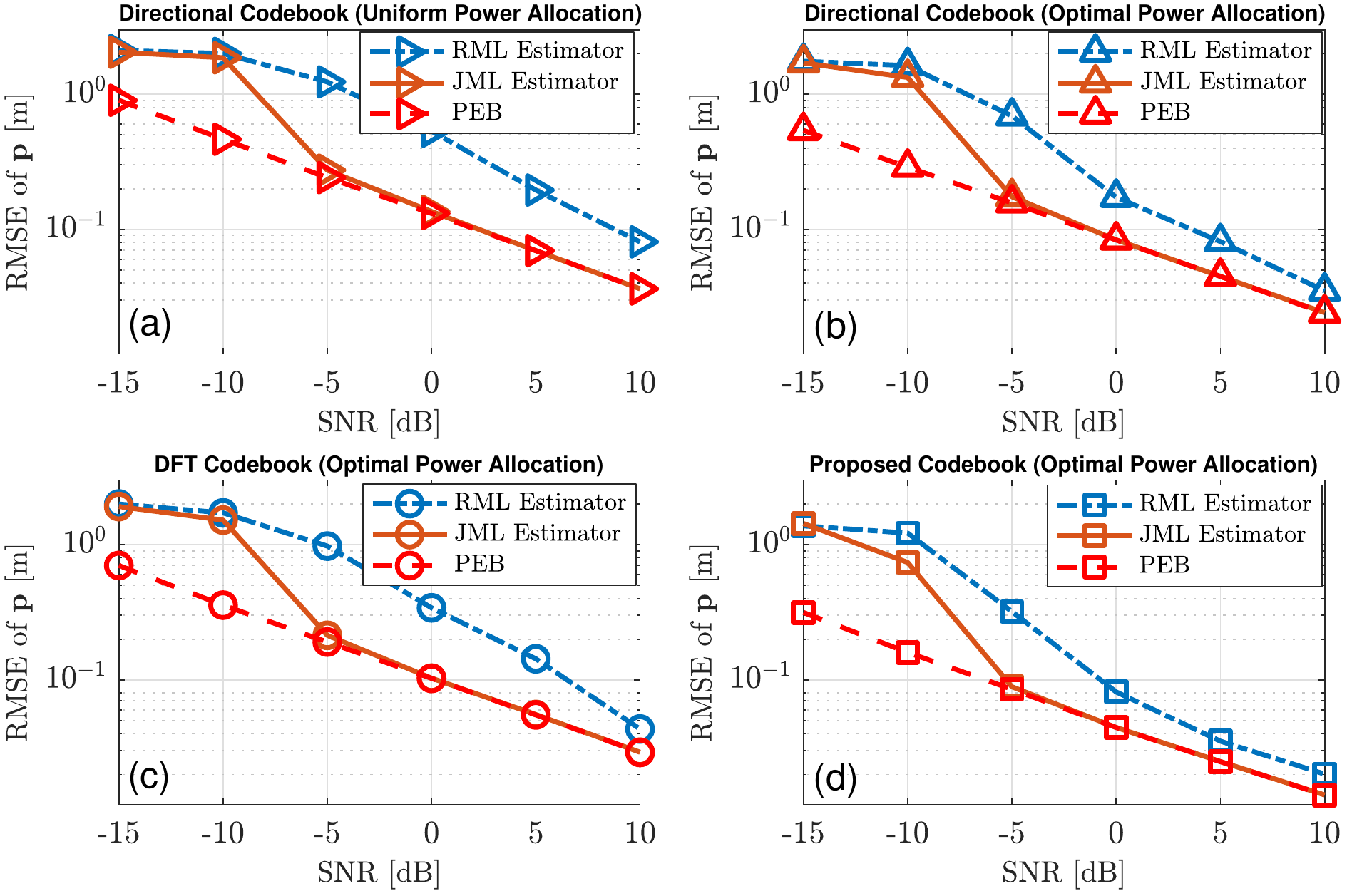}
 	\caption{RMSEs on the estimation of $\bm{p}$ as a function of the SNR for the directional codebook, DFT codebook, and proposed codebook.}
\label{fig:RMSE_pos_3m}
 \end{figure}

\begin{figure}
 \includegraphics[width=0.5\textwidth]{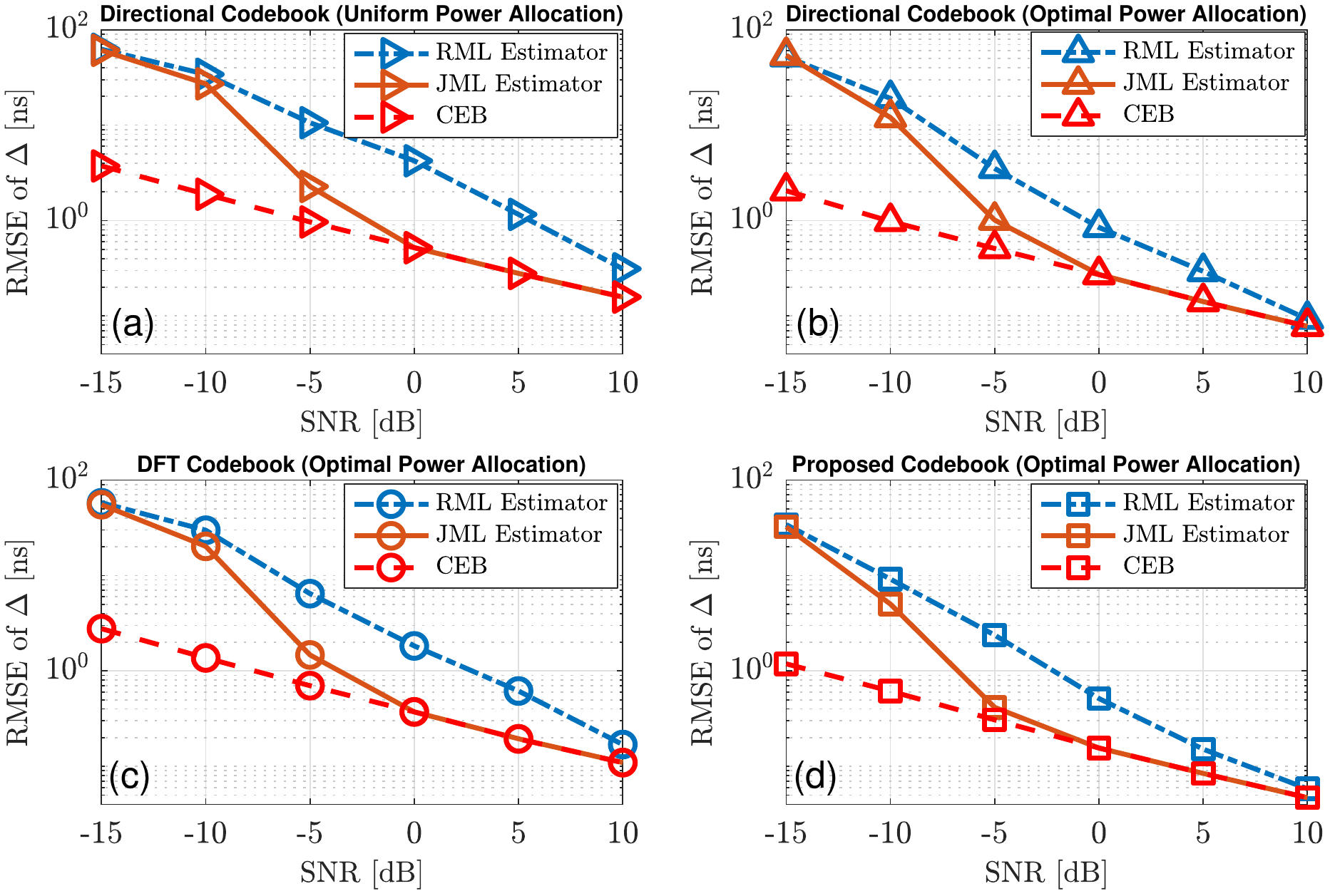}
 	\caption{RMSE on the estimation of $\Delta$ as a function of the SNR for the directional codebook, DFT codebook, and proposed codebook.}
\label{fig:RMSE_clock_3m}
 \end{figure}

\subsection{Results and Discussion}
\subsubsection{Comparison Between Uniform and Proposed Beam Power Allocation} We start the numerical analysis by assessing the validity of the beam power allocation strategy proposed in Sec.~\ref{sec_robust_design}. In this respect, we select as a precoding scheme the directional codebook given by \eqref{eq_directional_codebook_uniform} and perform a direct comparison between the case in which all the $G$ beams share the total transmitted power uniformly, i.e., each beam is transmitted with a  power equal to $P/G$, with the case in which at the $g$-th beam is allocated a fraction of the total power given by the corresponding $g$-th element of the allocation vector $\varrhob^*$, the latter obtained as a solution of the optimal power allocation problem in \eqref{eq_peb_codebook}.
Figs.~\ref{fig:RMSE_pos_3m}a-\ref{fig:RMSE_pos_3m}b and Figs.~\ref{fig:RMSE_clock_3m}a-\ref{fig:RMSE_clock_3m}b show the RMSEs on the estimation of the UE position $\bm{p}$ and clock offset $\Delta$, respectively, as a function of the SNR, for the directional codebook with uniform and optimal power allocation, also in comparison to the theoretical lower bounds (PEB and CEB\footnote{The \ac{CEB} is given by $\big[ \Jeta^{-1} \big]_{7,7}$, where $\Jeta$ is the FIM in \eqref{eq_Jeta}.}) derived in Sec.~\ref{sec_fim}. The proposed low-complexity estimation algorithm is denoted as ``JML" and it is implemented as described in Algorithm \ref{alg_estimator}, using the power allocation strategy proposed in Algorithm \ref{alg_codebook}. For completeness, we also report the performance of the RML estimator which is used to obtain an initial estimate of the vector $\bm{\Theta}$. By comparing the RMSEs in the Fig.~\ref{fig:RMSE_pos_3m}a and Fig.~\ref{fig:RMSE_pos_3m}b (analogously Fig.~\ref{fig:RMSE_clock_3m}a and Fig.~\ref{fig:RMSE_clock_3m}b), it clearly emerges that the proposed power allocation strategy yields more accurate estimates of the UE parameters compared to the uniform power allocation, as also confirmed by the gap between the corresponding lower bounds. Such results demonstrate that adopting a simple uncontrolled (uniform) power allocation scheme at the BS side likely leads to a waste of energy towards directions that do not provide useful contributions for the estimation process. The inefficient use of the transmitted power becomes even more critical in an RIS-assisted localization scenario, being the \ac{NLoS} channel linking the BS and the UE through the RIS subject to a more severe path loss resulting from the product of two separated tandem channels (ref. \eqref{eq_hbr} and \eqref{eq_hru}). In light of these considerations, the subsequent comparisons  will be  conducted assuming  optimal power allocation.

\subsubsection{Comparison Between State-of-the-art and Proposed Joint BS-RIS Signal Design} The set of figures reported in Fig.~\ref{fig:RMSE_pos_3m} and Fig.~\ref{fig:RMSE_clock_3m} show a detailed performance comparison between the proposed joint BS-RIS precoding scheme and the state-of-the-art approaches listed in Sec.~\ref{subs:benchmarks}, when used within the proposed low-complexity localization and synchronization algorithm. On the one hand, the obtained results demonstrate the effectiveness of the proposed estimation approaches: despite its intrinsic suboptimality, the RML algorithm (dash-dot curves) provides satisfactory initial estimates of both $\bm{p}$ and $\Delta$ parameters for all the considered precoding schemes, with an accuracy that tend to increase with the SNR and with a complexity reduced to a 2D search in the estimation process. {\ed Accordingly, the RMSEs of the RML and JML estimators are close when the SNR is small because in that regime the initial estimates $\hat{\bm{\Theta}}^{\text{\tiny RML}} = [\hat{\bm{p}}^\text{\tiny RML} \ \hat{\Delta}^{\text{\tiny FFT}}]^\mathsf{T}$ provided by the RML estimator are quite inaccurate. As a result, the iterative optimization procedure gets trapped into local wrong minima and leads to solutions (in terms of position and clock offset estimates) for the JML that are very close to that of the RML estimator.} Remarkably, the RMSEs of the JML estimator (solid curves) immediately attain the corresponding lower bounds as soon as the initialization provided by the RML becomes sufficiently accurate, providing excellent localization and synchronization performance already at $\text{SNR} = -5$ dB for all the considered precoding schemes.
\rev{To better highlight the necessity of adopting the more accurate initialization $\hat{\bm{\Theta}}^{\text{\tiny RML}}$ obtained via the proposed RML estimator, we also evaluate the performance of the JML estimator initialized with a random realization of the vector $\bm{\Theta}$. More specifically, we conduct an additional simulation analysis to directly compare the two versions of the JML estimator, using for the randomly-initialized JML an initial point $\bm{\Theta}^{\text{\tiny RND}} = [\bm{p}^{\text{\tiny RND}} \ \Delta^\text{\tiny RND}]^\mathsf{T}$ obtained by selecting $\bm{p}^{\text{\tiny RND}}$ as a random position within the assumed uncertainty region $\mathcal{P}$ and by generating $\Delta^\text{\tiny RND} \sim \mathcal{U}(0,2\Delta)$. To make the comparison fair, we re-generated the value of $\bm{\Theta}^{\text{\tiny RND}}$ for each independent Monte Carlo trial. The  results in terms of RMSEs reported in Fig.~\ref{fig::JMLrandvsinit} demonstrate that the performance of the JML estimator significantly worsen when $\bm{\Theta}^{\text{\tiny RND}}$ is used as initial point. This behavior is due to the fact that the iterative minimization of a highly non-linear cost function such as that of the JML estimator  gets trapped into local erroneous minima when a random initialization is used, and in turn produces wrong position and clock offset estimates. This confirms the need to seek for a more accurate initialization as the one proposed in Sec.~VI-B, which allows the JML estimator to attain the  theoretical lower bounds.
\begin{figure}
		\centering
		\includegraphics[width=0.35\textwidth]{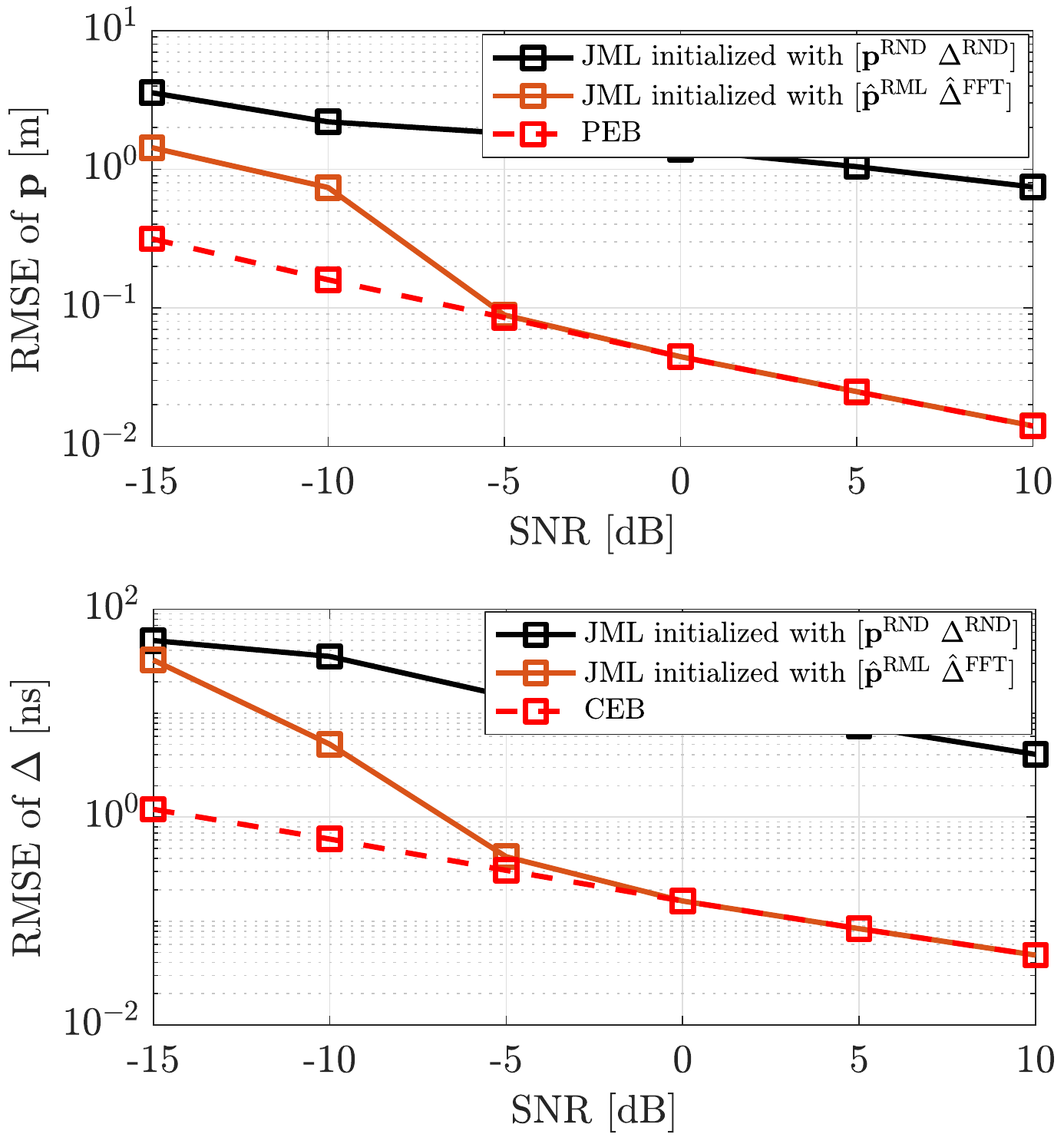}
		\caption{\rev{Comparison between the JML  initialized with the estimates $\hat{\bm{\Theta}}^{\text{\tiny RML}}$ provided by the proposed RML estimator and with  random values $\bm{\Theta}^{\text{\tiny RND}}$.}}\label{fig::JMLrandvsinit}
	\end{figure}
}

On the other hand, a direct comparison among the PEBs and CEBs in Figs.~\ref{fig:RMSE_pos_3m}-\ref{fig:RMSE_clock_3m} (dashed curves) reveals that the proposed robust joint BS-RIS precoding strategy offers the best localization and synchronization performance among all the considered schemes. To better highlight the advantages of the proposed approach, in Fig.~\ref{fig:comparison_RMSE_3m} we report an explicit comparison among the RMSEs on the estimation of $\bm{p}$ and $\Delta$ for the JML estimator fed with different precoding schemes.
\begin{figure}%
    \centering
    \subfloat[\centering RMSE on the estimation of $\bm{p}$.]{{\includegraphics[width=0.35\textwidth]{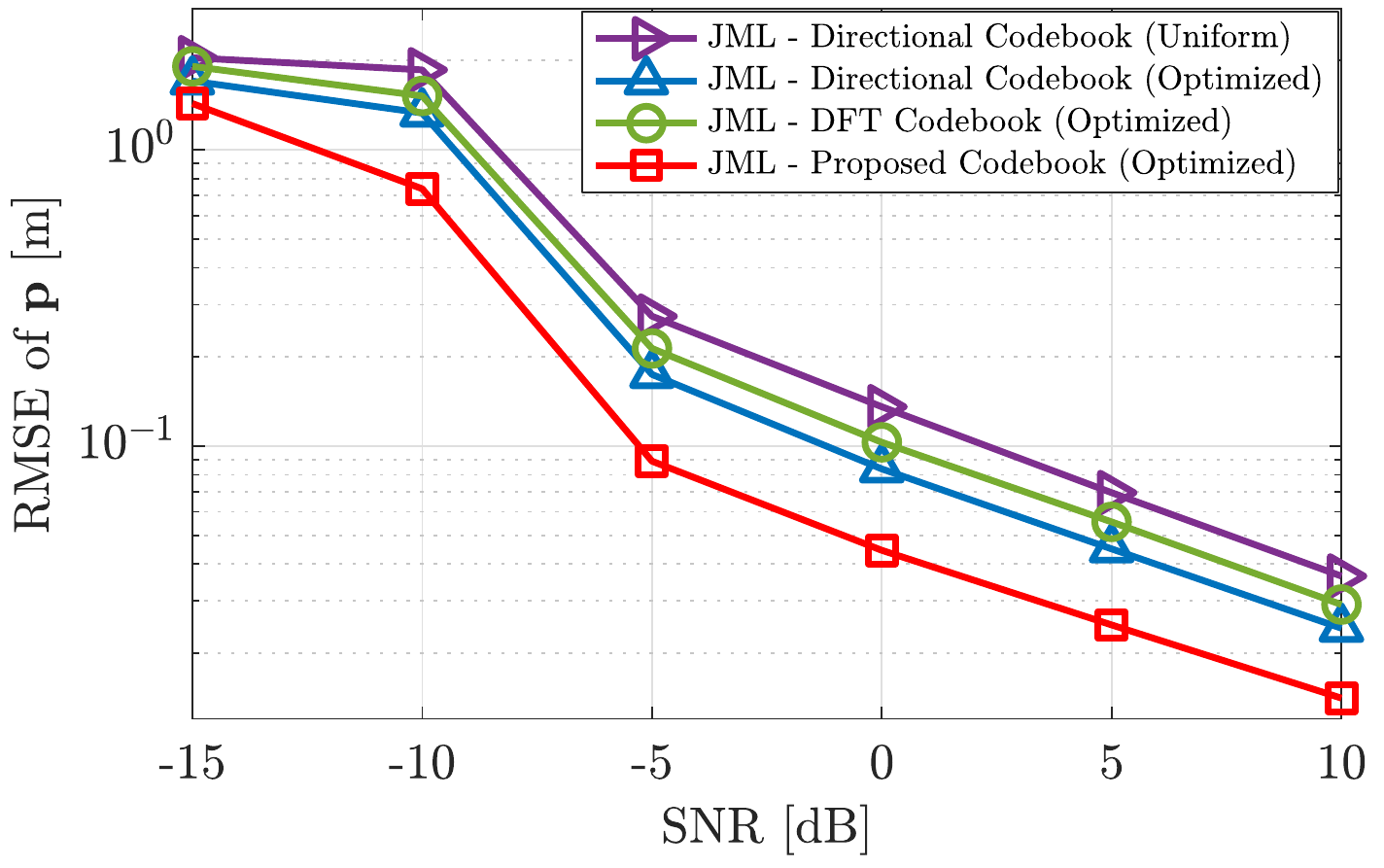} } \label{fig:comparison_pos_3m}}%
    \qquad
    \subfloat[\centering RMSE on the estimation of $\Delta$.]{{\includegraphics[width=0.35\textwidth]{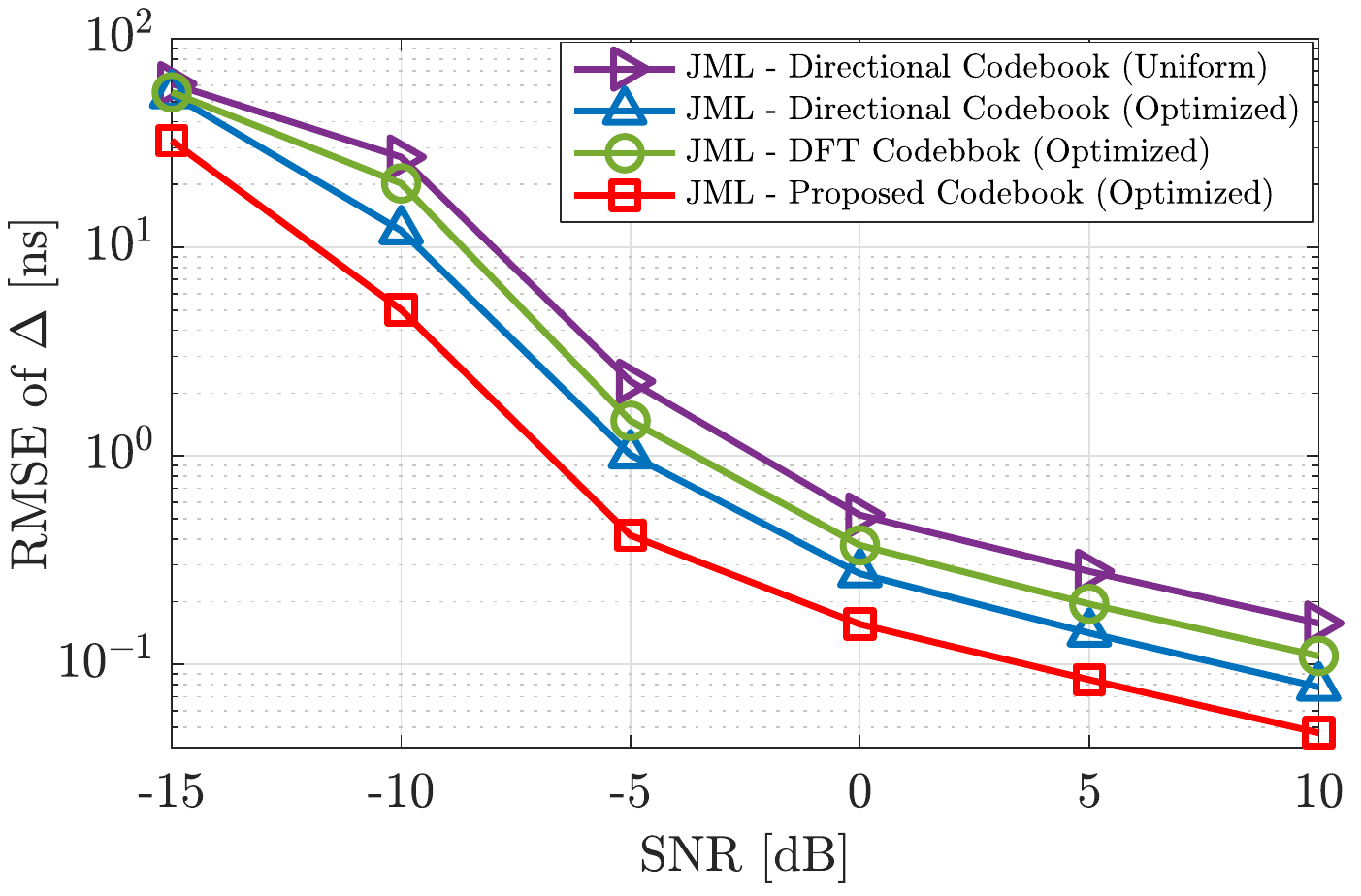} }\label{fig:comparison_clock_3m}}%
    \caption{Comparison between the RMSEs of (a) $\bm{p}$ and (b) $\Delta$ using the proposed JML estimator for different precoding schemes, as a function of the SNR.}%
    \label{fig:comparison_RMSE_3m}%
\end{figure}
As it can be noticed, the proposed robust joint BS-RIS precoding scheme significantly outperforms both the directional and DFT codebooks. Interestingly, the values assumed by the corresponding RMSEs  in Figs.~\ref{fig:comparison_pos_3m}-\ref{fig:comparison_clock_3m} (solid curves with $\square$ marker) demonstrate that the UE can be localized with an error lower than $10$ cm and, at the same time, synchronization can be recovered with a sub-nanosecond precision, for $\text{SNR} \geq -5$ dB. From this analysis, we can conclude that combining the proposed codebooks in \eqref{eq_fbs}-\eqref{eq_fris} with a power allocation strategy that aims at minimizing the worst-case PEB allows us to achieve a better coverage of the uncertainty region $\ppreg$, while properly taking into account the different  directional and derivative beams transmitted towards the UE and the RIS.

{\ed
\subsubsection{Robustness Analysis}  We now conduct an  analysis aimed at assessing the effective robustness of the proposed joint BS and RIS beamforming design strategy to different UE positions falling within the assumed uncertainty region $\mathcal{P}$. More specifically, we test the values assumed by the PEB when the UE  spans different locations around the nominal one $\mathbf{p}$, considering both the proposed robust design approach and its corresponding non-robust version, the latter obtained by simply shrinking the extent of the uncertainty region to a very small area of 0.1 m  around the nominal UE position. The results reported in Fig.~\ref{fig:PEB_robust}, obtained for $\text{SNR} = 0$ dB, show that the PEB exhibits quite similar values within the whole uncertainty region, confirming the robustness of the proposed joint active BS and passive RIS beamforming design strategy. Interestingly, the PEB keeps reasonable values even when the UE falls slightly outside the considered region $\mathcal{P}$. Conversely, the values assumed by the PEB in Fig.~\ref{fig:PEB_nonrobust} clearly indicate an evident position accuracy degradation for UE locations different from the nominal one, leading to errors that are almost three times those experienced in Fig.~\ref{fig:PEB_robust} with the proposed robust joint design strategy.
\begin{figure}%
    \centering
    \subfloat[\centering \ed PEB evaluation for the proposed robust joint BS-RIS design strategy.]{{\includegraphics[width=0.35\textwidth]{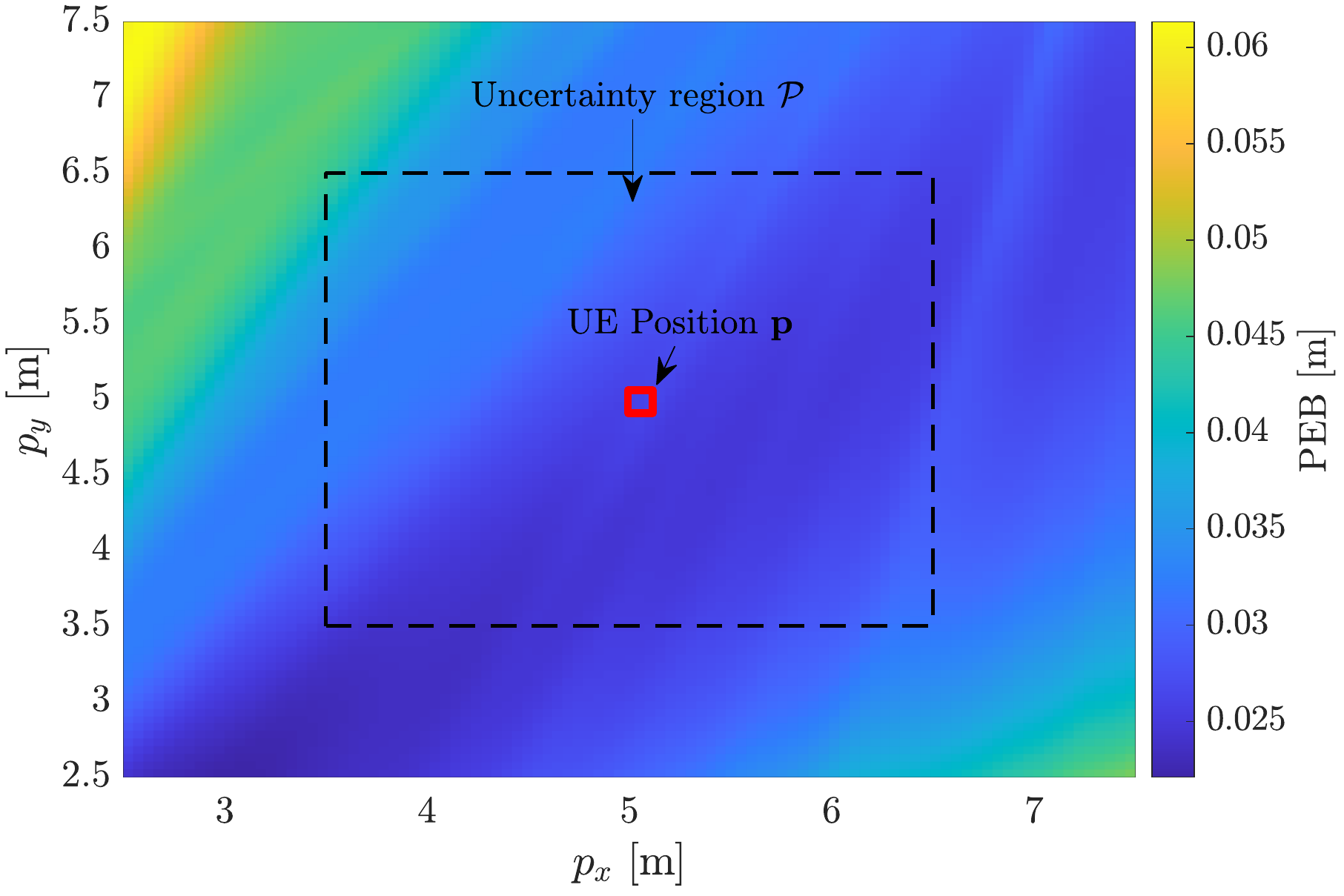} } \label{fig:PEB_robust}}%
    \quad
    \subfloat[\centering \ed PEB evaluation for  non-robust joint BS-RIS design strategy.]{{\includegraphics[width=0.35\textwidth]{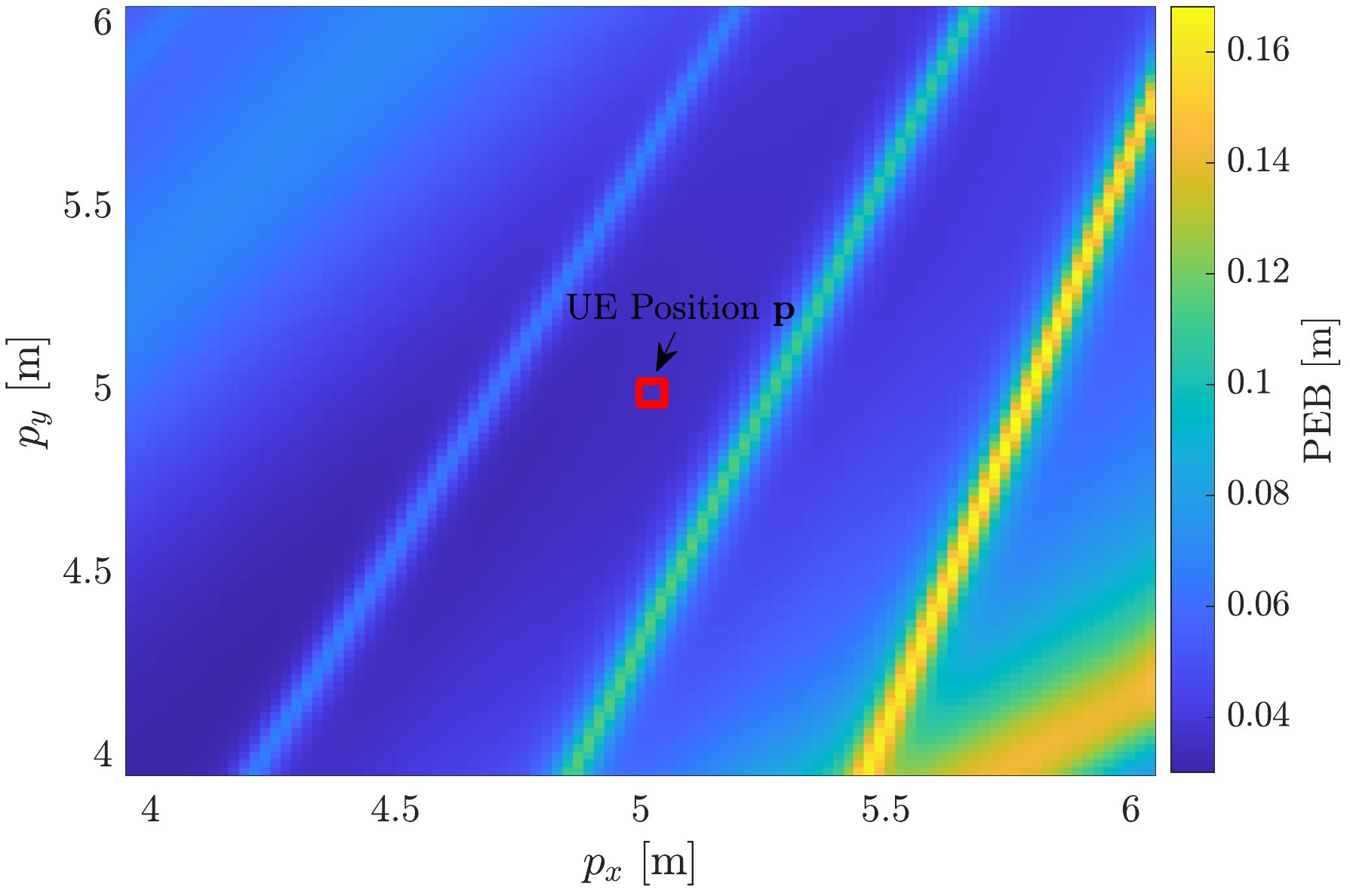} }\label{fig:PEB_nonrobust}}%
    \caption{\ed Comparison between the PEBs with robust and non-robust joint BS-RIS beamforming design strategies.}%
    \label{fig:comparison_robustness}%
\end{figure}
}

\subsubsection{Performance Assessment for Reduced Number of Transmitted Beams} \label{sec_reducedG} To corroborate the above results, we investigate the possibility to adopt an ad-hoc heuristic that allows to reduce the total number of transmitted beams $G$. The main idea originates from observing that, when the BS is transmitting a beam towards the UE, all the different configurations of the RIS phase profiles should not have a significant impact onto the ultimate localization and synchronization performance, being the BS-RIS path likely illuminated  with a negligible amount of transmitted power. In other words, when
$\FFbs_{:,i} = \bm{a}^*_\BS(\thetabm^{(i)})$ or $\FFbs_{:,i} = \aabsdt^\conj(\thetabm^{(i)})$, we propose to neglect the transmission of the $2\lr$ different RIS beams $\FFris_{:,j} = \aarisw^\conj(\thetarm^{(j)})$ and $\FFris_{:,j} = \aariswdttilde^\conj(\thetarm^{(j)})$, for $j = 1,\ldots, \lr$, and use for the corresponding pairs $\{ \FFbs_{:,i}, \FFris_{:,j} \}$ a single configuration of the RIS phase profile given by $\FFris_{:,j} = \aarisw(\thetarm^{c})$. {\ed Conversely, when the BS is transmitting the beam towards the RIS, that is, the precoding vector is set to $\FFbs_{:,i} = \aabs^*(\thetabr)$, we consider for $\FFris_{:,j}$ all the $2\lr$ possible configurations of the RIS phase profile. In doing so, the signal received by the UE in \eqref{eq_mgnbm} will be always observed for different RIS phase profiles, providing the necessary information to estimate the \ac{AoD} $\theta_{{\RM}}$. This procedure allows us to reduce the total number of transmitted beams to $G = 2\lb + 2\lr + 1$.}

To validate such an intuition, in Fig.~\ref{fig:reducedG} we compare the RMSEs on the estimation of $\bm{p}$ and the related PEBs as a function of the SNR, for both cases of full and reduced number of transmissions $G$.
\begin{figure}%
    \centering
    {\includegraphics[width=0.35\textwidth]{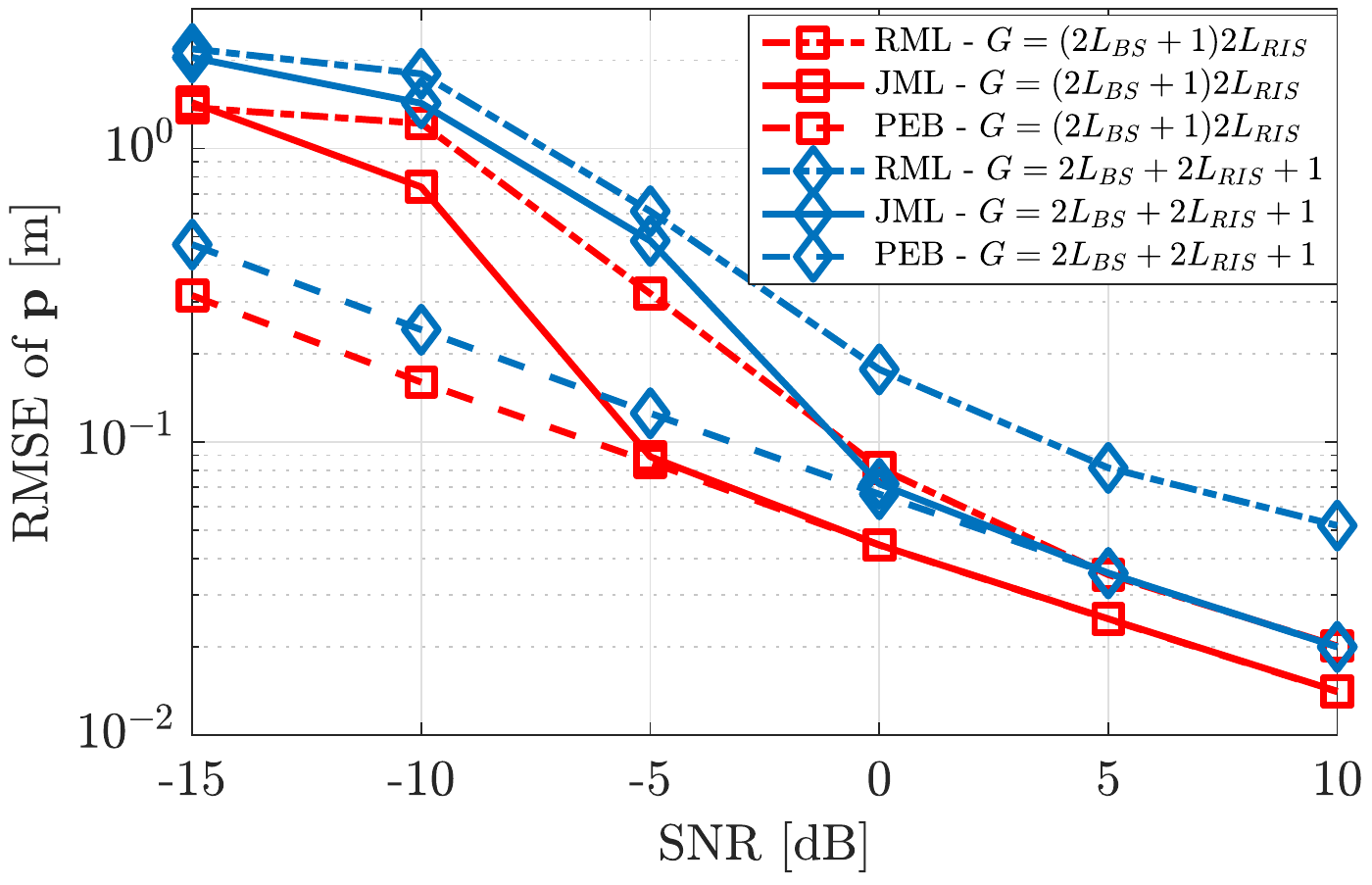}}
    \caption{Performance comparison between the case of full number of transmitted beams $G = (2 \lb + 1) 2 \lr$ and proposed heuristic using a reduced $G = 2\lb + 2\lr + 1$.}%
    \label{fig:reducedG}%
\end{figure}
By comparing the PEBs in Fig.~\ref{fig:reducedG}, we appreciate a slight degradation of the theoretical accuracy achievable in case of reduced $G$ (analogous behavior is obtained for $\Delta$). Interestingly, despite the more challenging scenario, the JML estimator combined with the proposed power allocation strategy (solid curves with $\diamond$ marker) is still able to provide very accurate localization performance, though attaining the bounds at higher SNR of $0$ dB. In this respect, an important trade-off between the estimation accuracy and the total number of transmission tends to emerge: for this specific case, the proposed heuristic leads to a $85$\% reduction in the number of involved transmissions (and, consequently, in the time needed to localize and synchronize the UE), but at the price of slightly increased values of RMSEs and related bounds. In Appendix~\ref{sec_5m_uncertainty}, we have conducted a similar analysis for the case in which the uncertainty region $\ppreg$ has been increased to 5 m along each direction. The obtained results reveal that the gaps between the estimation performance in cases of full and reduced $G$ tend to increase as the uncertainty increases. This behavior can be explained by noting that the proposed heuristic is based on the underling assumption that almost no power is received by the RIS when the BS is transmitting a beam in the directions of the UE. However, when the uncertainty region $\ppreg$ grows, the corresponding set of \ac{AoD}s from the BS to the UE $\{ \thetabmtilde^{(i)} \}_{i=1}^{2\lb}$ progressively spans an increased area and, consequently, some of the beams directed towards the UE could likely illuminate the RIS path with a non-negligible amount of power. In these cases, the different configurations of the RIS phase profiles start to have a noticeable effect on the resulting estimation accuracy, thus preventing the possibility to reduce $G$ without experiencing evident performance losses. Overall, an interesting insight can be derived from this analysis: the smaller the initial uncertainty about the UE position, the shorter the time required to localize and synchronize it accurately. 

{\ed To further corroborate these insights, we consider a second different scenario in which the UE position is moved to $\bm{p} = [3 \ -1]^\mathsf{T}$ m, so that the angular separation $|\theta_{{\BM}} - \theta_{{\RM}}|$ between the UE and the RIS  increases from $14.7^\circ$ to $48.7^\circ$ in this new configuration, while the rest of the parameters are kept the same as for Fig.~\ref{fig:reducedG}. The effect of this change on the gap between the estimation performances of the full-$G$ and reduced-$G$ cases can be observed in Fig.~\ref{fig:reducedG_largeseparation}. More specifically, the gap between the RMSEs on the estimation of $\bm{p}$ of the full-$G$ and reduced-$G$ cases is significantly reduced by moving the UE to a location where the \ac{AoD} difference becomes much larger (analogous behavior is obtained for $\Delta$). This behavior is perfectly in line with our previous findings and can be explained by noting that the RIS practically receives a negligible amount of power when the BS is illuminating the UE, being the extent of the uncertainty region $\mathcal{P}$ not sufficiently large to include beams that  illuminate the RIS. Hence, it can be concluded that, for scenarios with widely separated \ac{AoD}s and sufficiently small uncertainty regions, it is reasonable to employ the codebook with reduced $G$ since it provides almost the same performance as in the case of full $G$ using smaller number of transmissions. In this respect, it is worth noting that the gap between the two cases will similarly reduce also in the cases in which the \ac{AoD}s from BS to RIS and UE are close, but the uncertainty region is small enough to guarantee that the beams do not illuminate the RIS path. Overall, we can conclude that the performances in cases of full and reduced G are mainly related to both the UE location and the extent of the uncertainty region $\mathcal{P}$.

\begin{figure}%
    \centering
    {\includegraphics[width=0.35\textwidth]{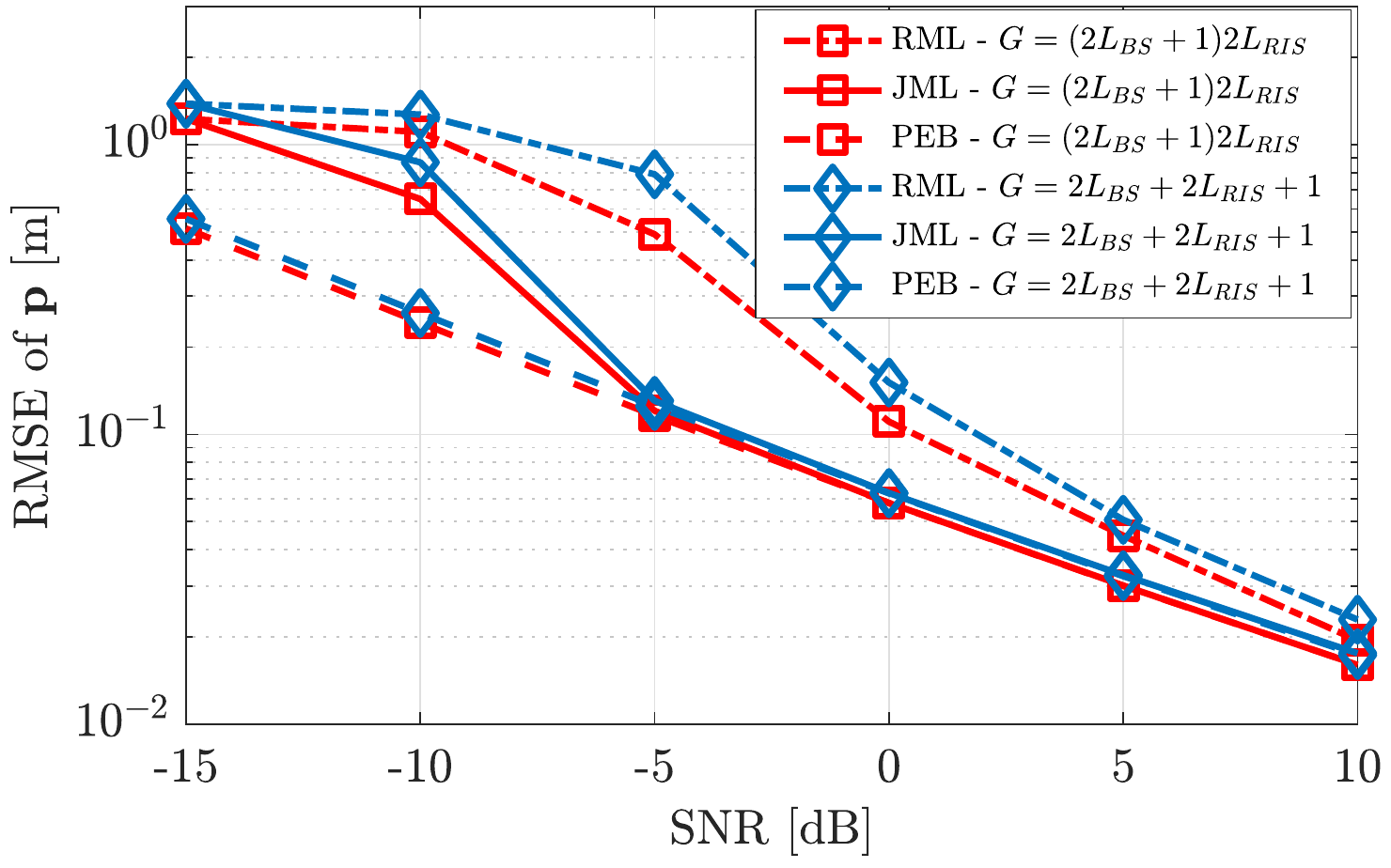}}
    \caption{\ed Performance comparison between the case of full number of transmitted beams $G = (2 \lb + 1) 2 \lr$ and proposed heuristic using a reduced $G = 2\lb + 2\lr + 1$ for increased \ac{AoD} separation between UE and RIS.}%
    \label{fig:reducedG_largeseparation}%
\end{figure}
}
\subsubsection{Performance Assessment in Presence of Uncontrollable Multipath}\label{sec_uncontrolled_path} 
\begin{figure}%
    \centering
    \includegraphics[width=0.36\textwidth]{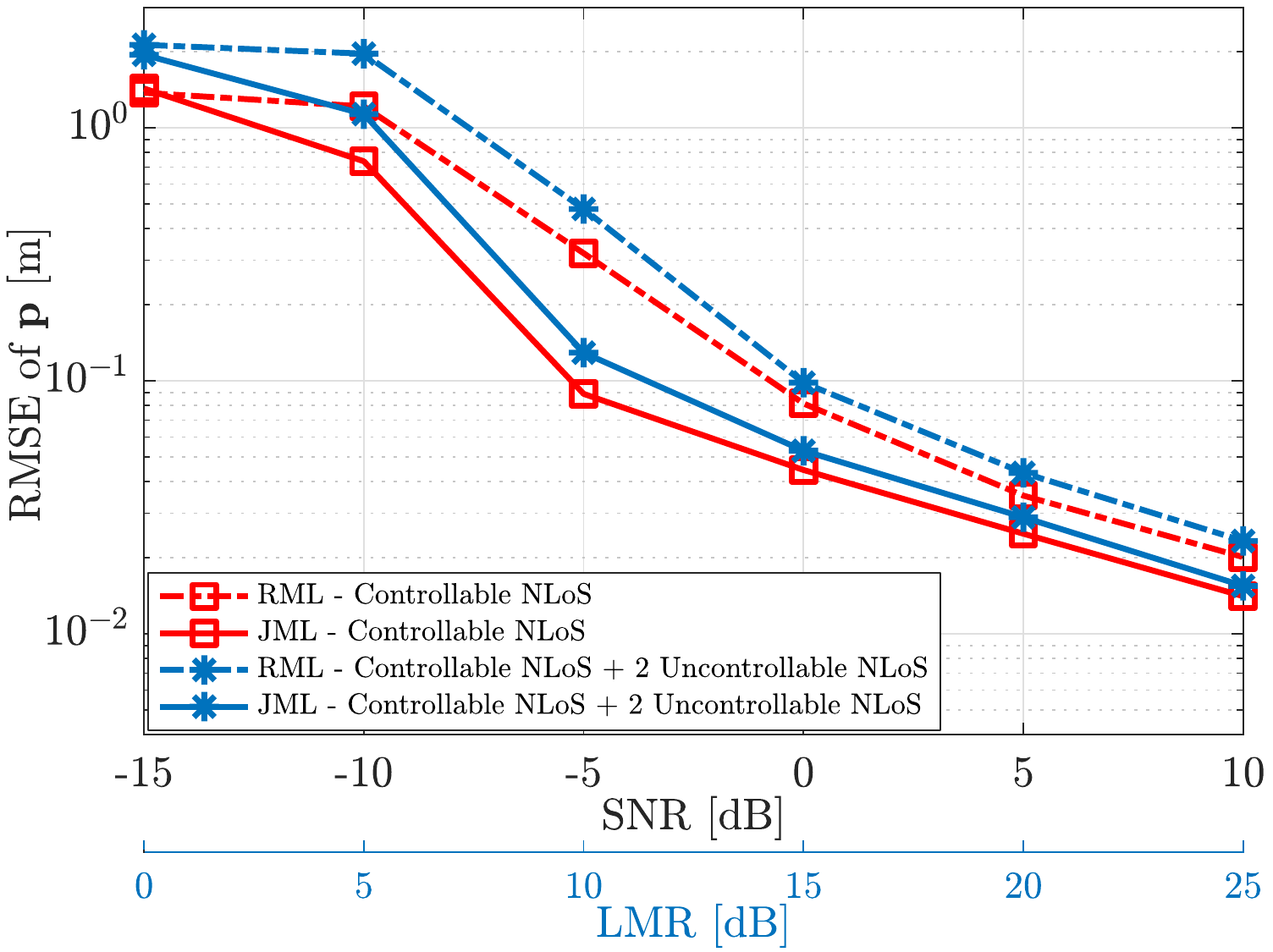}
     \caption{Performance comparison between the matched scenario with a single controllable \ac{NLoS} path through RIS and the mismatched scenario with two additional uncontrollable \ac{NLoS} paths.}%
    \label{fig:multipath}%
\end{figure}
To further challenge the proposed approach, we also investigate a scenario accounting for the simultaneous presence of the controllable \ac{NLoS} path through the RIS, as well as of two additional uncontrollable \ac{NLoS} paths generated by two local scatterers in the surrounding environment, located at unknown positions $\bm{m}_1 = [2 \ 7]^\mathsf{T}$ m and $\bm{m}_2 = [6 \ 2]^\mathsf{T}$, respectively. In doing so, we can test the robustness of the algorithms in a propagation environment that is mismatched with respect to the model considered at design stage. Assuming that each uncontrollable \ac{NLoS} path is characterized by a single dominant ray, we generate the absolute value of the complex amplitudes as $|\alpha_{\text{\tiny NLOS},i}| = \Gamma \lambda_c/(4\pi[\|\bm{m}_i\| + \|\bm{m}_i - \bm{p}\|)])$, with $\Gamma = 0.7$ reflection coefficient \cite{abu2018error}.
To analyze the robustness under different multipath conditions, we keep fixed the power of the uncontrollable \ac{NLoS} paths (denoted by $P^i_\text{\tiny \ac{NLoS}}, i=1,2$) and increase only the power along the \ac{LoS} path (denoted by $P_\text{\tiny \ac{LoS}}$) and the controllable RIS path, using the \ac{LMR} defined as $\text{LMR} = P_\text{\tiny LOS}/\sum_{i=1}^2 P^i_\text{\tiny \ac{NLoS}}$. For the considered setup, 
varying the SNR in the range from $-15$ dB up to $10$ dB corresponds to a LMR varying from $0$ dB up to $25$ dB, with $5$ dB steps.
In Fig.~\ref{fig:multipath}, we show the evolution of the RMSEs on the estimation of $\bm{p}$ as a function of the SNR, for both cases with and without uncontrollable \ac{NLoS} paths.
The obtained results reveal that the proposed approach is effective also in this more challenging scenario, with both the RML and JML algorithms that exhibit a slight degradation of the achieved localization performance only for small values of the SNR (similar considerations hold true for the clock offset $\Delta$), that is, when the multipath in terms of LMR is more severe. 

\section{Conclusion}
In this paper, we have considered the problem of joint localization and synchronization of a single-antenna \ac{UE} served by a single \ac{BS} in the presence of a \ac{RIS}, assuming the existence of a \ac{LoS} path and a controllable \ac{NLoS} path through the \ac{RIS}. To maximize the performance of localization and synchronization under UE location uncertainty, a novel codebook-based low-complexity design strategy for joint optimization of active BS precoding and passive RIS phase shifts has been proposed, based on the derived low-dimensional structure of precoders and phase profiles. In addition, we have developed a reduced-complexity \ac{ML}-based estimator by exploiting the special signal structure that enables decoupled estimation of UE location and clock offset. Extensive  simulations showed that the proposed joint BS-RIS beamforming approach provides significant improvements in both localization and synchronization performance (on the order of meters and nanoseconds, respectively, at SNR = $-10$ dB) over the state-of-the-art benchmarks. Moreover, the proposed estimator is able to attain the corresponding theoretical limits at a relatively low SNR (around $-5$ dB) and found to be resilient against uncontrolled multipath. As a future direction, we plan to evaluate the impact of discrete RIS phase shifts on the estimation performance.


\appendices

\section{Proof of Proposition~\ref{prop_BS_precoder}}\label{app_proof_BS}
\noindent Following similar arguments as in \cite[App.~C]{li2007range}, we represent the covariance matrix in \eqref{eq_problem_peb_relaxed_fixed_RIS} for the $g$-th transmission as
\begin{align}\label{eq_xxgstar}
    \XX_g = \Gammabbig \Gammabbig^\hermit ~,
\end{align}
where $\Gammabbig$ admits a decomposition
\begin{align}\label{eq_gammabig}
    \Gammabbig = \projrange{\aabsbig} \Gammabbig + \projnull{\aabsbig} \Gammabbig ~,
\end{align}
with $\projrange{\XX} \eqdef \XX (\XX^H \XX)^{-1} \XX^H$ denoting the orthogonal projector onto the columns of $\XX$ and $\projnull{\XX} \eqdef \Imatrix - \projrange{\XX}$. Then, $\XX_g$ in \eqref{eq_xxgstar} can be re-written using \eqref{eq_gammabig} as
\begin{align}\label{eq_xxgstar_dec}
    \XX_g = \XXgbar + \XXgtilde ~,
\end{align}
where
\begin{align} \label{eq_xxgbar_def}
\XXgbar &\eqdef \projrange{\aabsbig} \Gammabbig \Gammabbig^\hermit \projrange{\aabsbig} \\ \label{eq_xxgtilde_def}
    \XXgtilde &\eqdef \projrange{\aabsbig} \Gammabbig \Gammabbig^\hermit \projnull{\aabsbig} + \projnull{\aabsbig} \Gammabbig \Gammabbig^\hermit \projrange{\aabsbig}
    \\ \nonumber &~~~~+ \projnull{\aabsbig} \Gammabbig \Gammabbig^\hermit \projnull{\aabsbig}~.
\end{align}
Since $\projnull{\aabsbig} \aabsbig = \boldzero$ by definition, we have
\begin{align}\label{eq_xxgtilde_zero}
    \aabsbig^\hermit \XXgtilde \aabsbig = \boldzero ~.
\end{align}

\noindent We now provide three lemmas to facilitate the proof of Prop.~\ref{prop_BS_precoder}.
\begin{lemma} \label{lemma_fim_depend}
    The FIM $\JJgamma$ in \eqref{eq::FIMelemformula} does not depend on the component $\XXgtilde$ of $\XX_g$ in \eqref{eq_xxgstar_dec}.
\end{lemma}
\begin{proof}
Based on the definition of $\aabsbig$ in \eqref{eq_a_bs} and the FIM elements in Appendix~\ref{sec_fim_func}, we observe that the dependence of the FIM $\JJgamma$ on $\XX_g$ is only through the elements of $\aabsbig^\hermit \XX_g \aabsbig \in \complexset{3}{3}$. Then, it follows from \eqref{eq_xxgstar_dec} and \eqref{eq_xxgtilde_zero} that the FIM does not depend on $\XXgtilde$, i.e., the dependence of the FIM on $\XX_g$ is only through $\XXgbar$ in \eqref{eq_xxgstar_dec}.
\end{proof}

\rev{\begin{remark}\label{lemma_power}
The component $\XXgtilde$ of $\XX_g$ in \eqref{eq_xxgstar_dec} contributes non-negatively \! to the total power consumption, i.e., $\! \tracenormal{\! \XXgtilde}\! \geq\! 0$.
\end{remark}}
\begin{proof}
Opening up the terms in $\XXgtilde$ in \eqref{eq_xxgtilde_def}, we have
\begin{align}\nonumber
    \tracenormal{\XXgtilde} &= \tracesmall{\projrange{\aabsbig} \Gammabbig \Gammabbig^\hermit \projnull{\aabsbig}} + \tracesmall{ \projnull{\aabsbig} \Gammabbig \Gammabbig^\hermit \projrange{\aabsbig}} 
    \\ \nonumber
    &~~~~+ \tracesmall{ \projnull{\aabsbig} \Gammabbig \Gammabbig^\hermit \projnull{\aabsbig} }
    \\ \nonumber
    &= \tracesmall{\Gammabbig \Gammabbig^\hermit \projnull{\aabsbig} \projrange{\aabsbig}} + \tracesmall{ \Gammabbig \Gammabbig^\hermit \projrange{\aabsbig} \projnull{\aabsbig} } 
    \\ \nonumber
    &~~~~+ \norm{\Gammabbig^\hermit \projnull{\aabsbig}}_{\frob}^2
    \\ \label{eq_tr_xxg}
    &= \norm{\Gammabbig^\hermit \projnull{\aabsbig}}_{\frob}^2 \geq 0 ~,
\end{align}
where $\norm{\cdot}_{\frob}$ represents the matrix Frobenius norm. 
\end{proof}

\begin{lemma}\label{lemma_power_zero}
The component $\XXgtildestar$ in \eqref{eq_xxgstar_dec} of an optimal $\XXgstar$ obtained as the solution to \eqref{eq_problem_peb_relaxed_fixed_RIS} satisfies $\tracenormal{\XXgtildestar} = 0$.
\end{lemma}
\begin{proof}
To prove the lemma, we resort to proof by contradiction. For a given optimal solution 
\begin{align} \label{eq_xxgstar_opt}
    \XXgstar = \XXgbarstar + \XXgtildestar, ~~ g = 1, \ldots, G 
\end{align}
with $\tracenormal{\XXgtildestar} > 0$ for some $g$, consider an alternative solution
\begin{align} \label{eq_xxgstarr}
    \XXgstarr = \XXgbarstarr + \XXgtildestarr, ~~ g = 1, \ldots, G  
\end{align}
where
\begin{align} \label{eq_xxgbarstarr}
     \XXgbarstarr &\eqdef \XXgbarstar \left(1 + \frac{\tracesmall{\sum_{g=1}^{G}\XXgtildestar}}{\tracesmall{\sum_{g=1}^{G}\XXgbarstar}}  \right)
     \\ \label{eq_xxgtildestarr}
     \XXgtildestarr &\eqdef \boldzero ~.
\end{align}
It can be readily verified from \eqref{eq_xxgstar_opt}--\eqref{eq_xxgtildestarr} that
\begin{align} \label{eq_power_const_satisfied}
    \tracebig{\sum_{g=1}^{G} \XXgstarr} = \tracebig{\sum_{g=1}^{G} \XXgstar} ~.
\end{align}
In addition, from Lemma~\ref{lemma_fim_depend}, we note that the FIM obtained for $\XXgstar$ in \eqref{eq_xxgstar_opt} and $\XXgstarr$ in \eqref{eq_xxgstarr} depend only on $\aabsbig^\hermit \XXgbarstar \aabsbig$ and $\aabsbig^\hermit \XXgbarstarr \aabsbig$, respectively. Since $\aabsbig^\hermit \XXgbarstarr \aabsbig = \zeta \aabsbig^\hermit \XXgbarstar \aabsbig  $ for some $\zeta > 1$ according to \eqref{eq_xxgbarstarr}, the alternative solution $\XXgstarr$ in \eqref{eq_xxgstarr} would achieve smaller PEB than the optimal solution $\XXgstar$ in \eqref{eq_xxgstar_opt} (due to scaling of the FIM by $\zeta > 1$). Combining this with \eqref{eq_power_const_satisfied} shows that $\XXgstar$ cannot be an optimal solution of \eqref{eq_problem_peb_relaxed_fixed_RIS}, which completes the proof.
\end{proof}

Based on \eqref{eq_tr_xxg} in \rev{Remark}~\ref{lemma_power}, it can be observed that $\tracenormal{\XXgtilde} = 0$ implies $\Gammabbig^\hermit \projnull{\aabsbig} = 0$, which in turn yields $\XXgtilde = \boldzero$ using \eqref{eq_xxgtilde_def}. Hence, from \rev{Remark}~\ref{lemma_power} and Lemma~\ref{lemma_power_zero}, we infer that $\XXgtildestar$ of an optimal $\XXgstar$ should satisfy $\XXgtildestar = \boldzero$. Finally, from \eqref{eq_xxgstar_dec} and \eqref{eq_xxgbar_def}, an optimal $\XXgstar$ obtained as the solution to \eqref{eq_problem_peb_relaxed_fixed_RIS} can be expressed as
\begin{align} \label{eq_xxgstar_abs}
    \XXgstar &= \projrange{\aabsbig} \Gammabbig \Gammabbig^\hermit \projrange{\aabsbig} 
    \\ \nonumber
    &= \aabsbig \Upsilonb_g \aabsbig^\hermit ~,
\end{align}
where
$
    \Upsilonb_g \eqdef (\aabsbig^\hermit \aabsbig)^{-1} \aabsbig^\hermit \Gammabbig \Gammabbig^\hermit \aabsbig  (\aabsbig^\hermit \aabsbig)^{-1},
$
which completes the proof of Proposition~\ref{prop_BS_precoder}. 

\rev{We note that there exists an equivalent orthogonal solution $\Upsilonb_g$ (corresponding to the full-rank version of $\aabsbig$) in \eqref{eq_xxgstar_abs}, leading to the same covariance $\XX_g$, as shown in Appendix~\ref{sec_supp_orthogonal}. Moreover, it is worth highlighting that $\XX_g$ and $\Upsilonb_g$ are two identical solutions (having different dimensions) of the problem \eqref{eq_problem_peb_relaxed_fixed_RIS}, and thus one can always be obtained from the other using \eqref{eq_xxgstar_abs} and
\begin{align}\label{eq_ups_inv}
    \Upsilonb_g &= (\aabsbig^\hermit \aabsbig)^{-1} \aabsbig^\hermit \XX_g \aabsbig (\aabsbig^\hermit \aabsbig)^{-1} ~.
\end{align}
In other words, \textit{(i)} one can either solve \eqref{eq_problem_peb_relaxed_fixed_RIS} directly with respect to $\XX_g \in \complexset{\nbs}{\nbs}$, or, \textit{(ii)} one can solve \eqref{eq_problem_peb_relaxed_fixed_RIS} with respect to $\Upsilonb_g \in \complexset{3}{3}$ by inserting the relation $\XX_g = \aabsbig \Upsilonb_g \aabsbig^\hermit $ into both the objective \eqref{eq_problem_peb_relaxed_fixed_RIS} and the constraint \eqref{eq_const_xx}, and find the corresponding $\XX_g $ through $\XX_g = \aabsbig \Upsilonb_g \aabsbig^\hermit $.
}

\section{Proof of Proposition~\ref{prop_RIS_profile}}\label{app_proof_RIS}
From Appendix~\ref{sec_fim_nlos}, we observe that the FIM $\JJgammablkdiag$ in \eqref{eq_JJgamma_bdiag_approx} depends on $\Psibbig_g$ only through the elements of the matrix $\aariswbig^\hermit \Psibbig_g \aariswbig \in \complexset{2}{2}$. Then, the claim in the proposition can easily be proved by employing similar arguments to those in Appendix~\ref{app_proof_BS}.


\section{FIM in the Channel Domain}\label{sec_fim_channel}
In this part, we derive the entries of the channel domain FIM in \eqref{eq::FIMelemformula}.

\subsection{Derivatives as a Function of $\ff_g$ and $\oomegag$}\label{sec_der_fim}
Based on \eqref{eq_mgnbm}, we can obtain the derivatives as follows:
\begin{align*} 
 \frac{\partial \mgn}{\partial \taubm} &= \sqrt{P} \rhobm \e^{j\varphibm} \left[ \ccdt(\taubm) \right]_n \aabs^\trpose(\thetabm) \ff_g \sgn
 \\ 
    \frac{\partial \mgn}{\partial \thetabm} &= \sqrt{P} \rhobm \e^{j\varphibm} \left[ \cc(\taubm) \right]_n \aabsdt^\trpose(\thetabm) \ff_g \sgn
 \\ 
     \frac{\partial \mgn}{\partial \rhobm} &= \sqrt{P}  \e^{j\varphibm} \left[ \cc(\taubm) \right]_n \aabs^\trpose(\thetabm) \ff_g \sgn
 \\ 
     \frac{\partial \mgn}{\partial \varphibm} &= j \sqrt{P} \rhobm \e^{j\varphibm} \left[ \cc(\taubm) \right]_n \aabs^\trpose(\thetabm) \ff_g \sgn
     \\ 
     \frac{\partial \mgn}{\partial \taurm} &= \sqrt{P} \rhor \e^{j\varphir} \left[ \ccdt(\taur) \right]_n \aarisw^\trpose \oomegag \aabs^\trpose(\thetabr) \ff_g \sgn
     \\ 
     \frac{\partial \mgn}{\partial \thetarm} &= \sqrt{P} \rhor \e^{j\varphir} \left[ \cc(\taur) \right]_n \aariswdt^\trpose \oomegag \aabs^\trpose(\thetabr) \ff_g \sgn
     \\ 
     \frac{\partial \mgn}{\partial \rhor} &= \sqrt{P}  \e^{j\varphir} \left[ \cc(\taur) \right]_n \aarisw^\trpose \oomegag \aabs^\trpose(\thetabr) \ff_g \sgn
     \\ 
     \frac{\partial \mgn}{\partial \varphir} &= j \sqrt{P} \rhor \e^{j\varphir} \left[ \cc(\taur) \right]_n \aarisw^\trpose \oomegag \aabs^\trpose(\thetabr) \ff_g \sgn
\end{align*}
where $\aabsdt(\theta) \eqdef \partial \aabs(\theta) / \partial \theta$ and $\aariswdt(\theta) \eqdef \partial \aarisw(\theta) / \partial \theta$.

\subsection{FIM Entries as a Function of $\XX_g$, $\Psibbig_g $ and $\oomegag$}\label{sec_fim_func}

\subsubsection{FIM Submatrix for the LoS Path}\label{sec_fim_los}
The elements of $\JJbm$ in \eqref{eq_JJgamma_bdiag} can be obtained as follows:
\begin{align*}
    \Lambda(\taubm, \taubm) &= \frac{2 P \rhobm^2}{\sigma^2}  \sum_{g=1}^{G} \sum_{n=0}^{N-1}  \kappa_n^2 \, \sgnsq \aabs^\trpose(\thetabm)  \XX_g \aabs^\conj(\thetabm)  
    \\
    \Lambda(\thetabm, \taubm)& \\
    = \frac{2 P \rhobm^2}{\sigma^2} & \Re \left\{ j  \sum_{g=1}^{G} \sum_{n=0}^{N-1} \kappa_n \, \sgnsq \aabsdt^\trpose(\thetabm)  \XX_g \aabs^\conj(\thetabm)  \right\} 
    \\
    \Lambda(\rhobm, \taubm) &= 0 
    \\
     \Lambda(\varphibm, \taubm)& \\
     = -\frac{2 P \rhobm^2}{\sigma^2} & \sum_{g=1}^{G} \sum_{n=0}^{N-1} \kappa_n \, \sgnsq \aabs^\trpose(\thetabm)  \XX_g \aabs^\conj(\thetabm)   
     \\
     \Lambda(\thetabm, \thetabm) &= \frac{2 P \rhobm^2}{\sigma^2}   N \sum_{g=1}^{G} \aabsdt^\trpose(\thetabm)  \XX_g \aabsdt^\conj(\thetabm)  
      \\
     \Lambda(\rhobm, \thetabm) &= \frac{2 P \rhobm}{\sigma^2}   \Re \left\{ N  \sum_{g=1}^{G} \aabs^\trpose(\thetabm)  \XX_g \aabsdt^\conj(\thetabm)  \right\}
     \\
     \Lambda(\varphibm, \thetabm) &= \frac{2 P \rhobm^2}{\sigma^2}   \Re \left\{j N \sum_{g=1}^{G} \aabs^\trpose(\thetabm)  \XX_g \aabsdt^\conj(\thetabm)  \right\}
     \\
     \Lambda(\rhobm, \rhobm) &= \frac{2 P}{\sigma^2}  N  \sum_{g=1}^{G} \aabs^\trpose(\thetabm)  \XX_g \aabs^\conj(\thetabm)  
     \\
     \Lambda(\varphibm, \rhobm) &= 0
     \\
     \Lambda(\varphibm, \varphibm) &= \frac{2 P \rhobm^2}{\sigma^2} N \sum_{g=1}^{G}    \aabs^\trpose(\thetabm)  \XX_g \aabs^\conj(\thetabm) ~.
\end{align*}

\subsubsection{FIM Submatrix for the NLoS Path}\label{sec_fim_nlos}
The elements of $\JJr$ in \eqref{eq_JJgamma_bdiag} can be obtained as follows:
\begin{align*}
    \Lambda(\taurm, \taurm) & \\
    = \frac{2 P \rhor^2}{\sigma^2} & \sum_{g=1}^{G} \sum_{n=0}^{N-1} \kappa_n^2 \, \sgnsq  \aabs^\trpose(\thetabr)  \XX_g \aabs^\conj(\thetabr) \aarisw^\trpose \Psibbig_g \aarisw^\conj 
    \\
    \Lambda(\thetarm, \taurm) & \\
    = \frac{2 P \rhor^2}{\sigma^2} & \Re \left\{ j  \sum_{g=1}^{G} \sum_{n=0}^{N-1} \kappa_n \, \sgnsq \aabs^\trpose(\thetabr)  \XX_g \aabs^\conj(\thetabr) \aariswdt^\trpose \Psibbig_g \aarisw^\conj \right\} 
    \\
    \Lambda(\rhor, \taurm) &= 0 
    \\
     \Lambda(\varphir, \taurm)& \\
     = -\frac{2 P \rhor^2}{\sigma^2} & \sum_{g=1}^{G} \sum_{n=0}^{N-1} \kappa_n \, \sgnsq  \aabs^\trpose(\thetabr)  \XX_g \aabs^\conj(\thetabr) \aarisw^\trpose \Psibbig_g \aarisw^\conj 
     \\
     \Lambda(\thetarm, \thetarm) &= \frac{2 P \rhor^2}{\sigma^2} N \sum_{g=1}^{G} \aabs^\trpose(\thetabr)  \XX_g \aabs^\conj(\thetabr) \aariswdt^\trpose \Psibbig_g \aariswdt^\conj 
      \\
     \Lambda(\rhor, \thetarm) &=  \frac{2 P \rhor}{\sigma^2} \Re \left\{ N \sum_{g=1}^{G} \aabs^\trpose(\thetabr)  \XX_g \aabs^\conj(\thetabr) \aarisw^\trpose \Psibbig_g \aariswdt^\conj \right\}
     \\
     \Lambda(\varphir, \thetarm) &= \frac{2 P \rhor^2}{\sigma^2} \Re \left\{ j N \sum_{g=1}^{G} \aabs^\trpose(\thetabr)  \XX_g \aabs^\conj(\thetabr) \aarisw^\trpose \Psibbig_g \aariswdt^\conj \right\}
     \\
     \Lambda(\rhor, \rhor) &= \frac{2 P }{\sigma^2} N  \sum_{g=1}^{G} \aabs^\trpose(\thetabr)  \XX_g \aabs^\conj(\thetabr) \aarisw^\trpose \Psibbig_g \aarisw^\conj 
     \\
     \Lambda(\varphir, \rhor) &= 0
     \\ 
     \Lambda(\varphir, \varphir) &= \frac{2 P \rhor^2 }{\sigma^2} N \sum_{g=1}^{G} \aabs^\trpose(\thetabr)  \XX_g \aabs^\conj(\thetabr) \aarisw^\trpose \Psibbig_g \aarisw^\conj  ~.
\end{align*}

\subsubsection{FIM Submatrix for the LoS-NLoS Cross-Correlation}\label{sec_fim_corr}
The elements of $\JJcross$ in \eqref{eq_JJgamma_bdiag} can be computed as follows:
\begin{align*}
    \Lambda(\taurm, \taubm) &= \frac{2 P}{\sigma^2} \rhobm \rhor \, \realpbig{ \e^{j(\varphir-\varphibm)}  \sum_{g=1}^{G} \sum_{n=0}^{N-1} \kappa_n^2 \, \sgnsq
    \\ &~~~~ \times
    [\cc(\taur-\taubm) ]_n  \aabs^\trpose(\thetabr)  \XX_g \aabs^\conj(\thetabm) \aarisw^\trpose \oomegag}
    \\
    \Lambda(\thetarm, \taubm)&=  \frac{2 P}{\sigma^2} \rhobm \rhor \, \realpbig{ j \e^{j(\varphir-\varphibm)}  \sum_{g=1}^{G} \sum_{n=0}^{N-1} \kappa_n  \, \sgnsq
    \\ &~~~~ \times
    [\cc(\taur-\taubm) ]_n  \aabs^\trpose(\thetabr)  \XX_g \aabs^\conj(\thetabm) \aariswdt^\trpose \oomegag}
    \\
    \Lambda(\rhor, \taubm) &= \frac{2 P}{\sigma^2} \rhobm  \, \realpbig{ j \e^{j(\varphir-\varphibm)}  \sum_{g=1}^{G} \sum_{n=0}^{N-1} \kappa_n  \, \sgnsq
    \\ &~~~~ \times
    [\cc(\taur-\taubm) ]_n  \aabs^\trpose(\thetabr)  \XX_g \aabs^\conj(\thetabm) \aarisw^\trpose \oomegag}
    \\
    \Lambda(\varphir, \taubm) &= -\frac{2 P}{\sigma^2} \rhobm \rhor \, \realpbig{ \e^{j(\varphir-\varphibm)}  \sum_{g=1}^{G} \sum_{n=0}^{N-1} \kappa_n  \, \sgnsq
    \\ &~~~~ \times
    [\cc(\taur-\taubm) ]_n  \aabs^\trpose(\thetabr)  \XX_g \aabs^\conj(\thetabm) \aarisw^\trpose \oomegag}
    \\
    \Lambda(\taurm, \thetabm) &= -\frac{2 P}{\sigma^2} \rhobm \rhor \, \realpbig{ j \e^{j(\varphir-\varphibm)}  \sum_{g=1}^{G} \sum_{n=0}^{N-1} \kappa_n  \, \sgnsq
    \\ &~~~~ \times
    [\cc(\taur-\taubm) ]_n  \aabs^\trpose(\thetabr)  \XX_g \aabsdt^\conj(\thetabm) \aarisw^\trpose \oomegag}
    \\
    \Lambda(\thetarm, \thetabm)&= \frac{2 P}{\sigma^2} \rhobm \rhor \, \realpbig{ \e^{j(\varphir-\varphibm)}  \sum_{g=1}^{G} \sum_{n=0}^{N-1}  \sgnsq
    \\ &~~~~ \times
    [\cc(\taur-\taubm) ]_n  \aabs^\trpose(\thetabr)  \XX_g \aabsdt^\conj(\thetabm) \aariswdt^\trpose \oomegag}
    \\
    \Lambda(\rhor, \thetabm) &= \frac{2 P}{\sigma^2} \rhobm  \, \realpbig{ \e^{j(\varphir-\varphibm)}  \sum_{g=1}^{G} \sum_{n=0}^{N-1}  \sgnsq
    \\ &~~~~ \times
    [\cc(\taur-\taubm) ]_n  \aabs^\trpose(\thetabr)  \XX_g \aabsdt^\conj(\thetabm) \aarisw^\trpose \oomegag}
    \\
    \Lambda(\varphir, \thetabm) &= \frac{2 P}{\sigma^2} \rhobm \rhor  \, \realpbig{ j \e^{j(\varphir-\varphibm)}  \sum_{g=1}^{G} \sum_{n=0}^{N-1}  \sgnsq
    \\ &~~~~ \times
    [\cc(\taur-\taubm) ]_n  \aabs^\trpose(\thetabr)  \XX_g \aabsdt^\conj(\thetabm) \aarisw^\trpose \oomegag}
    \\
    \Lambda(\taurm, \rhobm) &= \frac{2 P}{\sigma^2}  \rhor  \, \realpbig{ -j \e^{j(\varphir-\varphibm)}  \sum_{g=1}^{G} \sum_{n=0}^{N-1} \kappa_n  \sgnsq
    \\ &~~~~ \times
    [\cc(\taur-\taubm) ]_n  \aabs^\trpose(\thetabr)  \XX_g \aabs^\conj(\thetabm) \aarisw^\trpose \oomegag}
    \\
    \Lambda(\thetarm, \rhobm)&=  \frac{2 P}{\sigma^2}  \rhor  \, \realpbig{  \e^{j(\varphir-\varphibm)}  \sum_{g=1}^{G} \sum_{n=0}^{N-1}   \sgnsq
    \\ &~~~~ \times
    [\cc(\taur-\taubm) ]_n  \aabs^\trpose(\thetabr)  \XX_g \aabs^\conj(\thetabm) \aariswdt^\trpose \oomegag}
    \\
    \Lambda(\rhor, \rhobm) &= \frac{2 P}{\sigma^2}  \realpbig{  \e^{j(\varphir-\varphibm)}  \sum_{g=1}^{G} \sum_{n=0}^{N-1}   \sgnsq
    \\ &~~~~ \times
    [\cc(\taur-\taubm) ]_n  \aabs^\trpose(\thetabr)  \XX_g \aabs^\conj(\thetabm) \aarisw^\trpose \oomegag}
    \\
    \Lambda(\varphir, \rhobm) &= \frac{2 P}{\sigma^2} \rhor \, \realpbig{ j \e^{j(\varphir-\varphibm)}  \sum_{g=1}^{G} \sum_{n=0}^{N-1}   \sgnsq
    \\ &~~~~ \times
    [\cc(\taur-\taubm) ]_n  \aabs^\trpose(\thetabr)  \XX_g \aabs^\conj(\thetabm) \aarisw^\trpose \oomegag}
    \\
    \Lambda(\taurm, \varphibm) &= -\frac{2 P}{\sigma^2} \rhobm \rhor  \, \realpbig{ \e^{j(\varphir-\varphibm)}  \sum_{g=1}^{G} \sum_{n=0}^{N-1} \kappa_n  \sgnsq
    \\ &~~~~ \times
    [\cc(\taur-\taubm) ]_n  \aabs^\trpose(\thetabr)  \XX_g \aabs^\conj(\thetabm) \aarisw^\trpose \oomegag}
    \\
    \Lambda(\thetarm, \varphibm)&=  -\frac{2 P}{\sigma^2} \rhobm \rhor  \, \realpbig{ j \e^{j(\varphir-\varphibm)}  \sum_{g=1}^{G} \sum_{n=0}^{N-1}   \sgnsq
    \\ &~~~~ \times
    [\cc(\taur-\taubm) ]_n  \aabs^\trpose(\thetabr)  \XX_g \aabs^\conj(\thetabm) \aariswdt^\trpose \oomegag}
    \\
    \Lambda(\rhor, \varphibm) &= -\frac{2 P}{\sigma^2} \rhobm   \, \realpbig{ j \e^{j(\varphir-\varphibm)}  \sum_{g=1}^{G} \sum_{n=0}^{N-1}   \sgnsq
    \\ &~~~~ \times
    [\cc(\taur-\taubm) ]_n  \aabs^\trpose(\thetabr)  \XX_g \aabs^\conj(\thetabm) \aarisw^\trpose \oomegag}
    \\
    \Lambda(\varphir, \varphibm) &= \frac{2 P}{\sigma^2} \rhobm \rhor   \, \realpbig{ \e^{j(\varphir-\varphibm)}  \sum_{g=1}^{G} \sum_{n=0}^{N-1}   \sgnsq
    \\ &~~~~ \times
    [\cc(\taur-\taubm) ]_n  \aabs^\trpose(\thetabr)  \XX_g \aabs^\conj(\thetabm) \aarisw^\trpose \oomegag} ~.
\end{align*}

\section{Entries of Transformation Matrix in \eqref{eq_Jeta}}\label{sec_transform_mat}

The transformation matrix in \eqref{eq_Jeta} can be written as
\begin{align} \nonumber
\bm{T} &= \frac{\partial \bm{\gamma}^{\mathsf{T}}}{\partial \bm{\eta}} =\begin{bmatrix}
\partial \tau_{{\BM}}/\partial p_x & \partial \theta_{{\BM}}/\partial p_x & \cdots & \partial \varphi_{\text{\tiny R}}/\partial p_x\\
\partial \tau_{{\BM}}/\partial p_y & \partial \theta_{{\BM}}/\partial p_y & \cdots & \partial \varphi_{\text{\tiny R}}/\partial p_y\\
\vdots & \vdots & \vdots \\
\partial \tau_{{\BM}}/\partial \Delta & \partial \theta_{{\BM}}/\partial \Delta & \cdots & \partial \varphi_{\text{\tiny R}}/\partial \Delta
\end{bmatrix}.
\end{align}

The entries of $\bm{T}$ are given by
\begin{align}
\frac{\partial \tau_\BM}{\partial p_x} = \frac{p_x}{c\|\bm{p}\|}, &   \qquad \frac{\partial \tau_\BM}{\partial p_y} = \frac{p_y}{c\|\bm{p}\|}, \nonumber \\
\frac{\partial \theta_\BM}{\partial p_x} = \frac{-p_y/p^2_x}{1+ (p_y/p_x)^2}, &\qquad \frac{\partial \theta_\BM}{\partial p_y} = \frac{1/p_x}{1+ (p_y/p_x)^2}, \nonumber \\
\frac{\partial \tau_\RM}{\partial p_x} = \frac{p_x - r_x}{c\|\bm{r}-\bm{p}\|}, & \qquad \frac{\partial \tau_\RM}{\partial p_y} =\frac{p_y - r_y}{c\|\bm{r}-\bm{p}\|}, \nonumber \\
\frac{\partial \theta_\RM}{\partial p_x} = \frac{-(p_y-r_y)/(p_x-r_x)^2}{1+\left(\frac{p_y - r_y}{p_x - r_x}\right)^2}, & \qquad \frac{\partial \theta_\RM}{\partial p_y} = \frac{1/(p_x-r_x)}{1+\left(\frac{p_y - r_y}{p_x - r_x}\right)^2}, \nonumber \\
\frac{\partial \tau_\BM}{\partial \Delta} = \frac{\partial \tau_\RM}{\partial \Delta} = 1, & \qquad \frac{\partial \rho_\BM}{\partial \rho_\BM} = \frac{\partial \varphi_\BM}{\partial \varphi_\BM} = \frac{\partial \rho_\R}{\partial \rho_\R} = \frac{\partial \varphi_\R}{\partial \varphi_\R} = 1, \nonumber
\end{align}
while the rest of the entries in $\bm{T}$ are zero.


\newpage
\newpage

{\ed 
\section{Robust Joint BS-RIS Beamforming Design Using CEB as Performance Metric}\label{sec::ceb_metric}
In this section, we derive the joint design of BS precoders and RIS phase profiles by using the clock error bound (CEB) as performance metric. We recall that the CEB is given by
\begin{align}
    \fceb = \big[ \Jeta^{-1} \big]_{7,7}~,
\end{align}
where $\Jeta$ is the location domain FIM in \eqref{eq_Jeta}. Then, the worst-case CEB minimizing version of the beam power allocation problem in \eqref{eq_peb_codebook} can be formulated as follows:
\begin{subequations} \label{eq_ceb_codebook}
\begin{align} \label{eq_ceb_codebook_obj}
	\mathop{\mathrm{min}}\limits_{ \substack{\varrhob,t \\\{u_{m}\}} } &~ t \\ \label{eq_ceb_codebook_lmi}
    \mathrm{s.t.}&~~ 
    \begin{bmatrix}  \JJeta(\{\XX_g, \oomegag, \Psibbig_g  \}_{g=1}^{G}; \etab(\pp_m)) & \ekk{7} \\ \ekkt{7} & u_m  \end{bmatrix} \succeq 0 ~, 
	\\ \nonumber &~~ u_m \leq t, ~m=0,\ldots,\ngrid-1 ~,
    \\ \nonumber &~~  \tracebig{\sum_{g=1}^{G}\XX_g} = 1 \, , \, \varrhob \succeq \boldzero \, , \, \XX_g = \varrho_g \FFbs_{:,i} (\FFbs_{:,i})^\hermit ~,
    \\ \nonumber &~~ \oomegag = \FFris_{:,j} \, , \, \Psibbig_g = \oomegag (\oomegag)^\hermit  \,, \, g=1,\ldots,G ~,
\end{align}
\end{subequations}
where $\ekk{7}$ is the $7$-th column of the identity matrix. 

To analyze if the CEB based design in \eqref{eq_ceb_codebook} has any differences compared to the PEB based design in \eqref{eq_peb_codebook}, we have obtained the optimal power allocation $\varrhob^{\star}$ across $G$ transmissions by solving the problem \eqref{eq_ceb_codebook} in Step~(c) of Algorithm~\ref{alg_codebook} instead of solving \eqref{eq_peb_codebook}. In Fig.~\ref{fig_peb_vs_ceb}, we plot the resulting PEB and CEB values corresponding to the optimal power allocation obtained by solving the CEB-based problem in \eqref{eq_ceb_codebook}, along with the PEB and CEB values corresponding to the solution of the PEB-based optimization in \eqref{eq_peb_codebook}. As seen from the figure, the CEB minimization yields practically the same results as the PEB minimization, which implies that positioning and synchronization are strongly coupled. Hence, by optimizing the PEB based criterion, we inherently take into account the error of synchronization.

\begin{figure}[h]
		\centering
		\includegraphics[width=0.35\textwidth]{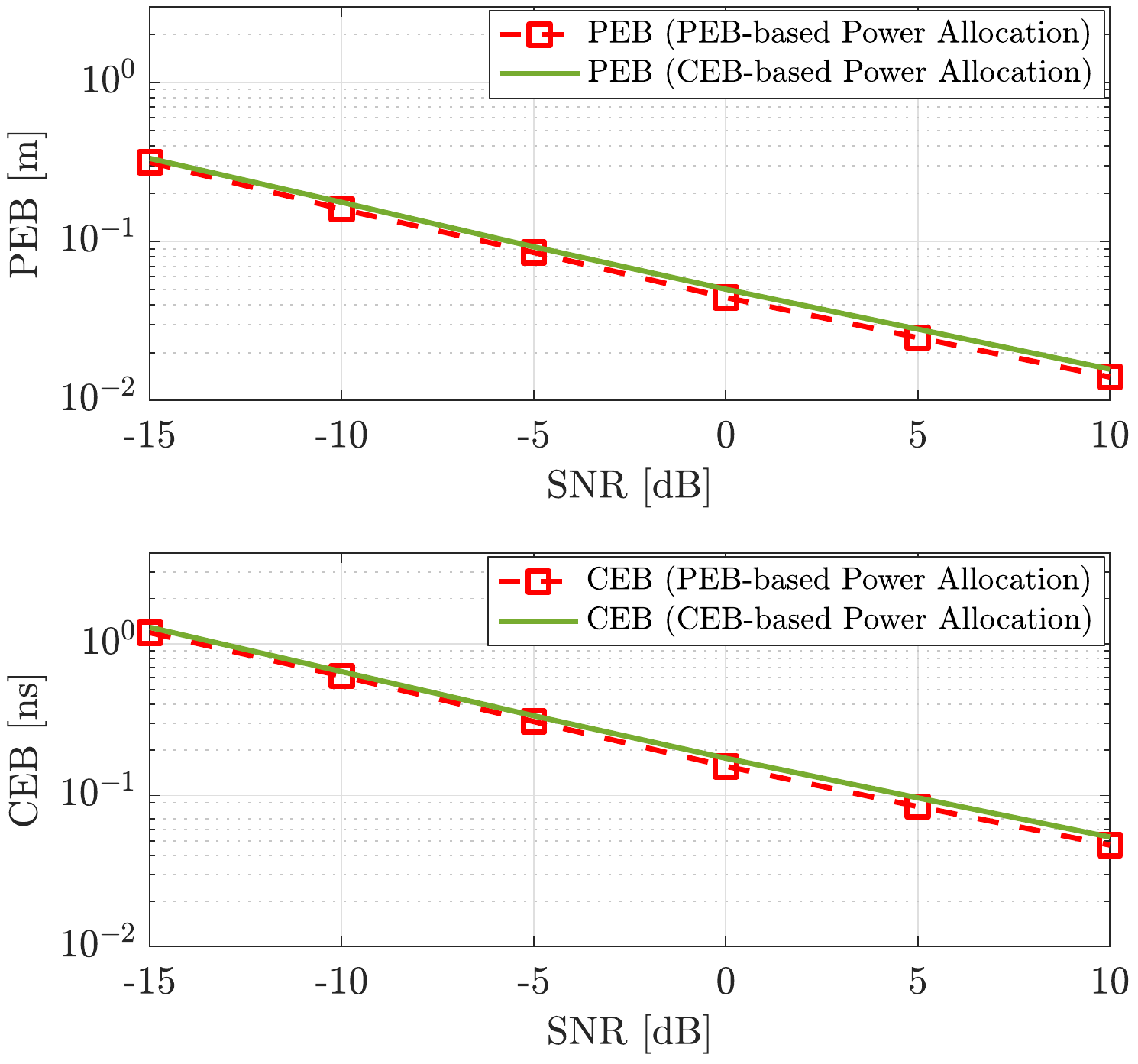}
		\caption{\ed Comparison of PEB-based power allocation in (31) and CEB-based power allocation in \eqref{eq_ceb_codebook} in terms of the resulting PEB and CEB values.}\label{fig_peb_vs_ceb}
	\end{figure}
}


{\ed 
\section{Obtaining Codebooks in \eqref{eq_prop_codebook} from Proposition~\ref{prop_BS_precoder} and Proposition~\ref{prop_RIS_profile}}\label{sec_supp_codebook}
In this section, we provide the methodology to arrive at the proposed codebooks in \eqref{eq_prop_codebook} based on the results in Prop.~\ref{prop_BS_precoder} and Prop.~\ref{prop_RIS_profile}. According to Prop.~\ref{prop_BS_precoder}, the optimal BS precoder covariance matrices $\XX_g = \ff_g \ff_g^\hermit$ in the absence of the rank-one constraints $\rank(\XX_g) = 1$ can be expressed as
\begin{align}\label{eq_xxg_r1}
    \XX_g = \aabsbig \Upsilonb_g \aabsbig^\hermit ~,
\end{align}
where $ \Upsilonb_g  \succeq 0$ and
\begin{align} \label{eq_a_bs_r1}
        \aabsbig = \left[ \aabs(\thetabr) ~ \aabs(\thetabm) ~  \aabsdt(\thetabm)  \right]^\conj \in \complexset{\nbs}{3} ~.
    \end{align}
Since $ \Upsilonb_g  \succeq 0$, it can be written as $ \Upsilonb_g  = \Lambdab_g \Lambdab_g^\hermit$ for some $\Lambdab_g \in \complexset{3}{3}$. Accordingly, \eqref{eq_xxg_r1} becomes
\begin{align}\label{eq_xxg_r11}
    \XX_g = \aabsbig \Lambdab_g (\aabsbig \Lambdab_g)^\hermit ~.
\end{align}
To make $\XX_g$ in \eqref{eq_xxg_r11} rank-one for every transmission $g$, we set $\Lambdab_g$ to be a diagonal matrix where only a single diagonal entry is non-zero, representing which beam (column) in $\aabsbig$ is selected. More rigorously, we let
\begin{align}\label{eq_lambda_r1}
    \Lambdab_g = \sqrt{\varrho_g} \, \diag{ \ekk{\chi(g)} } ~,
\end{align}
where $\ekk{k}$ denotes the $k$-th column of identity matrix, $\chi: \{1, \ldots, G\} \to \{1, 2, 3\} $ is a mapping from transmission indices to beam indices in $\aabsbig$, and $\varrho_g$ is the power allocated to the $g$-th transmission, as defined in \eqref{eq_peb_codebook}. Inserting \eqref{eq_lambda_r1} into \eqref{eq_xxg_r11} yields
\begin{align}\label{eq_xxg_r12}
    \XX_g = \varrho_g \, [ \aabsbig ]_{:, \chi(g)} [ \aabsbig ]_{:, \chi(g)}^\hermit ~,
\end{align}
where $[ \aabsbig ]_{:, k}$ represents the $k$-th column of $ \aabsbig$. 

It is obvious that $\XX_g$ in \eqref{eq_xxg_r12} is now rank-one, with $\varrho_g$ representing the adjustable power of the $g$-th transmission and $\chi(g)$ the adjustable beam index in $\aabsbig$. As a summary of the above procedure, to be able to satisfy the rank-one constraint for $\XX_g$ in \eqref{eq_xxg_r1}, we reduce the degrees of freedom in optimizing the matrix $\Upsilonb_g$, or, equivalently $\Lambdab_g$, by constraining $\Lambdab_g$ to have the structure in \eqref{eq_lambda_r1}, where the degrees of freedom are the \textit{selection of the beam index} $\chi(g)$ and the corresponding \textit{power level} $\varrho_g$. Therefore, the reason why the beamforming matrix $ \aabsbig$ becomes a single beamforming vector in the codebooks of \eqref{eq_prop_codebook} is due to the \textit{rank-one constraint} on $\XX_g$. We note that by sweeping through different beam indices $\chi(g)$'s for different transmissions, the BS transmit precoding can effectively span the subspace determined by $\aabsbig$ in \eqref{eq_a_bs_r1}, which is exactly what the codebooks in \eqref{eq_prop_codebook} do. Additionally, we remark that the same arguments above can be used to construct rank-one matrices $\Psibbig_g = \oomegag (\oomegag)^\hermit$ for the RIS phase profiles from $\aariswbig$ in \eqref{eq_bris}.

The above methodology to obtain rank-one solutions is developed for the case of perfect knowledge of UE location. To extend it to the case of uncertainty in UE location in Sec.~\ref{sec_robust_design}, we propose to span the corresponding uncertainty region of UE by adding more beamforming vectors to $\aabsbig$ and $\aariswbig$ such that the corresponding AoDs $\thetabm$ and $\thetarm$ in the arguments of the respective BS and RIS steering vectors cover the uncertainty interval in AoD. It is important to note that $\thetabr$ is a known constant (as it represents the AoD from the BS to RIS) and thus $\aabs(\thetabr)$ in \eqref{eq_a_bs_r1} remains to be a single beam in the codebooks of \eqref{eq_prop_codebook}. As seen from \eqref{eq_prop_codebook}, \eqref{eq_fris_sum} and \eqref{eq_fris_der}, in the proposed codebooks, the uncertainty region is covered by the uniformly spaced AoDs $\{ \thetabm^{(i)} \}_{i=1}^{\lb}$ and $\{ \thetarm^{(i)} \}_{i=1}^{\lr}$ from the BS to the UE and from the RIS to the UE, respectively.

}


{\ed 
\section{Equivalence of Orthogonal and Non-orthogonal Beams in Proposition~\ref{prop_BS_precoder}}\label{sec_supp_orthogonal}
In this section, we prove the equivalence of orthogonal and non-orthogonal beams in Prop.~\ref{prop_BS_precoder}, obtained as the solutions to \eqref{eq_problem_peb_relaxed_fixed_RIS} by showing that the orthogonality constraint is not needed to achieve the optimal solution. Let us consider the following two equivalent solutions of \eqref{eq_problem_peb_relaxed_fixed_RIS}:
\begin{align}\label{eq_xxg_r1_nonorth}
    \XX_g &= \aabsbig \Upsilonb_g \aabsbig^\hermit ~, \\ \label{eq_xxg_r1_orth}
    \XX_g &= \aabsbigtilde \Upsilonborth_g \aabsbigtilde^\hermit ~,
\end{align}
where
\begin{align}
    \aabsbigtilde &\eqdef \aabsbig ( \aabsbig^\hermit \aabsbig )^{-1/2} ~, \\ \label{eq_relation_orth}
    \Upsilonborth_g &\eqdef ( \aabsbig^\hermit \aabsbig )^{1/2} \Upsilonb_g ( \aabsbig^\hermit \aabsbig )^{1/2} ~.
\end{align}
Here, \eqref{eq_xxg_r1_nonorth} and \eqref{eq_xxg_r1_orth} represent equivalent solutions where the beams are non-orthogonal and orthogonal, respectively (notice that $\aabsbigtilde^\hermit \aabsbigtilde = \Imatrix$). We note that there is a one-to-one relation between the solutions $\Upsilonb_g$ and $\Upsilonborth_g$; one can always obtain one from the other by using \eqref{eq_relation_orth}. Hence, the beams in the proposed codebook \eqref{eq_prop_codebook}, derived from Prop.~\ref{prop_BS_precoder} and Prop.~\ref{prop_RIS_profile}, do not have to be orthogonal. For each non-orthogonal solution, there exists an equivalent orthogonal solution that leads to the same precoder covariance $\XX_g$.
}


{\ed 
\section{Setting of number of grid points $M$ in \eqref{eq_peb_codebook}}\label{sec::choiceM}
In this section, we investigate the impact of the number of grid points $M$ used in solving \eqref{eq_peb_codebook} in terms of both estimation performance and computational complexity. Being $M$ the number of discrete UE positions $\left\{\mathbf{p}_m\right\}_{m=1}^M$ on which the worst-case PEB is evaluated to solve eq. (31) using the proposed codebook-based approach, it is reasonable to expect that the larger the value of $M$ (i.e., the finer the grid), the more accurate the design of the joint active BS and passive RIS beamforming vectors, and in turn the more accurate the ultimate localization and synchronization performance. At the same time, it is quite intuitive to observe that the more UE positions $M$ need to be tested for the resolution of eq. (31), the increased will be the complexity involved in the design procedure. In order to clearly assess this trade-off under the considered parameter setup and motivate the choice of the adopted $M$, we report an additional analysis aimed at evaluating the RMSEs on the estimation of the UE position $\mathbf{p}$ when different values of $M$ are used to design the joint BS and RIS beamforming vectors, recording at the same time the total execution time required to solve the power optimization problem in eq. (31). For the sake of the analysis, we run 1000 independent Monte Carlo trials. The results reported in Fig.~\ref{fig:Discretization} reveal that values of $M$ lower than 9 produce evident performance losses being the corresponding values of RMSE on both position and clock-offset estimation higher, while for values of $M \geq 9$ there are no appreciable improvements in the ultimate accuracy. Considering that grids of $M = 16$ or $M = 25$ points would lead to doubling or even tripling the average total execution time compared to the case of $M = 9$, it is apparent that the latter represents the most convenient choice to achieve the best trade-off between accuracy and computational cost.

\begin{figure}[h]
		\centering
		\includegraphics[width=0.4\textwidth]{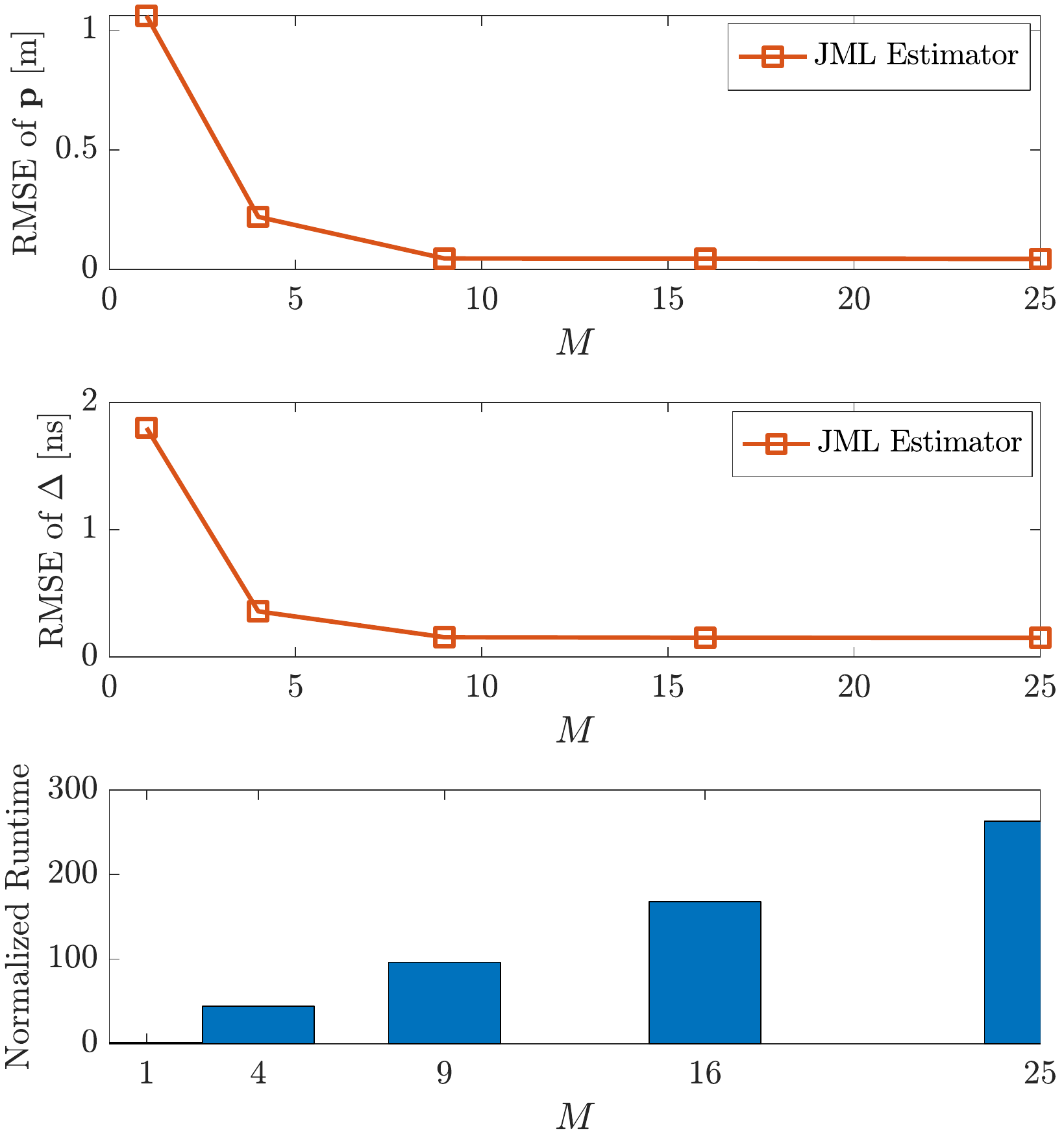}
		\caption{\ed Performance in terms of RMSEs and normalized execution time as a function of the number of discretization points $M$.}\label{fig:Discretization}
	\end{figure}
}

\section{Additional Simulation Results for Increased UE Uncertainty}\label{sec_5m_uncertainty}
In this section, we carry out an additional simulation analysis to investigate the performance of the proposed approach when considering an uncertainty region $\mathcal{P}$ for the UE location whose extent is increased to 5 m along each direction. For this setup, it turns out that $\lb = 20$ and $\lr = 18$, resulting in a total of $G = 378$ transmitted beams, while the rest of the simulation parameters remain unchanged.

\begin{figure}
 \includegraphics[width=0.5\textwidth]{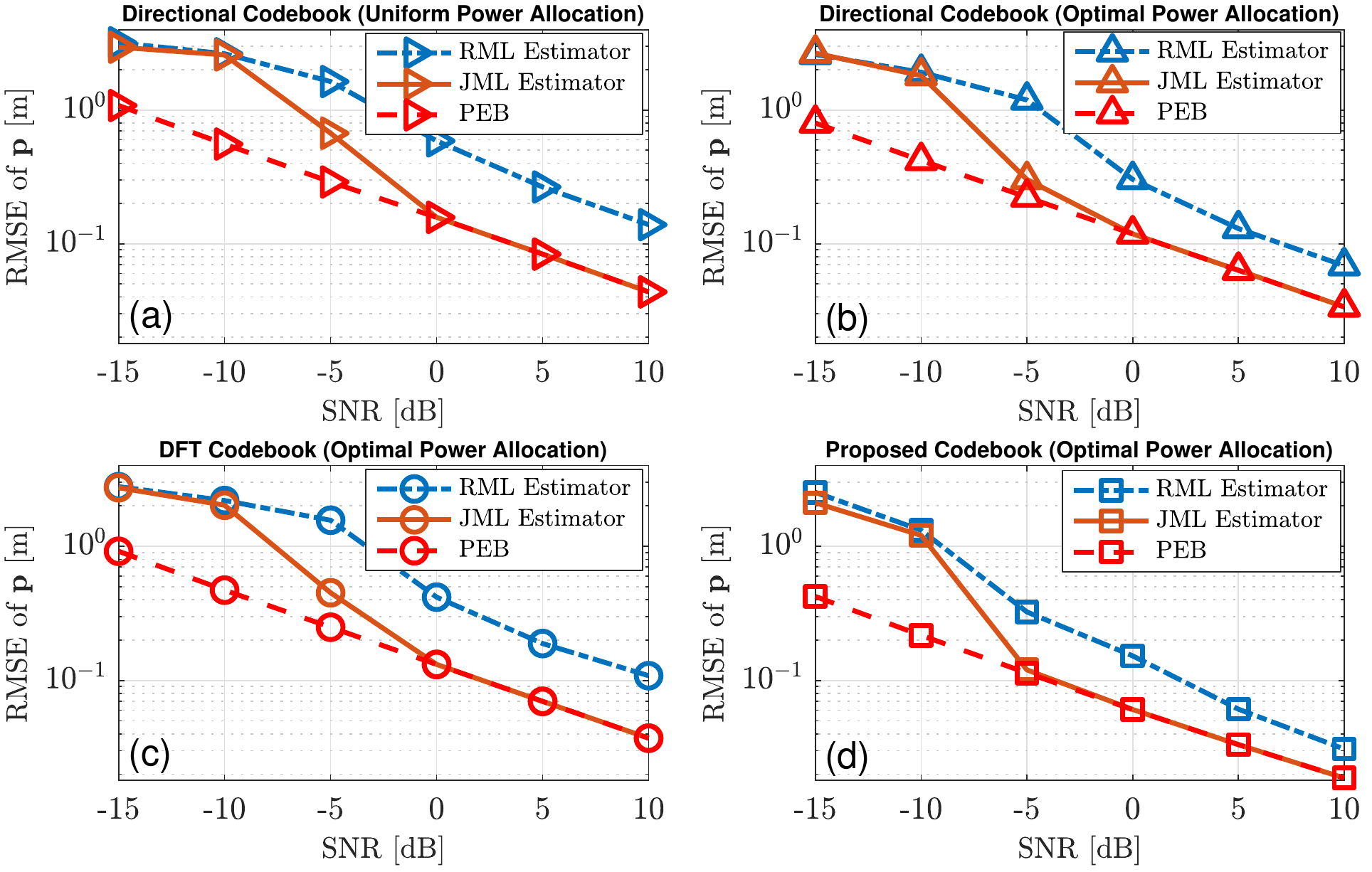}
 	\caption{RMSEs on the estimation of $\bm{p}$ as a function of the SNR for the directional codebook, DFT codebook, and proposed codebook.}
\label{fig:RMSE_pos_5m}
 \end{figure}

\begin{figure}
 \includegraphics[width=0.5\textwidth]{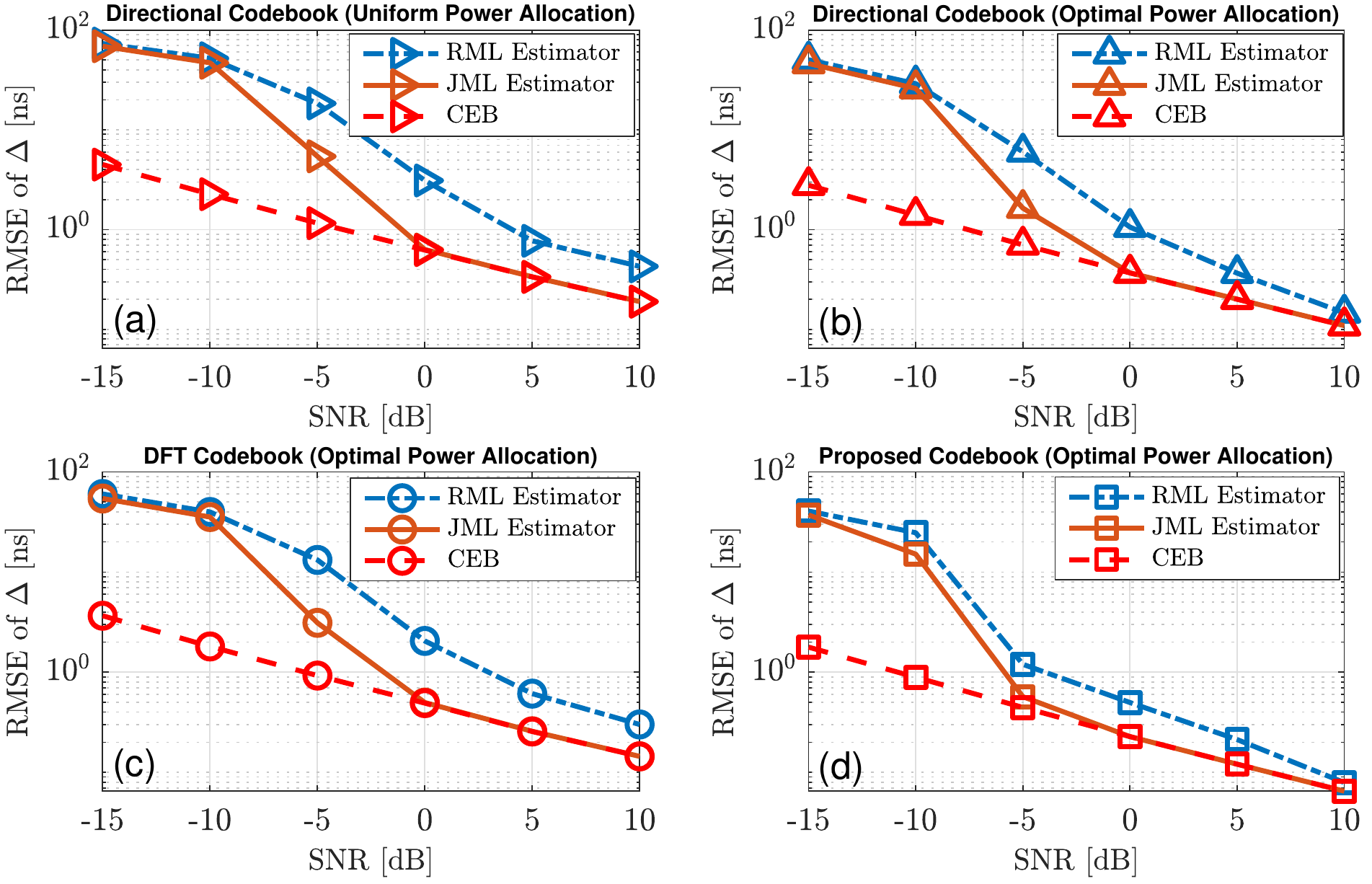}
 	\caption{RMSE on the estimation of $\Delta$ as a function of the SNR for the directional codebook, DFT codebook, and proposed codebook.}
\label{fig:RMSE_clock_5m}
 \end{figure}
 
In Figs.~\ref{fig:RMSE_pos_5m} and \ref{fig:RMSE_clock_5m}, we report the RMSE on the estimation of $\bm{p}$ and $\Delta$ as a function of the SNR, for all the considered precoding schemes, also in comparison with their corresponding CRLBs. As it would be expected, the values assumed by the PEBs and CEBs are higher than those obtained in case of 3 m uncertainty. Remarkably, the proposed low-complexity estimation algorithm is still able to provide very good performance in spite of the more challenging scenario at hand, attaining the theoretical bounds already at  $\text{SNR} = -5$ dB (the same as in case of 3 m uncertainty) when using the proposed BS-RIS precoding scheme, and at about $\text{SNR} = 0$ dB for the other precoding schemes.
 \begin{figure}%
    \centering
    \subfloat[\centering RMSE on the estimation of $\bm{p}$.]{{\includegraphics[width=0.38\textwidth]{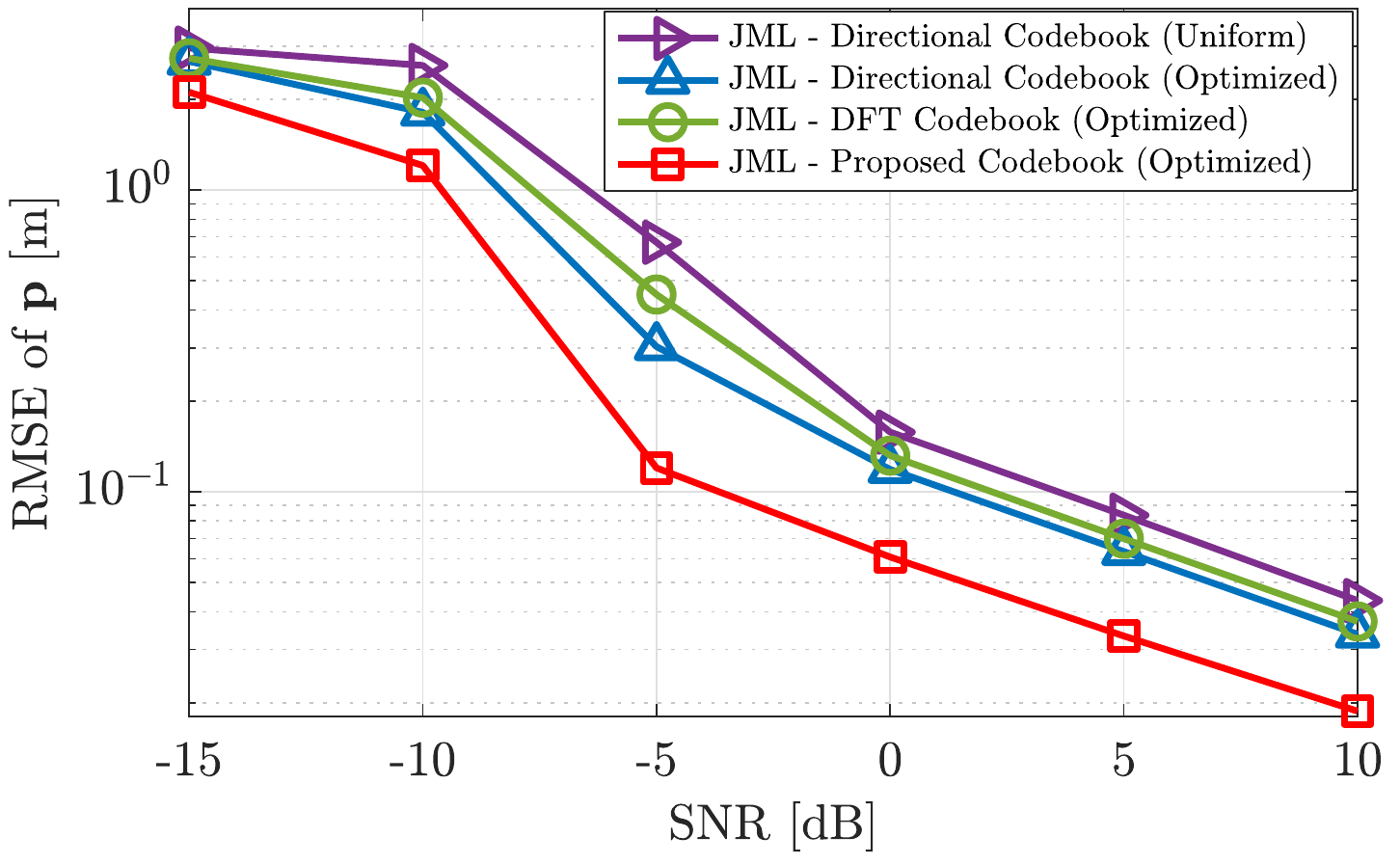} } \label{fig:comparison_pos_5m}}%
    \qquad
    \subfloat[\centering RMSE on the estimation of $\Delta$.]{{\includegraphics[width=0.38\textwidth]{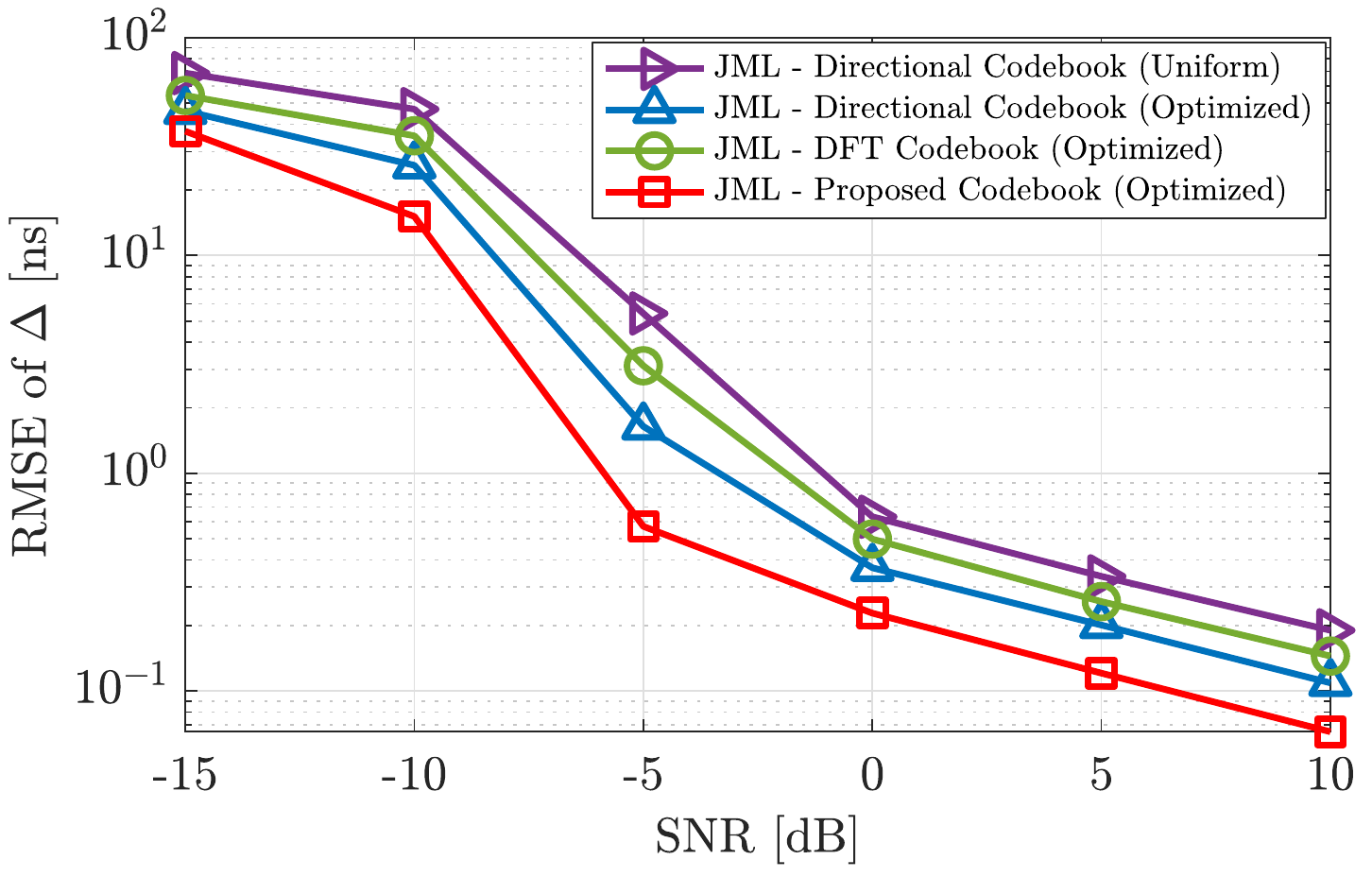} }\label{fig:comparison_clock_5m}}%
    \caption{Comparison between the RMSEs of (a) $\bm{p}$ and (b) $\Delta$ using the proposed JML estimator for different precoding schemes, as a function of the SNR.}%
    \label{fig:comparison_RMSE_5m}%
\end{figure}

In Fig.~\ref{fig:comparison_RMSE_5m}, we report a direct comparison among the RMSEs on the estimation of $\bm{p}$ and $\Delta$ for the proposed JML estimator fed  with different precoding  schemes. As it  can be noticed,  the  proposed  robust  joint  BS-RIS  precoding  scheme provides better localization and synchronization performance compared to the directional and DFT codebooks also in this case.

\begin{figure}%
    \centering
    \subfloat[\centering RMSE on the estimation of $\bm{p}$.]{{\includegraphics[width=0.38\textwidth]{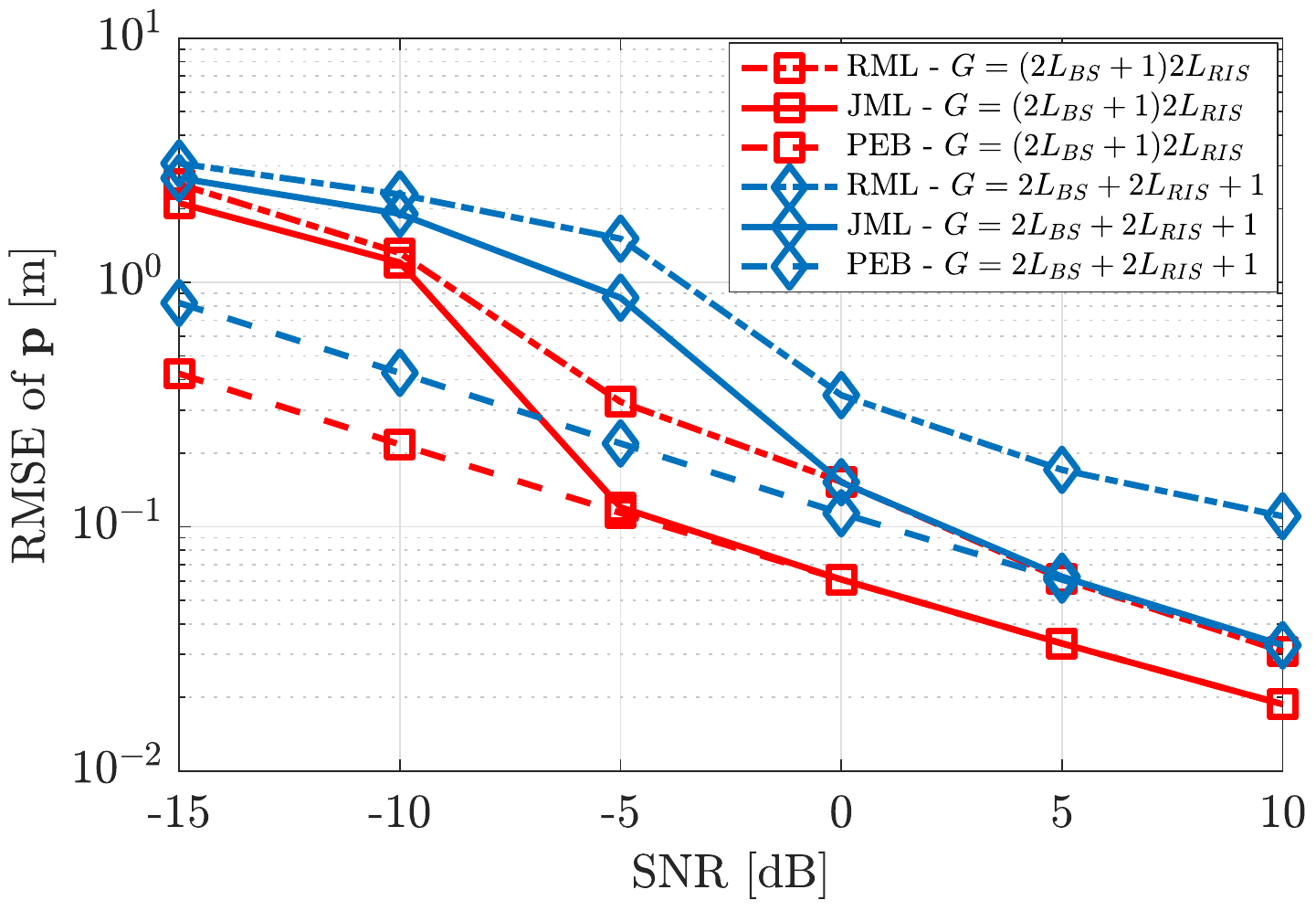} }\label{fig:reducedG_pos_5m}}%
    \qquad
    \subfloat[\centering RMSE on the estimation of $\Delta$.]{{\includegraphics[width=0.38\textwidth]{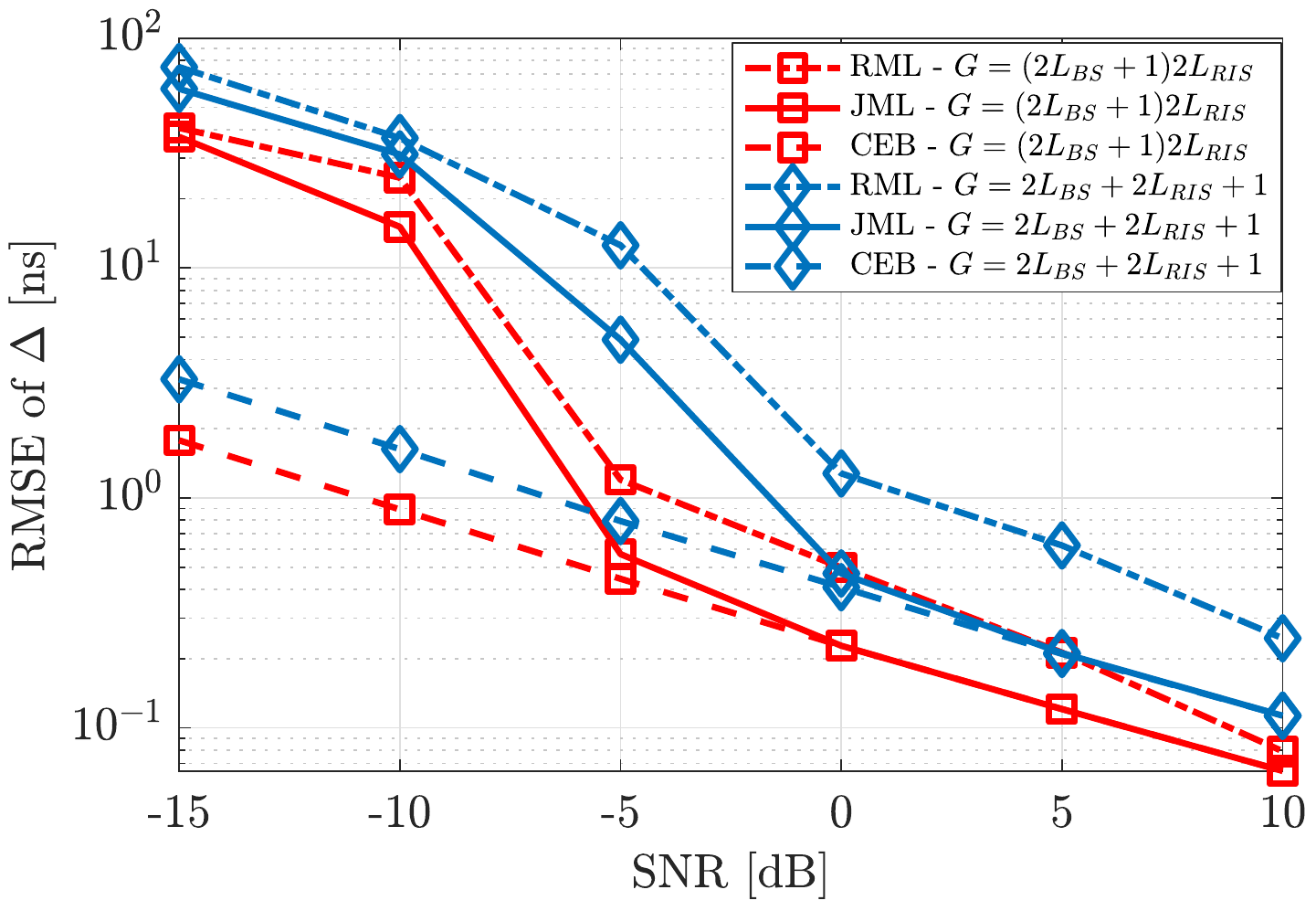} }\label{fig:reducedG_clock_5m}}%
    \caption{Performance comparison between the case of full number of transmitted beams $G = (2 \lb + 1) 2 \lr$ and proposed heuristic using a reduced $G = 2\lb + 2\lr + 1$.}%
    \label{fig:reducedG_5m}%
\end{figure}

To conclude the analysis, we investigate the performance of the proposed approach when using the ad-hoc heuristic that reduces the total number of transmitted beams from $G = (2\lb +1)\lr$ to $G = 2\lb + 2\lr +1$ as discussed in Sec.~VII-C3 of the main document.
Fig.~\ref{fig:reducedG_5m} shows the RMSEs on the estimation of $\bm{p}$ and $\Delta$ and the related lower bounds as a function of the SNR, for both cases of full and reduced number of transmissions $G$. Compared to the results in Sec.~VII-C3, the PEBs and CEBs in Fig.~\ref{fig:reducedG_pos_5m} and Fig.~\ref{fig:reducedG_clock_5m} exhibit a much  evident gap in the two considered settings, revealing that the increased uncertainty has a more significant impact onto the achievable estimation performance when the different configurations of the RIS phase profiles are not considered. This can be further confirmed by comparing the RMSEs of both RML and JML estimators reported in Figs.~\ref{fig:reducedG_pos_5m} and \ref{fig:reducedG_clock_5m}: notably, the estimation performances in case of reduced $G$ are visibly worse than those obtained when transmitting the full number of beams, especially for low values of the SNR.


\allowdisplaybreaks


\ifCLASSOPTIONcaptionsoff
  \newpage
\fi

\bibliographystyle{IEEEtran}
\bibliography{IEEEabrv,MISO_RIS}


\end{document}

%% file: commands.tex
\newcommand{\eqdef}{\stackrel{\textrm{\tiny def}}{=}}
\newcommand{\e}{e}
\newcommand{\ed}{\color{black}} 
\newcommand{\BM}{\text{\tiny B,U}}
\newcommand{\crosst}{\text{cross}}
\newcommand{\BR}{\text{\tiny B,R}}
\newcommand{\RM}{\text{\tiny R,U}}
\newcommand{\R}{\text{\tiny R}}
\newcommand{\BS}{\text{\tiny BS}}
\newcommand{\DFT}{\text{\tiny DFT}}
\newcommand{\RIS}{\text{\tiny RIS}}

\newcommand{\blkdiag}{{\rm{bd}}}

\newcommand{\ML}{\text{\tiny ML}}
\newcommand{\RML}{\text{\tiny RML}}
\newcommand{\FFT}{\text{\tiny FFT}}

\newcommand{\AAb}{\bm{A}}
\newcommand{\BB}{\bm{B}}

\newcommand{\FF}{\bm{F}}
\newcommand{\GG}{\bm{G}}
\newcommand{\ff}{\bm{f}}
\newcommand{\pp}{\bm{p}}
\newcommand{\pphat}{\widehat{\pp}}

\newcommand{\XX}{\bm{X}}

\newcommand{\sgn}{s_g[n]}
\newcommand{\sgnsq}{\abs{\sgn}^2}
\newcommand{\sgngen}[1]{s_g[#1]}
\newcommand{\ssb}{\bm{s}}
\newcommand{\Omegag}{\bm{\Omega}^g}
\newcommand{\oomegag}{\bm{\omega}^g}

\newcommand{\Gammabbig}{\bm{\Gamma}_g}
\newcommand{\Psibbig}{\bm{\Psi}}

\newcommand{\varrhob}{\bm{\varrho}}
\newcommand{\Lambdab}{\bm{\Lambda}}
\newcommand{\Upsilonb}{\bm{\Upsilon}}
\newcommand{\Xib}{\bm{\Xi}}

\newcommand{\ygn}{y^g[n]}

\DeclareMathOperator*{\E}{\mathbb{E}}

\newcommand{\thn}[1]{ {#1^{\rm{th} } } }

\newcommand{\ekk}[1]{\mathbf{e}_{#1}}
\newcommand{\ekkt}[1]{\mathbf{e}^T_{#1}}

\newcommand{\XXgtilde}{\widetilde{\XX}_g}
\newcommand{\XXgbar}{\widebar{\XX}_g}

\newcommand{\XXgstar}{\XX_g^\star}
\newcommand{\XXgtildestar}{\widetilde{\XX}_g^\star}
\newcommand{\XXgbarstar}{\widebar{\XX}_g^\star}

\newcommand{\XXgstarr}{\XX_g^{\star \star}}
\newcommand{\XXgtildestarr}{\widetilde{\XX}_g^{\star \star}}
\newcommand{\XXgbarstarr}{\widebar{\XX}_g^{\star \star}}

\newcommand{\mgn}{m^g[n]}
\newcommand{\mgnbm}{m_{\BM}^g[n]}

\newcommand{\mgnr}{m_{\R}^g[n]}

\newcommand{\rhobm}{\rho_{\BM}}
\newcommand{\varphibm}{\varphi_{\BM}}
\newcommand{\taubm}{\tau_{\BM}}
\newcommand{\thetabm}{\theta_{\BM}}

\newcommand{\rhor}{\rho_{\R}}
\newcommand{\varphir}{\varphi_{\R}}
\newcommand{\taurm}{\tau_{\RM}}
\newcommand{\thetarm}{\theta_{\RM}}

\newcommand{\thetarmtilde}{\widetilde{\theta}_{\RM}}
\newcommand{\thetabmtilde}{\widetilde{\theta}_{\BM}}

\newcommand{\taubr}{\tau_{\BR}}
\newcommand{\taur}{\tau_{\R}}
\newcommand{\thetabr}{\theta_{\BR}}
\newcommand{\thetabmc}{\theta_{\BM}^c}
\newcommand{\thetarmc}{\theta_{\RM}^c}
\newcommand{\phibr}{\phi_{\BR}}

\newcommand{\aabs}{\bm{a}_{\BS}}
\newcommand{\aabsbig}{\bm{A}_{\BS}}
\newcommand{\aaris}{\bm{a}_{\RIS}}
\newcommand{\aarisw}{\bm{b}_{\RIS}}
\newcommand{\aariswbig}{\bm{B}_{\RIS}}
\newcommand{\aariswdt}{\dt{\bm{b}}_{\RIS}}
\newcommand{\aariswdttilde}{\widetilde{\dt{\bm{b}}}_{\RIS}}

\newcommand{\hermit}{\mathsf{H}}
\newcommand{\trpose}{\mathsf{T}}
\newcommand{\frob}{\mathsf{F}}
\newcommand{\conj}{*}
\newcommand{\nr}{N_{\RIS}}
\newcommand{\lr}{L_{\RIS}}
\newcommand{\lb}{L_{\BS}}
\newcommand{\gammab}{\bm{\gamma}}
\newcommand{\gammabhat}{\widehat{\gammab}}

\newcommand{\ppreg}{{ \mathcal{P} }}

\newcommand{\boldone}{{ {\bm{1}} }}
\newcommand{\boldzero}{{ {\bm{0}} }}
\newcommand{\Imatrix}{{ \bm{\mathrm{I}} }}

\newcommand{\nbs}{N_{\BS}}
\newcommand{\cc}{\bm{c}}

\newcommand{\complexset}[2]{ \mathbb{C}^{#1 \times #2}  }

\newcommand{\realset}[2]{ \mathbb{R}^{#1 \times #2}  }

\newcommand{\diag}[1]{{\mathrm{diag}}\left(#1\right)}
\newcommand{\rank}{ \mathrm{rank}  }

\newcommand{\aabsbigtilde}{\widetilde{\bm{A}}_{\BS}}
\newcommand{\Upsilonborth}{\widetilde{\Upsilonb}}

\newcommand{\rev}[1]{\textcolor{black}{#1}} 

\newcommand{\realpbig}[1]{ \Re \Bigg\{#1\Bigg\}  }

\newcommand{\FFbs}{\FF^{\BS}}
\newcommand{\FFbssum}{\widebar{\FF}^{\BS}}
\newcommand{\FFbsdiff}{\dt{\FF}^{\BS}}
\newcommand{\FFris}{\FF^{\RIS}}
\newcommand{\FFrissum}{\widebar{\FF}^{\RIS}}
\newcommand{\FFrisdiff}{\dt{\FF}^{\RIS}}

\newcommand{\FFbsdft}{\FF^{\BS,\DFT}}
\newcommand{\FFrisdft}{\FF^{\RIS,\DFT}}

\newcommand{\ellbm}{\ell_{\BM}}
\newcommand{\ellrm}{\ell_{\RM}}

\newcommand{\ngrid}{M}

\newcommand{\FFbssumdouble}{\widetilde{\FF}^{\BS}}

\newcommand{\etab}{{ \bm{\eta} }}
\newcommand{\Jeta}{\bm{J}_{\etab}}
\newcommand{\fpeb}{{\rm{PEB}}}
\newcommand{\fceb}{{\rm{CEB}}}
\newcommand{\fpebblkdiag}{\fpeb^{\rm{\blkdiag}}}

\newcommand{\JJgamma}{\bm{J}_{\bm{\gamma}} }
\newcommand{\JJgammablkdiag}{\bm{J}^{\blkdiag}_{\bm{\gamma}} }
\newcommand{\JJeta}{\bm{J}_{\bm{\eta}} }
\newcommand{\JJbm}{\bm{J}_{\BM} }
\newcommand{\JJcross}{\bm{J}_{\crosst} }
\newcommand{\JJr}{\bm{J}_{\R} }

\newcommand{\tracebig}[1]{ {{{\rm{tr}}\Big( #1 \Big)}}  }
\newcommand{\tracesmall}[1]{ {{{\rm{tr}}\left( #1 \right)}}  }
\newcommand{\tracenormal}[1]{ {{{\rm{tr}}( #1 )}}  }

\usepackage{accents}
\newcommand*{\dt}[1]{%
	\accentset{\mbox{\large .}}{#1}}
	
\newcommand{\aabsdt}{\dt{\bm{a}}_{\BS} }

\newcommand{\ccdt}{\dt{\bm{c}} }

\providecommand{\abs}[1]{\lvert#1\rvert}

\newcommand{\norm}[1]{\left\lVert#1\right\rVert}

\newcommand{\projrange}[1]{\boldsymbol{\Pi}_{#1}}
\newcommand{\projnull}[1]{\boldsymbol{\Pi}^{\perp}_{#1}}

\usepackage{amsthm}

\theoremstyle{plain}

\newtheoremstyle{iremark}
  {\topsep}   
  {\topsep}   
  {\upshape}  
  {0.2in}       
  {\itshape}  
  {.}         
  {5pt plus 1pt minus 1pt} 
  {\thmname{#1}\thmnumber{ \itshape#2}\thmnote{ (#3)}} 

\usepackage{amsthm}

\newtheorem{theorem}{Theorem}
\newtheorem{lemma}[theorem]{Lemma}
\newtheorem{remark}{Remark}

\newtheorem{proposition}{Proposition}

\theoremstyle{definition}
\newtheorem*{proof}{Proof}

\makeatletter
\newcommand*\rel@kern[1]{\kern#1\dimexpr\macc@kerna}
\newcommand*\widebar[1]{%
  \begingroup
  \def\mathaccent##1##2{%
    \rel@kern{0.8}%
    \overline{\rel@kern{-0.8}\macc@nucleus\rel@kern{0.2}}%
    \rel@kern{-0.2}%
  }%
  \macc@depth\@ne
  \let\math@bgroup\@empty \let\math@egroup\macc@set@skewchar
  \mathsurround\z@ \frozen@everymath{\mathgroup\macc@group\relax}%
  \macc@set@skewchar\relax
  \let\mathaccentV\macc@nested@a
  \macc@nested@a\relax111{#1}%
  \endgroup
}
\makeatother